\documentclass{article}

\usepackage{graphicx, color}
\usepackage{amsmath} 
\usepackage{amssymb}
\usepackage{amsthm}
\usepackage{amsxtra}
\usepackage{amscd}
\usepackage{bbm}
\usepackage{mathrsfs}
\usepackage{bm}
\usepackage[british]{babel}

\renewcommand{\b}[1]{\bm{\mathrm{#1}}} 

\newcommand{\rr}{\mathrm} 
\renewcommand{\cal}{\mathcal} 
\newcommand{\scr}{\mathscr} 
\newcommand{\fra}{\mathfrak} 
\newcommand{\ul}[1]{\underline{#1} \!\,} 
\newcommand{\ol}[1]{\overline{#1} \!\,} 
\newcommand{\wh}{\widehat}
\newcommand{\wt}{\widetilde}

\newcommand{\ee}{\mathrm{e}}
\newcommand{\ii}{\mathrm{i}}
\newcommand{\dd}{\mathrm{d}}
\newcommand{\col}{\mathrel{\mathop:}}
\newcommand{\st}{\,\col\,}
\newcommand{\deq}{\mathrel{\mathop:}=}
\newcommand{\eqd}{=\mathrel{\mathop:}}
\newcommand{\umat}{\mathbbmss{1}} 
\renewcommand{\leq}{\leqslant}
\renewcommand{\geq}{\geqslant}

\newcommand{\ind}[1]{\b 1 (#1)}
\newcommand{\indb}[1]{\b 1 \pb{#1}}
\newcommand{\indB}[1]{\b 1 \pB{#1}}

\renewcommand{\epsilon}{\varepsilon}

\renewcommand{\P}{\mathbb{P}}
\newcommand{\E}{\mathbb{E}}
\newcommand{\R}{\mathbb{R}}
\newcommand{\C}{\mathbb{C}}
\newcommand{\N}{\mathbb{N}}
\newcommand{\Z}{\mathbb{Z}}

\newcommand{\p}[1]{({#1})}
\newcommand{\pb}[1]{\bigl({#1}\bigr)}
\newcommand{\pB}[1]{\Bigl({#1}\Bigr)}
\newcommand{\pbb}[1]{\biggl({#1}\biggr)}
\newcommand{\pBB}[1]{\Biggl({#1}\Biggr)}

\newcommand{\qB}[1]{\Bigl[{#1}\Bigr]}
\newcommand{\qbb}[1]{\biggl[{#1}\biggr]}
\newcommand{\qBB}[1]{\Biggl[{#1}\Biggr]}

\newcommand{\h}[1]{\{{#1}\}}
\newcommand{\hb}[1]{\bigl\{{#1}\bigr\}}
\newcommand{\hB}[1]{\Bigl\{{#1}\Bigr\}}
\newcommand{\hbb}[1]{\biggl\{{#1}\biggr\}}

\newcommand{\abs}[1]{\lvert #1 \rvert}
\newcommand{\absb}[1]{\bigl\lvert #1 \bigr\rvert}
\newcommand{\absB}[1]{\Bigl\lvert #1 \Bigr\rvert}
\newcommand{\absbb}[1]{\biggl\lvert #1 \biggr\rvert}

\newcommand{\norm}[1]{\lVert #1 \rVert}
\newcommand{\normb}[1]{\bigl\lVert #1 \bigr\rVert}

\newcommand{\avg}[1]{\langle #1 \rangle}

\newcommand{\scalar}[2]{\langle{#1} \mspace{2mu}, {#2}\rangle}
\newcommand{\scalarb}[2]{\bigl\langle{#1} \mspace{2mu}, {#2}\bigr\rangle}
\newcommand{\scalarB}[2]{\Bigl\langle{#1} \,\mspace{2mu},\, {#2}\Bigr\rangle}

\newcommand{\comb}[2]{\bigl[{#1} \mspace{2mu}, {#2}\bigr]}

\DeclareMathOperator{\tr}{Tr}

\DeclareMathOperator{\im}{Im}

\theoremstyle{plain} 
\newtheorem{theorem}{Theorem}[section]
\newtheorem*{theorem*}{Theorem}
\newtheorem{lemma}[theorem]{Lemma}
\newtheorem*{lemma*}{Lemma}

\newtheorem*{corollary*}{Corollary}
\newtheorem{proposition}[theorem]{Proposition}
\newtheorem*{proposition*}{Proposition}
\newtheorem{definition}[theorem]{Definition}
\newtheorem*{definition*}{Definition}

\newtheorem*{example*}{Example}
\newtheorem{remark}[theorem]{Remark}

\newtheorem*{remark*}{Remark}
\newtheorem*{remarks*}{Remarks}

\usepackage{cite}

\usepackage{graphicx}
\usepackage[nottoc,notlof,notlot]{tocbibind}

\usepackage{psfrag}

\numberwithin{equation}{section}
\numberwithin{figure}{section}

\newcommand{\nitem}[2]{#1\> \parbox[t]{10.4cm}{#2}\\[1ex]} 

\begin{document}

\author{
L\'aszl\'o Erd\H os${}^1$\thanks{Partially supported
by SFB-TR 12 Grant of the German Research Council.} \quad Antti Knowles${}^2$\thanks{Partially supported by U.S.\ 
National Science Foundation Grant DMS-0757425.} \\ \\
Institute of Mathematics, University of Munich, \\
Theresienstr.\ 39, D-80333 Munich, Germany \\ lerdos@math.lmu.de ${}^1$ \\ \\
Department of Mathematics, Harvard University\\
Cambridge MA 02138, USA \\  knowles@math.harvard.edu ${}^2$ \\ \\
\\}
\title{Quantum Diffusion and Delocalization for Band Matrices with General Distribution}

\date{7 March 2011}

\maketitle

\begin{abstract}
We consider Hermitian and symmetric random band matrices $H$ in $d \geq 1$ dimensions.
The matrix elements $H_{xy}$, indexed by $x,y \in \Lambda \subset \Z^d$, are
independent and their variances satisfy $\sigma_{xy}^2:=\E \abs{H_{xy}}^2 = W^{-d} f((x - y)/W)$ for some probability 
density $f$. We assume that the law of each matrix element $H_{xy}$ is symmetric and exhibits subexponential decay.
We prove that the time evolution of a quantum particle
subject to the Hamiltonian $H$ is diffusive on time scales $t\ll W^{d/3}$. We also show that the localization length of 
the eigenvectors of $H$ is larger than a factor $W^{d/6}$ times the band width $W$.
All results are uniform in the size $\abs{\Lambda}$ of the matrix. This extends our recent result \cite{erdosknowles} to 
general band matrices. As another consequence of our proof we show that, for a larger class of random matrices 
satisfying $\sum_x\sigma_{xy}^2=1$ for all $y$, the largest eigenvalue of $H$ is bounded with high probability by $2 + 
M^{-2/3 + \epsilon}$ for any $\epsilon > 0$, where $M \deq 1 / (\max_{x,y} \sigma_{xy}^2)$.
\end{abstract}

\vspace{1cm}

\medskip


\newpage

\section{Introduction}

We proved recently  \cite{erdosknowles}
that the quantum time evolution $\ee^{-\ii tH /2}$ generated by a band matrix $H$ with band width $W$ is diffusive on 
time scales $t \ll W^{d/3}$, where $d=1,2,3, \ldots$ is the number of spatial
 dimensions. As a consequence, we showed that typical eigenvectors are \emph{delocalized} on a scale at least $W^{1 + 
d/6}$,
i.e.\ the localization length is much larger than the band width.
A key assumption in \cite{erdosknowles} was that the matrix entries $H_{xy}$ satisfy
\begin{equation} \label{deterministic norm}
\abs{H_{xy}}^2 \;=\; \frac{1}{M} \ind{1 \leq \abs{x - y} \leq W}\,, \qquad x, y\in \Lambda,
\end{equation}
where $\Lambda$ is a large finite box in $\Z^d$ and $M\sim W^d$ is a normalization to ensure that $\sum_y |H_{xy}|^2=1$.
  For the physical significance of this result in connection with the extended states conjecture for random 
Schr\"odinger operators,
see the introduction of \cite{erdosknowles}, where we also presented an overview of
related results and references.

The goal of this paper is  to replace
the rather restrictive \emph{deterministic} condition \eqref{deterministic norm} on
the matrix elements with a natural general class of
random variables.  We consider symmetric or Hermitian random band matrices $H = (H_{xy})$
such that $\E H_{xy}=0$ and the  variances
 $\sigma_{xy}^2 \deq \E \abs{H_{xy}^2}$ are given by $\sigma_{xy}^2 = W^{-d} f\pb{(x - y)/W}$, 
where $f$ is a nonnegative function satisfying 
$\int_{\R^d} \dd x \, f(x) = 1$. Thus, $f$ describes the shape of a band of width $W$.
 The matrix entries are assumed to have an even law with subexponential decay. Under these
 assumptions we show that 
all results of \cite{erdosknowles} remain valid.

The proof of quantum diffusion for general band matrices is considerably more involved than for matrices satisfying 
\eqref{deterministic norm}.  Our proofs are based on an expansion
 in so-called nonbacktracking 
powers of $H$.  As observed by Feldheim and Sodin \cite{FSo, So1}, under the
 assumption \eqref{deterministic norm} these 
powers satisfy a simple algebraic recursion relation which immediately implies that they
 are given by Chebyshev polynomials in $H$.  In the language of perturbative quantum field theory, the nonbacktracking
powers correspond to a self-energy renormalization up to all orders.
The underlying algebraic identity, however, heavily relies on the special form
\eqref{deterministic norm}. If \eqref{deterministic norm} 
 does not hold, the renormalization is no longer
algebraically exact and the
 recursion relation becomes much more complicated.
 There are two main reasons for this complication. The first  is that the absolute value of each matrix element is 
genuinely random, and hence powers of matrix elements $|H_{xy}|^k$ cannot be replaced by a constant.
 The second reason is that the variance $\sigma_{xy}^2$ is no longer given by a
 step function in $x - y$. These two 
complications give rise to different types
of error terms that substantially increase the
complexity of the Feynman graphs to be estimated. For instance if, instead of \eqref{deterministic norm}, we assumed
\begin{equation}\label{simpler}
\sigma_{xy}^2=\E \abs{H_{xy}}^2 \;=\; \frac{1}{M}
 \ind{1 \leq \abs{x - y} \leq W}\,, \qquad x,y, \in \Lambda,
\end{equation}
i.e.\ if the band were given by a step function,
then our proof would be simpler (in the language of the graphical representation
 of Section \ref{section: path expansion}, we would not have any wiggly lines).

We remark that some of the additional complications when considering
ensembles more general than \eqref{deterministic norm} have been
tackled in \cite{FSo} and \cite{So1}. In particular,
 Feldheim and Sodin, in Section III of \cite{FSo}, describe how to
extend their result on the expectation value of
traces of Chebyshev polynomials of Wigner matrices from
\eqref{deterministic norm} to more general distributions.
In Section 9 of his paper on band matrices \cite{So1}, Sodin states
that the procedure of Section III of \cite{FSo} can  be extended to
band matrices satisfying the restriction \eqref{simpler}, but no
details are given. It seems,
however, that $\sigma_{xy}^2$ being either a fixed constant or zero plays an
important role. In this paper we consider
more general band matrices (assuming less decay of the law of the
matrix elements, and an arbitrary band shape), and we need to compute
squares of matrix elements. Hence the structure of our expansion is
more involved, and a novel approach is required to control it.

As a simple consequence of our proof, we also derive a bound on the largest eigenvalue $\lambda_{\rr{max}}$ of a band 
matrix. This result holds in fact for a more general class of random matrices in which the spatial structure (and hence 
the dependence on the spatial dimension $d$) is absent. The relevant parameter for such matrices is
\begin{equation*}
M \;\deq\; \frac{1}{\max_{x,y} \E \abs{H_{xy}^2}}\,,
\end{equation*}
characterizing, very roughly, the number of nontrivial entries in each row of $H$.
It is easy to see that, in the special case of $d$-dimensional band matrices introduced in Section \ref{section: setup}, 
we have $M \sim C W^d$ where $W$ is the band width; for a Wigner matrix we have $M = N$, where $N$ denotes the size of 
the matrix.  We show that $\lambda_{\rr{max}} \leq 2 + M^{-2/3 + \epsilon}$ with high probability for any $\epsilon > 
0$, provided that $\log N \ll M^{c \epsilon}$; here $c$ is a constant. For a smaller class of band matrices, Sodin 
\cite{So1} previously proved that $\lambda_{\rr{max}} = 2 + o(1)$ in distribution, under the assumption $\log N \ll M$.  
(In fact, for $M \ll N^{5/6}$ he computes the asymptotic integrated density of states near the spectral edge, and for $M 
\gg N^{5/6}$ he even identifies the limiting distribution of the largest eigenvalue as the Tracy-Widom distribution.) 
For other previous results on the largest eigenvalue of random band matrices see the references in \cite{So1}. In the 
special case ($M = N$) of Wigner matrices, similar estimates on the largest eigenvalue have been known for some time; we 
refer to the works of Soshnikov \cite{Soshnikov} and Vu \cite{Vu}, as well as references therein.

The outline of this paper is as follows. In Section~\ref{section: setup} 
we introduce the model and give the precise 
definition of the class of random band matrices we shall consider. Our 
main results are stated in Section~\ref{section: results}.  In Section~\ref{section: summary} 
we briefly summarize the 
Chebyshev expansion of the propagator from \cite{erdosknowles}.
 In Section~\ref{section: truncations} we perform a 
series of preliminary truncations using the subexponential decay of the matrix elements.
 The truncations are in the 
lattice size, the support of the matrix entries, and the tail of the 
Chebyshev expansion.  Section~\ref{section: path 
expansion} is devoted to a derivation of a path expansion for the propagator
 $\ee^{-\ii t H/2}$, as well as a graphical 
scheme for the various terms appearing in the expansion. In this graphical
 representation, the propagator $\ee^{-\ii t 
H/2}$ is expressed as a sum over graphs which consist of a distinguished path, called the \emph{stem},
 to which are attached trees, called \emph{boughs}. The boughs carry the error terms arising from
the non-exact renormalization.
 In Section~\ref{section: lumping} we take the expectation of
 our expansion, and describe the resulting \emph{lumpings} corresponding to higher-order cumulants.
 Section~\ref{section: bare stem} is devoted to the analysis of the bare stem, which 
yields the main contribution to our expansion. The arguments in this section
 are similar to those of \cite{erdosknowles}, except that we also need to
 analyse higher-order cumulants.  Finally, in the most involved part of the paper
we show that the contribution of the boughs is subleading.
For the convenience of the reader, we split the argument
into two parts. In Section~\ref{section: boughs} we present
a simplified proof that is valid up to time scales $t \lesssim W^{\kappa d}$ with $\kappa<1/5$.
Section~\ref{section: boughs for kappa=1/3} presents the additional arguments
needed to reach larger times scales  $t \lesssim W^{d \kappa}$ with $\kappa<1/3$. In the final Section \ref{section: 
bound on lambda max} we derive a bound on the largest eigenvalue of $H$.

We remark that the restriction $\kappa<1/3$ needs to be imposed for
 several different reasons; see the discussion in Section~\ref{sect:longsketch}.  This restriction is natural and can 
also be understood as follows. If (i) we do not resum terms associated with different $n$ and $n'$ (see \eqref{first 
expansion} below), and (ii) we do not make systematic use of detailed heat kernel bounds\footnote{As explained in 
Section 11 of \cite{erdosknowles}, this involves a refined classification of all skeleton graphs in terms of how much 
they deviate from the 2/3 rule (Lemma 7.7 in \cite{erdosknowles}).}, then our method must fail for $\kappa > 1/3$.  For 
otherwise we could prove, as in Section \ref{section: bound on lambda max}, that the largest eigenvalue of an $N \times 
N$ Wigner matrix is less than $2 + N^{-2/3 - \epsilon}$ with high probability; this is known to be false.

\subsubsection*{Conventions} We use the letters $C,c$ to denote arbitrary positive
 constants whose values are not 
important and may change from one equation to the next. They may depend on fixed parameters
 (such as $d$, $f$, $\alpha$, and $\beta$ defined below). We use $C$ for large
 constants and $c$ for small constants.  
For easy reference, we include a list of commonly used symbols and concepts 
in Appendix \ref{section: notations}.

\subsubsection*{Acknowledgements} We are grateful to a referee for suggesting improvements in the presentation as well 
as for pointing out some inaccuracies in a previous version of this manuscript.

\section{The setup} \label{section: setup}

Let the dimension $d \geq 1$ be fixed and consider the $d$-dimensional lattice $\Z^d$
equipped with the Euclidean norm $\abs{\cdot}_{\Z^d}$. We index points of $\Z^d$ with 
$x,y,z,\dots$. In order to avoid dealing with the infinite lattice directly, we restrict the 
problem to a finite periodic lattice $\Lambda_N$ of linear size $N$. More precisely,
for $N \in \N$ we set
\begin{equation*}
\Lambda_N \;\deq\; \{-[N/2], \dots, N  - 1 - [N/2]\}^d \;\subset\; \Z^d\,,
\end{equation*}
a cube with side length $N$ centred around the origin. Here $[\cdot]$ denotes integer part.
Unless stated otherwise, all summations $\sum_x$ are understood to mean $\sum_{x \in 
\Lambda_N}$.
We work on the Hilbert space $\ell^2(\Lambda_N)$, and use $\norm{\psi}$ to denote the 
$\ell^2$-norm of $\psi \in \ell^2(\Lambda_N)$. We also use $\norm{A}$ to denote the $\ell^2$ 
operator norm of $A : \ell^2(\Lambda_N) \to \ell^2(\Lambda_N)$.

For any $x \in \Z^d$ denote by $[x]_N$ the unique point in $\Lambda_N$ satisfying $x - [x]_N 
\in N \Z^d$. Define the periodic distance on $\Lambda_N$ through
\begin{equation*}
\abs{x - y} \;\deq\; \absb{[x -y]_N}_{\Z^d}\,.
\end{equation*}

We consider Hermitian (or symmetric) random band matrices $H^\omega \equiv H$ whose entries 
$H_{xy}$ are indexed by $x,y \in \Lambda_N$. Here $\omega \in \Omega$ denotes the element of a 
probability space $\Omega$. The entries $H_{xy}$ are always taken to be independent random 
variables, with the obvious restriction that $H_{yx} = \ol{H}_{xy}$.

Roughly speaking, we shall allow matrices $H$ whose variances
\begin{equation*}
\sigma_{xy}^2 \;=\; \E \abs{H_{xy}}^2
\end{equation*}
form a (doubly) stochastic matrix, such that the law of each matrix element $H_{xy}$ is symmetric.

In order to define $H$ precisely, we need the following definitions.  Let $A \equiv A^\omega$
 be a Hermitian matrix with 
independent entries that satisfy $\E \abs{A_{xy}}^2 = 1$.  (Note that we do not assume 
identical distribution of the 
entries.)
We assume that the law of $A_{xy}$ is symmetric, i.e.\ that $A_{xy}$ and $-A_{xy}$ have
 the same law. In particular, $A$ 
may be a real symmetric matrix with symmetric entries.
Moreover, we assume that the entries $A_{xy}$ have \emph{uniformly subexponential decay}: There 
exist $\alpha, \beta > 0$, independent of $x$ and $y$, such that
\begin{equation} \label{subexponential decay}
\P(\abs{A_{xy}} > \xi) \;\leq\; \beta \ee^{-\xi^\alpha}
\end{equation}
for all $x,y$ and $\xi \geq 0$. In particular, we may consider Gaussian entries.

In order to describe a band of general shape, we choose some nonnegative continuous\footnote{More generally, it suffices 
that $f$ be continuous almost everywhere. In particular, $f$ may be a step function.} function $f : \R^d \to \R$ 
satisfying
 $\int \dd x \, f(x) = 1$ and $\int \dd 
x \, f(x) \, x_i = 0$ for all $i = 1, \dots, d$.
We define
\begin{equation*}
\wt f(x) \;\deq\; \sup \{f(y)\,:\, \abs{y - x} \leq 1\}
\end{equation*}
and assume that there is a $\eta > 0$ such that
\begin{equation} \label{decay of f}
\int_{\R^d} \dd x \;  \wt f(x) \, \abs{x}^{d + 2 + \eta} \;<\; \infty\,.
\end{equation}
We also assume that the \emph{covariance matrix} $\Sigma = (\Sigma_{ij})_{1 \leq i,j \leq d}$ 
of $f$, defined by
\begin{equation} \label{covariance of f}
\Sigma_{ij} \;\deq\; \int_{\R^d} \dd x \; f(x) \, x_i x_j\,,
\end{equation}
is nonsingular.

Let $W$, $1 \leq W \leq N$, be the \emph{band width}, and
define the family of standard deviations $\sigma_{xy} \geq 0$ through
\begin{equation} \label{definition of sigma}
\sigma^2_{xy} \;\equiv\; \sigma^2_{xy}(W,f) \;\deq\; \frac{1}{M} \, f \pbb{\frac{[x - 
y]_N}{W}}\,,
\end{equation}
where
\begin{equation} \label{definition of M}
M \;\equiv\; M(W,N,f) \;\deq\; \sum_{x} f \pbb{\frac{[x]_N}{W}}\,.
\end{equation}
We then define the matrix $H$ through
\begin{equation*}
H_{xy} \;\deq\; \sigma_{xy} A_{xy}\,.
\end{equation*}

We have the asymptotic identity
\begin{equation} \label{equivalence of M and W}
\frac{M}{W^d} \;\longrightarrow\; 1
\end{equation}
as $W \to \infty$, uniformly for all $N \geq W$. In the following we make use of 
\eqref{equivalence of M and W} without further comment. For notational convenience, we use both 
$W$ and $M$ in tandem.
The definition of $H$ immediately implies that
\begin{equation} \label{variances sum to one}
\sum_{y} \E \abs{H_{xy}}^2 \;=\; \sum_y \sigma_{xy}^2 \;=\; 1
\end{equation}
for all $x$.  Moreover, by symmetry of the law of $A_{xy}$, we have
\begin{equation} \label{condition on moments}
\E H_{xy}^n \ol{H}_{xy}^m \;=\; 0
\end{equation}
whenever $n+m$ is odd.
Finally, we assume that
\begin{equation} \label{N is large enough}
N \;\geq\; W M^{1/6}\,.
\end{equation}
We regard $W$ as the free parameter.

\section{Results} \label{section: results}
As in \cite{erdosknowles}, our central quantity is
\begin{equation}\label{def:density}
\varrho(t,x) \;\deq\; \E \absb{\scalar{\delta_x}{\ee^{-\ii t H /2} \delta_0}}^2\,,
\end{equation}
where $t \in \R$ and $x \in \Lambda_N$. One readily sees that $\varrho(t, \cdot)$ is a 
probability measure on $\Z^d$ for 
all $t \in \R$, i.e.
\begin{equation} \label{rho is a probability}
\sum_x \varrho(t,x) \;=\; 1\,.
\end{equation}
The quantity $\varrho(t,x)$ has the interpretation of the probability of finding a quantum 
particle at the lattice site $x$ 
at time $t$, provided it started from the origin at time $0$. Here the time evolution of
 the quantum particle is 
governed by the Hamiltonian $H$. See \cite{erdosknowles} for more details.

We consider time scales of order $M^\kappa$ where $\kappa < 1/3$.  Thus, we set
\begin{equation*}
t \;=\; M^\kappa T\,,
\end{equation*}
where $T \geq 0$ is a quantity of order one. We consider \emph{diffusive} length scales in $x$, 
i.e.\ distances
\begin{equation*}
x \;=\; M^{\kappa/2} W X\,,
\end{equation*}
where $X$ is a quantity of order one.

Our main result generalizes Theorem 3.1 of \cite{erdosknowles} to the class of band matrices 
with general distribution and covariance introduced in Section \ref{section: setup}.

\begin{theorem} \label{theorem: main result}
Let $0 < \kappa < 1/3$ be fixed. Then for any $T_0 > 0$ and any continuous bounded function 
$\varphi \in C_b(\R^d)$ we have
\begin{equation} \label{main result}
\lim_{W \to \infty} \sum_{x\in \Lambda_N}
 \varrho\pb{W^{d \kappa} T, x} \, \varphi \pbb{\frac{x}{W^{1 + d \kappa/2}}} \;=\; \int_{\R^d} \dd 
X \; L(T, X) \, \varphi(X)\,,
\end{equation}
uniformly in $N \geq W^{1 + d/6}$ and $0 \leq T \leq T_0$. Here
\begin{equation}\label{def:L}
L(T,X) \;\deq\; \int_0^1 \dd \lambda \; \frac{4}{\pi} \frac{\lambda^2}{\sqrt{1 - \lambda^2}}
\, G(\lambda T, X)
\end{equation}
is a superposition of heat kernels
\begin{equation*}
G(T, X) \;\deq\; \frac{1}{(2 \pi T)^{d/2} \sqrt{\det \Sigma}} \, \ee^{- \frac{1}{2 T} \, X \cdot \Sigma^{-1} X}\,,
\end{equation*}
where, we recall, $\Sigma$ is the covariance matrix \eqref{covariance of f} of the
 probability density $f$.
\end{theorem}

\begin{remark}
The number $\lambda \in [0,1]$ in \eqref{def:L} represents the fraction of the macroscopic
 time $T$ that the particle spends moving effectively; the remaining fraction $1-\lambda$ of T
 represents time the particle ``wastes'' in backtracking. The 
expression \eqref{def:L} gives us an explicit formula for the probability density $\frac{4}{\pi} 
\frac{\lambda^2}{\sqrt{1 - \lambda^2}} \ind{0 \leq \lambda \leq 1}$ of the 
particle moving a fraction $\lambda$ of the 
total macroscopic time $T$.
See Section 3 of \cite{erdosknowles} for a more detailed discussion.

\end{remark}
\begin{remark}
As a corollary of Theorem \ref{theorem: main result}, we get \emph{delocalization} 
of eigenvectors of $H$ on scales $W^{1 + d\kappa/2}$. Indeed, the methods of \cite{erdosknowles}, Section 10, imply that
 the localization length of the eigenvectors of $H$ is with high probability larger than the band width times $W^{d 
\kappa /2}$. See \cite{erdosknowles}, Theorem 3.3 and Corollary 3.4, for a precise statement as well as a proof.
\end{remark}

Our methods also yield a new bound on the largest eigenvalue of a band matrix. This bound is in fact valid for a larger 
class of random matrices, for which the spatial structure and dimensionality are irrelevant.

\begin{theorem} \label{theorem: bound on lambda max}
Let the $N \times N$ matrix $A$ be as in Section \ref{section: setup}, and take a family $\{\sigma_{xy}^2\}_{x,y = 1}^N$ 
of variances that satisfy \eqref{variances sum to one}. Define
\begin{equation*}
M \;\deq\; \frac{1}{\max_{x,y} \sigma^2_{xy}}
\end{equation*}
and set $H_{xy} \deq \sigma_{xy} A_{xy}$\,.
Then there is a constant $c > 0$ such that for any $\epsilon$ satisfying $0 < \epsilon < 2/3$ we have
\begin{equation*}
\P \pB{\lambda_{\rr{max}} \geq 2 + M^{-2/3 + \epsilon}} \;\leq\; C_\epsilon N^2 \ee^{- M^{c \epsilon}}\,,
\end{equation*}
where $\lambda_{\rr{max}}$ denotes the largest eigenvalue of $H$ and $C_\epsilon$ is a constant depending on $\epsilon$.
\end{theorem}

We stress here that the condition \eqref{N is large enough} applies to Theorem \ref{theorem: main result} only, and is 
not imposed in Theorem \ref{theorem: bound on lambda max}.

The rest of this paper is devoted to the proof of Theorem \ref{theorem: main result}, with the exception of Section 
\ref{section: bound on lambda max} which contains the proof of Theorem \ref{theorem: bound on lambda max}.

\section{Summary of the Chebyshev expansion from \cite{erdosknowles}}
\label{section: summary}

 For the following, we fix $T \geq 
0$; the claimed uniformity on compacts is a trivial consequence of our analysis and we
 shall not mention it any more. 
For notational convenience, we often abbreviate 
$$
t = W^{d \kappa}T.
$$

The starting point of our proof is the same as in \cite{erdosknowles}, i.e.\ the Chebyshev 
expansion of the propagator,
\begin{equation} \label{Chebyshev expansion}
\ee^{-\ii t \xi} \;=\; \sum_{n = 0}^\infty \alpha_n(t) \, U_n(\xi)\,.
\end{equation}
Here $U_n$ denotes the $n$-th Chebyshev polynomial of the second kind,
defined through
\begin{equation} \label{definition of Chebyshev}
U_n(\cos \theta) \;\deq\; \frac{\sin (n+1)\theta}{\sin\theta}\,.
\end{equation}

For our purposes it is 
more convenient to work with the rescaled polynomials $\wt U_n(\xi) \deq U_n(\xi /2)$. They 
satisfy the recursion relation
\begin{equation} \label{recursion for Chebyshev}
\wt U_n(\xi) \;=\; \xi \wt U_{n - 1}(\xi) - \wt U_{n - 2}(\xi)
\end{equation}
as well as
\begin{equation*}
\wt U_0(\xi) \;=\; 1\,, \qquad \wt U_1(\xi) \;=\; \xi\,.
\end{equation*}

The Chebyshev transform $\alpha_n(t)$ of the propagator $\ee^{-\ii t \xi}$ was computed in 
\cite{erdosknowles} (see \cite{erdosknowles}, Lemma 5.1),
\begin{equation*}
\alpha_n(t) \;=\; 2 (-\ii)^n \frac{n+1}{t} \, J_{n + 1}(t)\,,
\end{equation*}
where $J_n(t)$ is the $n$-th Bessel function of the first kind. We shall need the following basic estimates on 
$\alpha_n(t)$; see \cite{erdosknowles}, Equations (5.4) and (7.14). We have the bound
\begin{equation} \label{bound on Bessel function}
\abs{\alpha_n(t)} \;\leq\; \frac{t^n}{n!}\,,
\end{equation}
as well as the identity
\begin{equation} \label{sum over alphas}
\sum_n \abs{\alpha_n(t)}^2 \;=\; 1\,,
\end{equation}
for all $t \in \R$.
A trivial consequence of \eqref{sum over alphas} that we shall sometimes need is
\begin{equation} \label{Bessel functions are bounded}
\abs{\alpha_n(t)} \;\leq\; 1\,,
\end{equation}
for all $n$ and $t$.

Using the Chebyshev expansion \eqref{Chebyshev expansion} we may write
\begin{equation} \label{first expansion}
\varrho(t,x) \;=\; \sum_{n, n' \geq 0}  \alpha_{n}(t) \, \ol{\alpha_{n'}(t)} \;\E \qB{\pb{\wt U_{n}(H)}_{0x} \pb{\wt 
U_{n'}(H)}_{x0}}\,.
\end{equation}
The expansion \eqref{first expansion} is the starting point of our analysis.

\section{Truncations} \label{section: truncations}
We begin the proof of Theorem \ref{theorem: main result} by introducing
 a series of truncations in the expansion 
\eqref{first expansion}. First, we truncate in the lattice size $N$ by 
showing that the error we make by assuming $N 
\leq W^C$ is negligible (see \eqref{bound on N}). Second, we use the 
subexponential decay of the matrix elements of $A$ 
to cut off $\abs{A_{xy}}$ at scales $M^\delta$ for an arbitrary $\delta > 0$. 
Third, we introduce a cutoff in the 
summation over $n$ and $n'$ in \eqref{first expansion}; this will prove
 necessary because the combinatorial estimates 
for the right-hand side of \eqref{first expansion} that we shall derive 
in Sections~\ref{section: bare stem} -- 
\ref{section: boughs for kappa=1/3} deteriorate for very large $n$ and $n'$.

\subsection{Truncation in $N$}
We replace the matrix $H$ with a truncated matrix $\wh H$, whereby we truncate in both the size of the lattice and the 
support of the distribution of the matrix entries. Both truncations are made possible by the following estimate on the 
speed of propagation of $H$.

\begin{proposition} \label{proposition: finite speed of propagation}
Let $\wt N \equiv \wt N(W) = \min\pb{W^{10d + 16},N}$ and introduce the truncated Hamiltonian 
$\wt H$ defined by
\begin{equation*}
\wt H_{xy} \;\deq\; \ind{\abs{x} \leq \wt N} \ind{\abs{y} \leq \wt N} 
H_{xy}\,.
\end{equation*}
Then there is a constant $C > 0$ such that, for all $t \leq M$ we have
\begin{equation*}
\P \pbb{\normb{\ee^{-\ii t H/2} \delta_0 - \ee^{-\ii t \wt H / 2} \delta_0} \geq \frac{C}{M}} \;\leq\; C \ee^{- 
W^\alpha}\,,
\end{equation*}
where $\alpha$ is the constant from \eqref{subexponential decay}.
\end{proposition}
\begin{proof}
See Appendix \ref{appendix: speed of propagation}.
\end{proof}

In a first step we truncate the lattice size $N$. Defining
\begin{equation*}
\wt \varrho(t,x) \;\deq\; \E \absb{\scalar{\delta_x}{\ee^{-\ii t \wt H /2} \delta_0}}^2\,,
\end{equation*}
we therefore need to estimate
\begin{equation} \label{error from lattice cutoff}
\sum_x \varphi \pbb{\frac{x}{W^{1 + d \kappa / 2}}} \pB{\wt \varrho(t,x) - \varrho(t,x)}
\end{equation}
for any $\varphi \in C_b(\R^d)$ and $t = W^{d \kappa} T$. Define the diagonal matrix $E$ through
\begin{equation*}
E_{xy} \;\deq\; \varphi \pbb{\frac{x}{W^{1 + d \kappa / 2}}} \delta_{xy}\,.
\end{equation*}
Then the absolute value of \eqref{error from lattice cutoff} is equal to
\begin{align*}
&\mspace{-40mu} \absbb{\E \qbb{\scalarb{\ee^{-\ii t \wt H/2} \delta_0}{E \ee^{-\ii t \wt H/2} \delta_0} - 
\scalarb{\ee^{-\ii t H/2} \delta_0}{E \ee^{-\ii t H/2} \delta_0}}}
\\
&\leq\; \E \normb{\ee^{- \ii t \wt H / 2} \delta_0 - \ee^{-\ii t H / 2} \delta_0} \pB{\normb{E \ee^{-\ii t H /2} 
\delta_0}
+
\normb{E \ee^{-\ii t \wt H /2} \delta_0}}
\\
&\leq\; C \, \E \normb{\ee^{- \ii t \wt H / 2} \delta_0 - \ee^{-\ii t H / 2} \delta_0}\,,
\end{align*}
where we used that $H$ and $\wt H$ are Hermitian, and $\norm{E} \leq C$. Using Proposition \ref{proposition: finite 
speed of propagation} we therefore conclude that \eqref{error from lattice cutoff} vanishes as $W \to \infty$, uniformly 
for $t \leq M$. Note that the matrix $a(W,N) \wt H$, where $a(W,N) \deq \frac{M(W,N,f)}{M(W, \wt N, f)}$, satisfies
\eqref{variances sum to one}. Since $\lim_{W\to\infty} a(W, N) = 1$, is is enough to prove Theorem \ref{theorem: main 
result} for the matrix $a(W,N) \wt H$ (it is straightforward to check that replacing $T$ with $a(W,N) T$ in our proof 
has no effect).

We conclude that it is enough to prove Theorem \ref{theorem: main result} for
\begin{equation} \label{bound on N}
N \leq W^{10d + 16}\,.
\end{equation}
We shall always assume \eqref{bound on N} from now on.

\subsection{Truncation in $\abs{A_{xy}}$}
In a second step we truncate the support of the entries of $A$. Let $\delta$ satisfy
\begin{equation} \label{condition on delta}
0 \;<\; 12 \delta \;<\; 1/3 - \kappa
\end{equation}
and define the matrix $\wh A$ through
\begin{equation} \label{truncation of matrix elements}
\wh A_{xy} \;\deq\; A_{xy} \, \ind{\abs{A_{xy}} \leq M^\delta}\,.
\end{equation}

In following we adopt the convention that adding a hat $\wh{(\cdot)}$ to a quantity $(\cdot)$ means that in the 
definition of $(\cdot)$ we replace $A$ with $\wh A$.
In particular, we set
\begin{equation*}
\wh H_{xy}: = \sigma_{xy} \wh A_{xy}\qquad \text{and}
\qquad \wh \varrho(t,x) \;\deq\; \E \absb{\scalar{\delta_x}{\ee^{-\ii t \wh H /2} \delta_0}}^2.
\end{equation*}
By the uniform subexponential decay of the entries \eqref{subexponential decay}, we have
\begin{equation*}
\P (\wh H_{xy} \neq H_{xy}) \;\leq\; 2 \, \P (\abs{A_{xy}} > M^\delta) \;\leq\; 2 \beta \, \ee^{-M^{\alpha \delta}}\,.
\end{equation*}
Therefore
\begin{equation} \label{probability bound for entry cutoff}
\P(\wh H \neq H) \;\leq\; \sum_{x,y} \P(\wh H_{xy} \neq H_{xy}) \;\leq\; 2 \, \beta \, N^{2d} \ee^{-M^{\alpha 
\delta}}\,.
\end{equation}

It is now easy to prove the main result of this subsection.

\begin{proposition} \label{proposition: cutoff in H}
We have
\begin{equation*}
\sum_x \absb{\wh \varrho(t,x) - \varrho(t,x)} \;\leq\; C \ee^{-M^c}\,.
\end{equation*}
\end{proposition}
\begin{proof}
Using the bound $\abs{\varrho(t,x)} \leq 1$, \eqref{probability bound for entry cutoff}, and \eqref{bound on N} we find
\begin{equation*}
\sum_x \absb{\wh \varrho(t,x) - \varrho(t,x)} \;\leq\; 2 N^d \, \P(\wh H \neq H) \;\leq\; 4 \, \beta \, N^{3d} \ee^{-M^{\alpha 
\delta}} \;\leq\; C \ee^{-M^c}\,. \qedhere
\end{equation*}
\end{proof}

Note that, by the definition \eqref{truncation of matrix elements}, the law of $\wh A_{xy}$ is symmetric. In particular, 
$\wh H$ satisfies \eqref{condition on moments}. Moreover, we have the following bounds on the variance of $\wh H_{xy}$.
\begin{lemma} \label{lemma: bounds on truncated variance}
There is a constant $C$ independent of $x$ and $y$ such that
\begin{equation*}
\pb{1 - C \ee^{-M^{\alpha \delta / 2}}} \sigma_{xy}^2 \;\leq\; \E \abs{\wh H_{xy}}^2 \;\leq\; \sigma_{xy}^2\,.
\end{equation*}
\end{lemma}
\begin{proof}
The upper bound is obvious from \eqref{truncation of matrix elements}. In order to prove the lower bound, we write
\begin{align*}
\sigma_{xy}^2 - \E \abs{\wh H_{xy}}^2 &\;=\; \sigma_{xy}^2 \, \E \pb{\abs{A_{xy}}^2 - \abs{\wh A_{xy}}^2}
\\
&\;\leq\; \sigma_{xy}^2 \, \E \abs{A_{xy}}^2 \, \ind{\abs{A_{xy}} \geq M^\delta}
\\
&\;=\; \sigma_{xy}^2 \int_0^\infty \dd s \; \P \pb{\abs{A_{xy}} \geq \max(\sqrt{s}, M^\delta)}
\\
&\;\leq\; \sigma_{xy}^2 \,  \beta \int_0^\infty \dd s \; \ee^{- \max(\sqrt{s}, M^\delta)^\alpha}\,,
\end{align*}
which yields the claim.
\end{proof}

\subsection{The tail of the expansion}
Now we control the tail of the expansion
\begin{equation} \label{first expansion cut off}
\wh \varrho(t,x) \;=\;
\E \qB{\absb{\scalar{\delta_x}{\ee^{-\ii t \wh H /2} \delta_0}}^2 }
\;=\;
\sum_{n, n' \geq 0}  \alpha_{n}(t) \, \ol{\alpha_{n'}(t)} \;\E \qB{\pb{\wt U_{n}(\wh H)}_{0x} \pb{\wt U_{n'}(\wh 
H)}_{x0}}\,.
\end{equation}
As observed in \cite{erdosknowles}, the coefficient $\alpha_n(t)$ is very small for $n \gg t$. Thus, we choose a cutoff 
exponent $\mu$ satisfying
\begin{equation} \label{assumption on kappa and mu}
\kappa + 4 \delta \;<\; \mu \;<\; 1/3 - 8 \delta\,.
\end{equation}
The key ingredient for controlling the tail, i.e.\ the terms $n + n' \geq M^\mu$ in \eqref{first expansion cut off}, is 
the following a priori estimate on the norm of $\wh H$.
\begin{proposition} \label{proposition: bound on H}
There are constants $C, \epsilon > 0$, depending on $\delta$, such that
\begin{equation*}
\P \pB{\norm{\wh H} \geq C M^{2 \delta}} \;\leq\; M^{-\epsilon M}
\end{equation*}
for $M$ large enough.
\end{proposition}
\begin{proof}
See Appendix \ref{section: proof of bound on H}.
\end{proof}

Defining
\begin{equation*}
\wh \varrho_b(t,x) \;\deq\; \E \qbb{\absb{\scalar{\delta_x}{\ee^{-\ii t \wh H /2} \delta_0}}^2 \, 
\indB{\norm{\wh H} \leq C M^{2 \delta}}}\,,
\end{equation*}
we find, using Proposition \ref{proposition: bound on H}, that
\begin{equation*} \label{difference between rho hat rho_b}
\sum_x \absb{\wh \varrho(t,x) - \wh \varrho_b(t,x)} \;\leq\; N^d \, \P \pB{\norm{\wh H} \geq C M^{2 \delta}} \;\leq\; 
N^d M^{-\epsilon M} \;\leq\; C M^{-c M}\,.
\end{equation*}

Next, write
\begin{equation} \label{expansion of cutoff distribution}
\wh \varrho_b(t,x) \;=\; \sum_{n,n' \geq 0} \alpha_{n}(t) \, \ol{\alpha_{n'}(t)} \;\E \qB{\pb{\wt U_{n}(\wh H)}_{0x} 
\pb{\wt U_{n'}(\wh H)}_{x0} \, \indB{\norm{\wh H} \leq C M^{2 \delta}}}\,.
\end{equation}
Split $\wh \varrho_b(t,x) = \wh \varrho_{b, \leq}(t,x) + \wh \varrho_{b, >}(t,x)$ by splitting the summation over $n,n'$ 
in \eqref{expansion of cutoff distribution} into the parts $n+n' \leq M^\mu$ and $n+n' > M^\mu$.

We now estimate $\sum_{x} \abs{\wh \varrho_{b, >}(W^{d \kappa} T,x)}$. To this end, we use the following rough estimate 
on Chebyshev polynomials.

\begin{lemma} \label{lemma: estimate on Chebyshev}
For any $n \in \N$ and $\xi \in \R$ we have
\begin{equation*}
\abs{\wt U_n(\xi)} \;\leq\; C^n (1 + \abs{\xi})^n\,.
\end{equation*}
\end{lemma}
\begin{proof}
The recursion relation \eqref{recursion for Chebyshev} combined with a simple induction argument shows that the 
coefficients of $\wt U_n$ are bounded in absolute value by $2^n$. This implies that
\begin{equation*}
\abs{\wt U_n(\xi)} \;\leq\; (n + 1) 2^n (1 + \abs{\xi})^n\,,
\end{equation*}
and the claim follows.
\end{proof}

Using Lemma \ref{lemma: estimate on Chebyshev} we therefore get
\begin{align*}
\sum_x \absb{\varrho_{b, >}(t,x)} &\;\leq\; \sum_x \sum_{n + n' > M^\mu} \absb{\alpha_{n}(t) \, \alpha_{n'}(t)} \; 
\absbb{\E \qB{\pb{\wt U_{n}(\wh H)}_{0x} \pb{\wt U_{n'}(\wh H)}_{x0} \, \indB{\norm{\wh H} \leq C M^{2 \delta}}}}
\\
&\;\leq\; N^d \sum_{n+n' > M^\mu} \absb{\alpha_{n}(t) \, \alpha_{n'}(t)} \;\absbb{\E \qB{\normb{\wt U_{n}(\wh H)} 
\normb{\wt U_{n'}(\wh H)} \, \indB{\norm{\wh H} \leq C M^{2 \delta}}}}
\\
&\;\leq\; N^d \sum_{n+n' > M^\mu} \absb{\alpha_{n}(t) \, \alpha_{n'}(t)} \; \pb{C M^{2 \delta}}^{n+n'}\,.
\end{align*}
Now from \eqref{bound on Bessel function} we get
\begin{equation*}
\abs{\alpha_{n}(t) \alpha_{n'}(t)} \;\leq\; C \frac{t^{n+n'}}{n!\, n'!} \;\leq\; C \frac{(2t)^{n+n'}}{(n + n')!}\,.
\end{equation*}
Therefore
\begin{multline} \label{tail vanishes}
\sum_x \absb{\varrho_{b, >}(W^{d \kappa} T,x)} \;\leq\; N^d \sum_{n+n' > M^\mu} \pbb{\frac{C W^{d \kappa}T}{n + 
n'}}^{n + n'} \pb{C M^{2 \delta}}^{n+n'}
\\
\leq\; N^d \sum_{n + n' > M^\mu} \pb{CT M^{\kappa + 2 \delta - \mu}}^{n + n'} \;\leq\; N^d \pb{CT M^{\kappa + 2 \delta - 
\mu}}^{M^\mu} \;\leq\; C M^{-c M^\mu}\,.
\end{multline}

Let us now consider the main term $\wh \varrho_{b,\leq}(t,x)$. In order to get a graph expansion scheme from 
\eqref{condition on moments}, we need to get rid of the conditioning on the norm of $\wh H$, i.e.\ recover the 
expression
\begin{equation} \label{definition of rho small}
\wh \varrho_{\leq}(t,x) \;\deq\; \sum_{n + n' \leq M^\mu}  \alpha_{n}(t) \, \ol{\alpha_{n'}(t)} \;\E \qB{\pb{\wt 
U_{n}(\wh H)}_{0x} \pb{\wt U_{n'}(\wh H)}_{x0}}\,.
\end{equation}
Therefore we need to estimate
\begin{multline*}
\sum_x \absb{\wh \varrho_{b,\leq}(W^{d \kappa} T,x) - \wh \varrho_{\leq}(W^{d \kappa} T,x)} \\
\leq\; \sum_x \sum_{n + n' \leq M^\mu} \absb{\alpha_{n}(W^{d \kappa} T) \, \alpha_{n'}(W^{d \kappa} T)}
\, \absbb{\E \qB{\pb{\wt U_{n}(\wh H)}_{0x} \pb{\wt U_{n'}(\wh H)}_{x0} \, \indB{\norm{\wh H} > C M^{2 \delta}}}}\,.
\end{multline*}
The expectation is estimated, using Lemma \ref{lemma: estimate on Chebyshev}, by
\begin{align*}
&\mspace{-40mu}\sum_x \absbb{\E \qB{\pb{\wt U_{n}(\wh H)}_{0x} \pb{\wt U_{n'}(\wh H)}_{x0} \, \indB{\norm{\wh H} > C 
M^{2 \delta}}}}\\
&\leq\; C^{n+n'} \E \qbb{\pb{1 + \norm{\wh H}}^{n + n'} \, \indB{\norm{\wh H} > C M^{2 \delta}}}
\\
&\leq\; C^{n + n'} \, (N^d M^\delta)^{n + n'}  \, \P \pB{\norm{\wh H} > C M^{2 \delta}}\,,
\end{align*}
where in the last step we used the trivial bound
\begin{equation*}
\norm{\wh H} \;\leq\; N^d \, M^\delta\,.
\end{equation*}
Thus, using \eqref{Bessel functions are bounded}, \eqref{bound on N}, and Proposition \ref{proposition: bound on H}, we 
find
\begin{equation} \label{difference between rho b and rho}
\sum_x \absb{\wh \varrho_{b,\leq}(W^{d \kappa} T,x) - \wh \varrho_{\leq}(W^{d \kappa} T,x)}
\;\leq\; M^{C M^\mu} \P \pB{\norm{\wh H} > C M^{2 \delta}} \;\leq\; M^{C M^\mu} M^{-\epsilon M} \;\leq\; M^{- M^c}\,
\end{equation}
as $W \to \infty$.

The following proposition summarizes our results
from this section. It shows that on time scales $t\lesssim W^{d\kappa}$, instead of the original
density $\varrho(x,t)$ defined in \eqref{def:density} it will be sufficient to
deal with the density $\wh \varrho_{\leq}(x,t)$ of the truncated dynamics  defined in \eqref{definition of rho small}. 
In the rest of the paper we shall work with  $\wh \varrho_{\leq}(x,t)$.

\begin{proposition}\label{prop:summary of truncation}
We have
\begin{equation*}
\sum_x \absb{\varrho(W^{d \kappa} T, x) - \wh \varrho_{\leq}(W^{d \kappa} T,x)} \;\leq\; M^{-M^c}
\end{equation*}
for some $c>0$, where $\wh \varrho_{\leq}$ is defined in \eqref{definition of rho small}.
\end{proposition}

\begin{proof}
Proposition \ref{prop:summary of truncation} is an immediate consequence of 
Proposition~\ref{proposition: cutoff in H} and the
equations 
\eqref{difference between rho hat rho_b}, \eqref{tail vanishes}, and \eqref{difference between rho b and rho}.
\end{proof}

Note moreover that in the definition \eqref{definition of rho small} the sum ranges only over indices $n$ and $n'$ such 
that $n + n'$ is even. This follows from the fact that $U_n$ is odd (even) for odd (even) $n$, and that $\wh H$ 
satisfies the moment condition \eqref{condition on moments}.

\section{The path expansion} \label{section: path expansion}

In this section we develop a graphical expansion to compute the
matrix elements of $\wt U_n(\wh H)$ needed to evaluate  $\wh \varrho_{\leq}(x,t)$; see \eqref{definition of rho small}.  
The result of this expansion is summarized in Proposition \ref{proposition: U_n as a sum over graphs}, which expresses 
$\wt U_n(\wh H)$ as a sum over graphs. The main idea is that, thanks to the special properties of the Chebyshev
polynomials, we can express $\wt U_n(\wh H)$ in terms of nonbacktracking powers
of $\wh H$, up to some error terms.
The nonbacktracking powers
make it easier to identify the main terms and the error terms
in the computation of  the expectation in \eqref{definition of rho small}.
The expectation will be computed in Section \ref{section: lumping}
by introducing an additional structure, the lumping of edges, to
the graphical representation.  Eventually, the main terms will correspond to certain very simple graphs with a trivial 
lumping (ladders) and their contribution
yields the final limiting equation (Section \ref{section: bare stem}).
 The contribution of all other nontrivial graphs or nontrivial lumpings will be negligible in the $W\to \infty$ limit;
the estimate of these error terms constitutes the rest of the paper.

\subsection{Derivation of the expansion}
For $n \in \N$ abbreviate
\begin{equation*}
U_n \;\deq\; \wt U_n(\wh H)\,.
\end{equation*}
(Note that in \eqref{Chebyshev expansion} $U_n=U_n(\xi)$ denoted
the standard Chebyshev polynomials, but for the rest of the paper
we shall use $U_n$ to denote the matrix $\wt U_n(\wh H)$.)
Thus we have
\begin{subequations} \label{defining properties of U_n}
\begin{equation}
U_0 \;=\; \umat \,,\qquad
U_1 \;=\; \wh H\,,\qquad
U_2 \;=\; \wh H^2 - \umat\,,
\end{equation}
as well as
\begin{equation}
U_n \;=\; \wh H U_{n - 1} - U_{n - 2} \qquad (n \geq 2)\,.
\end{equation}
\end{subequations}

Next, for $n \geq 2$ we define $V_n$ as the $n$-th \emph{nonbacktracking power} of $\wh H$, 
i.e.
\begin{equation*}
(V_n)_{x_0 x_n} \;\deq\; \sum_{x_1, \dots, x_{n - 1}} \qBB{\prod_{i = 0}^{n - 2} \ind{x_i \neq 
x_{i+2}}} \, \wh H_{x_0 x_1} \wh H_{x_1 x_2} \cdots \wh H_{x_{n - 1} x_n}\,.
\end{equation*}
We also define $V_0 \deq \umat$, $V_1 \deq \wh H$, and $V_n \deq 0$ for $n < 0$.
In order to derive a recursion relation for $V_n$, we define the matrices $\Phi_2$ and $\Phi_3$ 
through
\begin{subequations} \label{definition of small matrices}
\begin{align} \label{definition of Phi_2}
(\Phi_2)_{xy} &\;\deq\; \delta_{xy} \pbb{\sum_z \abs{\wh H_{xz}}^2 - 1} \;=\; \delta_{xy} 
\sum_z \pbb{\abs{\wh H_{xz}}^2 - \sigma^2_{xz}}\,,
\\
(\Phi_3)_{xy} &\;\deq\; - \abs{\wh H_{xy}}^2 \wh H_{xy}\,,
\end{align}
\end{subequations}
where in \eqref{definition of Phi_2} we used \eqref{variances sum to one}.
Moreover, we introduce the shorthand $\ul{\Phi_3 V_n}$, defined by
\begin{equation} \label{nonbacktracking product}
(\ul{\Phi_3 V_n})_{x_0 x_{n + 1}} \;\deq\; \sum_{x_1, \dots, x_n} \qBB{\prod_{i = 0}^{n-1} 
\ind{x_i \neq x_{i + 2}}} \, (\Phi_3)_{x_0 x_1} \wh H_{x_1 x_2} \cdots \wh H_{x_n x_{n+1}}\,;
\end{equation}
we use the convention that $\ul{\Phi_3 V_0} = \Phi_3$.
\begin{lemma} \label{lemma: defining properties of V_n}
We have that
\begin{equation*}
V_0 \;=\; \umat \,, \qquad V_1 \;=\; \wh H\,, \qquad V_2 \;=\; \wh H^2 - \umat - \Phi_2\,,
\end{equation*}
as well as
\begin{equation*}
V_n \;=\; \wh H V_{n - 1} - V_{n - 2} - \Phi_2 V_{n - 2} - \ul{\Phi_3 V_{n - 3}} \qquad (n \geq 
2)\,.
\end{equation*}
\end{lemma}

\begin{proof}
The expressions for $V_0, V_1, V_2$ are easy to derive from the definition of $V_n$.  Moreover, 
for $n \geq 3$ we find
\begin{align*}
(\wh H V_{n - 1})_{x_0 x_n} &\;=\; \sum_{x_1, \dots, x_{n - 1}} \qBB{\prod_{i = 1}^{n - 2} 
\ind{x_i \neq x_{i + 2}}} \, \wh H_{x_0 x_1} \wh H_{x_1 x_2} \cdots \wh H_{x_{n - 1} x_n}
\\
&\;=\; \sum_{x_1, \dots, x_{n - 1}} \qBB{\prod_{i = 0}^{n - 2} \ind{x_i \neq x_{i + 2}}} \, \wh 
H_{x_0 x_1} \wh H_{x_1 x_2} \cdots \wh H_{x_{n - 1} x_n}
\\
&\qquad + \sum_{x_1, \dots, x_{n - 1}} \ind{x_0 = x_2} \qBB{\prod_{i = 1}^{n - 2} \ind{x_i \neq 
x_{i + 2}}} \, \wh H_{x_0 x_1} \wh H_{x_1 x_2} \cdots \wh H_{x_{n - 1} x_n}
\\
&\;=\; (V_n)_{x_0 x_n} + \sum_{x_1, \dots, x_{n - 1}} \ind{x_0 = x_2} \qBB{\prod_{i = 2}^{n - 
2} \ind{x_i \neq x_{i + 2}}} \, \wh H_{x_0 x_1} \wh H_{x_1 x_2} \cdots \wh H_{x_{n - 1} x_n}
\\
&\qquad - \sum_{x_1, \dots, x_{n - 1}} \ind{x_0 = x_2} \ind{x_1 = x_3} \qBB{\prod_{i = 2}^{n - 
2} \ind{x_i \neq x_{i + 2}}} \, \wh H_{x_0 x_1} \wh H_{x_1 x_2} \cdots \wh H_{x_{n - 1} x_n}
\\
&\;=\; (V_n)_{x_0 x_n} + \sum_{x_1} \abs{\wh H_{x_0 x_1}}^2 \, (V_{n - 2})_{x_0 x_n} + 
(\ul{\Phi_3 V_{n - 3}})_{x_0 x_n}\,,
\end{align*}
by \eqref{nonbacktracking product}.
This yields
\begin{equation*}
\wh H V_{n - 1} \;=\; V_n + V_{n - 2} + \Phi_2 V_{n - 2} + \ul{\Phi_3 V_{n - 3}}\,,
\end{equation*}
and the claim follows.
\end{proof}

We may now derive the path expansion of $U_n$. To streamline notation, it is convenient to define $\ul{\Phi_2 V_n} \deq 
\Phi_2 V_n$.

\begin{proposition} \label{proposition: main path expansion}
We have
\begin{equation} \label{main path expansion}
U_n \;=\; \sum_{k \geq 0} \sum_{a \in \{2,3\}^k} \; \sum_{\ell_0 + \cdots + \ell_k = n - \abs{a}} V_{\ell_0} \, 
\ul{\Phi_{a_1} V_{\ell_1}} \cdots \ul{\Phi_{a_k} V_{\ell_k}}\,,
\end{equation}
where the sum ranges over $\ell_i \geq 0$ for $i = 0, \dots, k$. Here we use the abbreviation 
$a = (a_1, \dots, a_k)$ as well as $\abs{a} \deq \sum_{i = 1}^k a_i$.
\end{proposition}
\begin{proof}
Define the matrix $D_n$ through
\begin{equation*}
U_n \;=\; V_n + D_n\,.
\end{equation*}
It is easy to see from \eqref{defining properties of U_n} and Lemma \ref{lemma: defining 
properties of V_n} that
\begin{equation*}
D_0 \;=\; 0\,, \qquad D_1 \;=\; 0\,, \qquad D_2 \;=\; \Phi_2\,,
\end{equation*}
as well as
\begin{equation} \label{recursion for D_n}
D_n \;=\; \wh H D_{n - 1} - D_{n - 2} + \Phi_2 V_{n - 2} + \ul{\Phi_3 V_{n - 3}}\,.
\end{equation}

We prove
\begin{equation} \label{path expansion for D}
D_n \;=\; \sum_{k \geq 1} \sum_{a \in \{2,3\}^k} \; \sum_{\ell_0 + \cdots + \ell_k = n - \abs{a}} V_{\ell_0} \, 
\ul{\Phi_{a_1} V_{\ell_1}} \cdots \ul{\Phi_{a_k} V_{\ell_k}}
\end{equation}
using a simple induction argument. The cases $n = 0,1,2$ are trivial.  Assuming the claim holds up to $n - 1$, we get 
from \eqref{recursion for D_n}
\begin{align*}
D_n &\;=\; \sum_{k \geq 1} \sum_{a \in \{2,3\}^k} \; \sum_{\ell_0 + \cdots + \ell_k = n - \abs{a} - 1} \wh H \, 
V_{\ell_0} \, \ul{\Phi_{a_1} V_{\ell_1}} \cdots \ul{\Phi_{a_k} V_{\ell_k}}
\\
&\qquad -  \sum_{k \geq 1} \sum_{a \in \{2,3\}^k} \; \sum_{\ell_0 + \cdots + \ell_k = n - \abs{a} - 2} V_{\ell_0} \, 
\ul{\Phi_{a_1} V_{\ell_1}} \cdots \ul{\Phi_{a_k} V_{\ell_k}}
\\
&\qquad
+ \ul{\Phi_2 V_{n - 2}} + \ul{\Phi_3 V_{n - 3}}
\\
&\;=\; \sum_{k \geq 1} \sum_{a \in \{2,3\}^k} \; \sum_{\ell_0 + \cdots + \ell_k = n - \abs{a} - 2} \wh H \, V_{\ell_0 + 
1} \, \ul{\Phi_{a_1} V_{\ell_1}} \cdots \ul{\Phi_{a_k} V_{\ell_k}}
\\
&\qquad -  \sum_{k \geq 1} \sum_{a \in \{2,3\}^k} \; \sum_{\ell_0 + \cdots + \ell_k = n - \abs{a} - 2} V_{\ell_0} \, 
\ul{\Phi_{a_1} V_{\ell_1}} \cdots \ul{\Phi_{a_k} V_{\ell_k}}
\\
&\qquad
+ \ul{\Phi_2 V_{n - 2}} + \ul{\Phi_3 V_{n - 3}}
+ \sum_{k \geq 1} \sum_{a \in \{2,3\}^k} \; \sum_{\ell_1 + \cdots + \ell_k = n - \abs{a} - 1}  V_1 \, \ul{\Phi_{a_1} 
V_{\ell_1}} \cdots \ul{\Phi_{a_k} V_{\ell_k}}
\\
&\;=\; \sum_{k \geq 1} \sum_{a \in \{2,3\}^k} \; \sum_{\ell_0 + \cdots + \ell_k = n - \abs{a} - 2} V_{\ell_0 + 2} \, 
\ul{\Phi_{a_1} V_{\ell_1}} \cdots \ul{\Phi_{a_k} V_{\ell_k}}
\\
&\qquad + \sum_{k \geq 1} \sum_{a \in \{2,3\}^k} \; \sum_{\ell_0 + \cdots + \ell_k = n - \abs{a} - 2} \pb{\ul{\Phi_2 
V_{\ell_0}} + \ul{\Phi_3 V_{\ell_0 - 1}}} \, \ul{\Phi_{a_1} V_{\ell_1}} \cdots \ul{\Phi_{a_k} V_{\ell_k}}
\\
&\qquad
+ \ul{\Phi_2 V_{n - 2}} + \ul{\Phi_3 V_{n - 3}}
+ \sum_{k \geq 1} \sum_{a \in \{2,3\}^k} \; \sum_{\ell_1 + \cdots + \ell_k = n - \abs{a} - 1}  V_1 \, \ul{\Phi_{a_1} 
V_{\ell_1}} \cdots \ul{\Phi_{a_k} V_{\ell_k}}\,,
\end{align*}
where in the last step we used Lemma \ref{lemma: defining properties of V_n}. Thus \eqref{path expansion for D} is 
proved.

Finally, \eqref{main path expansion} is an immediate consequence of \eqref{path expansion for 
D}.
\end{proof}

\subsection{Graphical representation}
The path expansion \eqref{main path expansion} is the key algebraic identity of our proof. We 
now introduce a graphical representation of \eqref{main path expansion} by associating a rooted tree graph $G$ with each 
summand in \eqref{main path expansion}.

\begin{figure}[ht!]
\begin{center}
\includegraphics{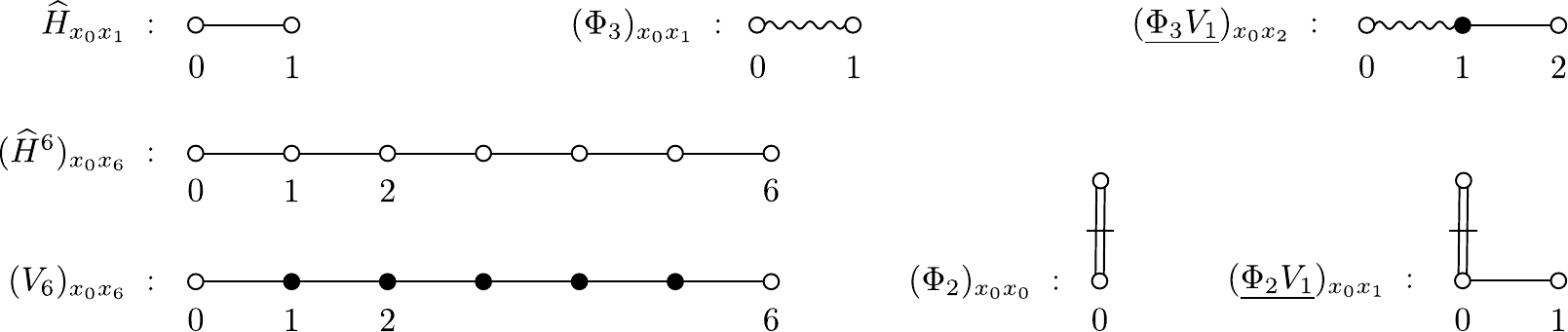}
\end{center}
\caption{The basic graphical units. \label{figure: basic definitions}}
\end{figure}
Before giving a precise definition of our graphs, we outline how they arise from \eqref{main 
path expansion}.
A matrix element $\wh H_{x_0 x_1}$ is represented by two vertices, $0$ and $1$. To each vertex 
$v$ we assign a label $x_v \in \Lambda_N$.  Matrix multiplication is represented by 
concatenating such edges. Thus, $\wh H_{x_0 x_1} \cdots \wh H_{x_{n - 1} x_n}$ is represented 
as a sequence of vertices $0, \dots, n$ joined by $n$ edges. The root is always the leftmost 
vertex, and the edges are directed away from the root. If two neighbouring vertices $u,w$ of a 
vertex $v$ are constrained to have different labels (the nonbacktracking condition), we draw 
$v$ using a black dot; otherwise, we draw $v$ using a white dot.
A factor $\Phi_2$ gives rise to a directed edge, represented by a slashed double line, whose 
final vertex is ``dangling'' in the sense that it has degree one. A factor $\Phi_3$ is 
represented by a wiggly edge.
See Figure \ref{figure: basic definitions} for an illustration of these rules.

Using these graphical building blocks we may conveniently represent any summand of \eqref{main 
path expansion}. See Figure \ref{figure: example path} for an example.
\begin{figure}[ht!]
\begin{center}
\includegraphics{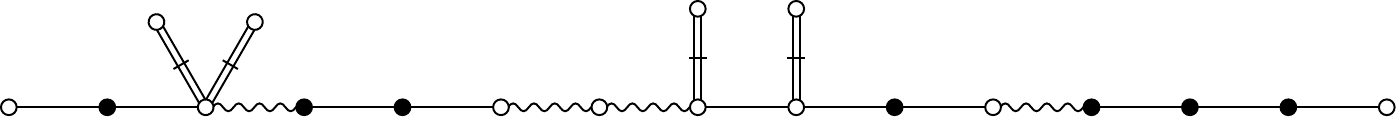}
\end{center}
\caption{Graphical representation of the term $V_2 \, \Phi_2 \, \Phi_2 \, \underline{\Phi_3 
V_2} \, \Phi_3 \, \Phi_3 \, \underline{\Phi_2 V_1} \, \underline{\Phi_2 V_2} \, 
\underline{\Phi_3 V_3}$. \label{figure: example path}}
\end{figure}

\subsection{Definition of graphs} \label{page: graphs}
We now give a precise definition of a set of graphs that is sufficiently general for our 
purposes. Let $G$ be a finite, oriented, unlabelled, rooted tree. We denote by $\cal V(G)$ the 
set of vertices of $G$, by $\cal E(G)$ the set of edges of $G$, and by $a(G) \in \cal V(G)$ the 
root of $G$.
That $G$ is oriented means that $G$ is drawn in the plane, and the edges incident to any vertex 
are ordered. (Thus, each edge $e$ adjacent to a vertex $v$ has a successor, defined as the next edge adjacent to $v$ 
counting anticlockwise from $e$.) In particular, two graphs are considered different even if they are isomorphic in the 
usual graph-theoretical sense but the ordering of the edges at some vertex differs. This notion of orientation can be 
formalized using \emph{Dick paths} (see e.g.\ \cite{AGZ}, Chapter 1). Such a formal definition is not necessary for our 
purposes however.

The choice of a root $a(G)$ implies that we may view $G$ as a directed graph, whereby edges are 
directed away from the root. Thus we shall always regard an edge $e = (v,w)$ as an ordered pair 
of vertices. Given an edge $e = (v,w) \in \cal E(G)$, we denote by $a(e) = v$ the initial 
vertex of $e$ and by $b(e) = w$ the final vertex of $e$.

There is a natural notion of distance between vertices: For $v,w \in \cal V(G)$ we set $d(v,w)$ 
to be equal to the number of edges in the shortest path from $v$ to $w$.
Each vertex $v \neq a(G)$ has a \emph{parent} $w$, defined as the unique vertex adjacent to $v$ 
and satisfying $d(a(G), w) = d(a(G),v) - 1$. If $w$ is the parent of $v$ we also say that $v$ 
is a \emph{child} of $w$. Similarly, if an edge $e$ is not incident to $a(G)$, we call the 
(unique) edge $e'$ satisfying $a(e) = b(e')$ the \emph{parent} of $e$; in this case we also 
call $e$ a \emph{child} of $e'$.

\label{page: stem}
We require that $G$ have an additional distinguished vertex $b(G) \in \cal V(G)$, which need 
not be different from $a(G)$. The path connecting $a(G)$ to $b(G)$ is called the \emph{stem} of 
$G$, and denoted by $\cal S(G)$. When drawing $G$ in the plane, we draw the stem as a 
horizontal path from $a(G)$ at its left edge to $b(G)$ at its right edge. We require that all 
edges not belonging to the stem lie above it (see Figure \ref{figure: bare tree}).
Ultimately, the vertices $a(G)$ and $b(G)$ will receive the fixed labels $x_{a(G)} = x$ and 
$x_{b(G)} = y$ in the graphical expansion of the matrix element $(U_n)_{xy}$.

We denote the set of such graphs by  $\fra W$. We call an edge $e \in \cal E(G)$ a \emph{stem 
edge} if it belongs to $\cal E(\cal S(G))$, and a \emph{bough edge} otherwise.  If $G$ has no 
bough edges, we call it a \emph{bare stem}. A bare stem is uniquely determined by its number of 
edges.

\label{page: boughs}
Thus, a graph $G \in \fra W$ consists of a stem and a collection of rooted trees, called \emph{boughs}. Each bough is 
directed away from its root vertex, which belongs to the stem $\cal S(G)$. We abbreviate with $\cal B(G)$ the subgraph 
of $G$ consisting of all bough edges. We call a bough edge $e \in \cal E (\cal B(G)) = \cal E(G) \setminus \cal E(\cal 
S(G))$ a \emph{leaf} if $b(e)$ has degree one.
See Figure \ref{figure: bare tree} for an example of a graph in $\fra W$.
\begin{figure}[ht!]
\begin{center}
\includegraphics{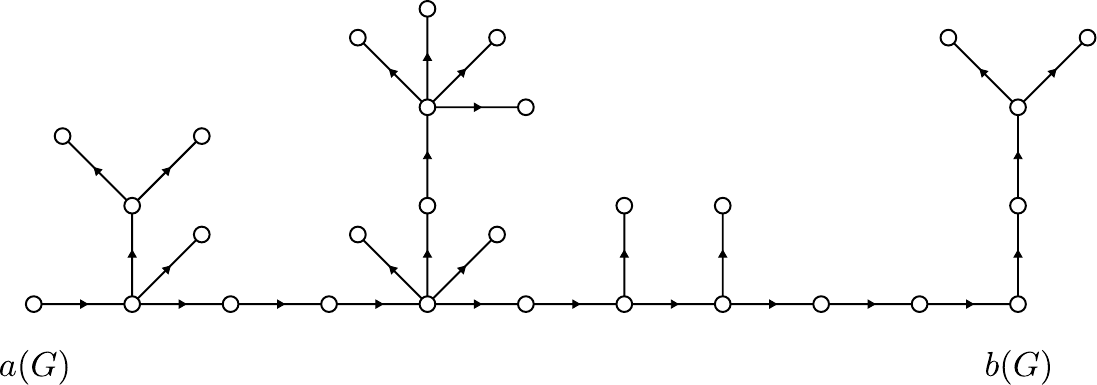}
\end{center}
\caption{A graph in $\fra W$. \label{figure: bare tree}}
\end{figure}

\label{page: tags}
Next, we \emph{decorate} graphs $G \in \fra W$ as follows. First, we \emph{tag} the edges,
 i.e.\ we choose a map 
$\tau_G$ on $\cal E(G)$, called a \emph{tagging}, with values in the \emph{set of tags}
\begin{equation}
\hb{(s, 0), (s, 1), (b,0), (b, 1), (b, 2), (b,3), (b,4)}\,.
\end{equation}
Here $s$ stands for ``stem'' and $b$ for ``bough''. We require that the tag $\tau_G(e)$
 be of the form $(s,i)$ if $e \in 
\cal E(\cal S(G))$ and of the form $(b,i)$ otherwise. The index $i$ (taking values
 in $\{0,1\}$ for stem edges and $\{0, 
\dots, 4\}$ for bough edges) is used to tag different types of edges. Edges whose
 tag is $(s,0)$ or $(b,0)$ are called 
\emph{large}; other edges are called \emph{small}. The reason for this nomenclature
 lies in the magnitude of their 
contribution to the value of the graph after taking the expectation; see Section \ref{section: boughs}.
\label{page: l}
Second, we choose a symmetric map $l_G : \cal V(G)^2 \to \{0,1\}$ which will be used to encode all nonbacktracking 
conditions on $G$. The idea is that $l_{G}(v,w) = 1$ induces a constraint $x_v \neq x_w$ on the labels. We require that 
$l(v,w) = 0$ unless $d(v,w) = 2$. We call the triple $(G, \tau_G, l_G)$ a \emph{decorated graph}, and denote the set of 
decorated graphs by $\fra G$. \label{page: decorated graph}

\label{page: label}
Next, we associate a \emph{value} $\fra V_{xy}(\cal G)$ with each decorated graph $\cal G \in 
\fra G$. The value $\fra V_{xy}(\cal G)$ is a random variable that depends on two labels $x,y 
\in \Lambda_N$. For the following we fix $\cal G = (G, \tau_G, l_G)$. We shall assign a 
\emph{label} $x_v \in \Lambda_N$ to each vertex $v \in \cal V(G)$ in such a way that $x = 
x_{a(G)}$ and $y = x_{b(G)}$. To define $\fra V_{xy}(G)$ we first assign a polynomial in the 
matrix entries to each edge. Let $e \in \cal E(G)$ and abbreviate $x_0 = x_{a(e)}$ and $x_1 = 
x_{b(e)}$.  We associate a polynomial $P_{\tau_G(e)}(\wh H_{x_0 x_1}, \wh H_{x_1 x_0})$, and a 
\emph{degree} $\deg_{\tau_G}(e) \equiv \deg(e)$, with $e$ according to the following table.
\begin{center}
\begin{tabular}{c|c|c} \label{polynomial table}
$\tau_G(e)$ & $P_{\tau_G(e)}(\wh H_{x_0 x_1}, \wh H_{x_1 x_0})$ & $\deg(e)$
\\[0.3em]
\hline
 & &
\\[-0.7em]
$(s, 0)$ & $\wh H_{x_0 x_1}$ & $1$
\\[0.3em]
$(s, 1)$ &  $-\abs{\wh H_{x_0 x_1}}^2 \wh H_{x_0 x_1}$ & $3$
\\[0.3em]
$(b, 0)$  & $\abs{\wh H_{x_0 x_1}}^2$ & $2$
\\[0.3em]
$(b, 1)$ & $\abs{\wh H_{x_0 x_1}}^2 - \sigma^2_{x_0 x_1}$ & $2$
\\[0.3em]
$(b, 2)$, $(b, 3)$  & $-\abs{\wh H_{x_0 x_1}}^4$ & $4$
\\[0.3em]
$(b, 4)$ & $\abs{\wh H_{x_0 x_1}}^6$ & $6$
\end{tabular}
\end{center}
\label{page: degree}
Note that $\deg(e)$ is nothing but the degree of the polynomial $P_{\tau_G(e)}$. The 
\emph{degree} of $\cal G$ is
\begin{equation}
\deg(\cal G) \;\deq\; \sum_{e \in \cal E(G)} \deg(e)\,.
\end{equation}

In order to define $\fra V_{xy}(\cal G)$ it is convenient to abbreviate the family of labels by 
$\b x = \pb{x_v : v \in \cal V(G)}$.
Then we set
\begin{equation} \label{definition of value of graph}
\fra V_{xy}(\cal G) \;\deq\; \sum_{\b x} \delta_{x x_{a(G)}} \delta_{y x_{b(G)}} 
\qBB{\prod_{v,w \in \cal V(G)} \pb{1 - l_G(v,w) \delta_{x_v x_w}}}
\qBB{\prod_{e \in \cal E(G)} P_{\tau_G(e)}\pb{\wh H_{x_{a(e)} x_{b(e)}}, \wh H_{x_{b(e)} 
x_{a(e)}}}}\,.
\end{equation}
The summation over $\b x$ means unrestricted summation for all $x_v \in \Lambda_N$, $v \in \cal 
V(G)$.

We call a stem vertex $v \in \cal V(\cal S(G)) \setminus \{a(G), b(G)\}$ \emph{nonbacktracking} 
if the two stem edges adjacent to $v$, $(u,v)$ and $(v,w)$, satisfy $l_G(u,w) = 1$; according 
to \eqref{definition of value of graph}, this means that we have the constraint $x_u \neq x_w$.  
Otherwise we call $v$ \emph{backtracking}. We call the stem $\cal S(G)$ \emph{completely 
nonbacktracing} if all vertices in $\cal V(\cal S(G)) \setminus \{a(G), b(G)\}$ are 
nonbacktracking.
Decorated graphs $(G, \tau_G, l_G) \in \fra G$ are represented graphically as follows. Each 
edge of $G$ is drawn using a decoration that identifies its tag $\tau_G(e)$; see Figure 
\ref{figure: edge decorations}. (Note that, although Figure \ref{figure: edge decorations} 
suggests that decorated bough edges are double, they are in fact single. This graphical 
representation using double lines is chosen in the light of the graph operations $\cal F_n$, 
$\cal F_c$, and $\cal R$ defined below.) Non-backtracking stem vertices are drawn with a black 
dot; other vertices are drawn with a white dot. Note that using black and white dots to draw the vertices displays only 
partial information about $l_G$: Only nonbacktracking restrictions pertaining to pairs of vertices both in the stem are 
indicated in our graphical representation.
\begin{figure}[ht!]
\begin{center}
\includegraphics{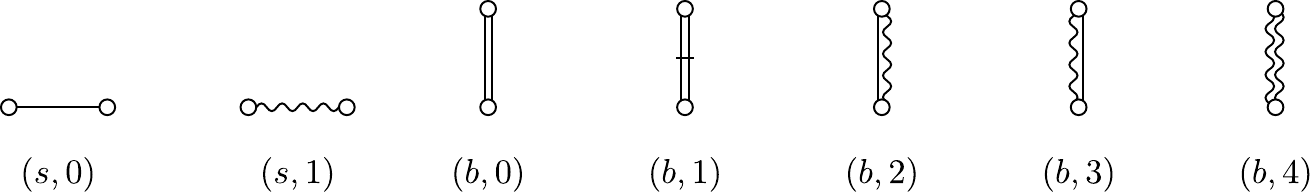}
\end{center}
\caption{The edge decorations along with their associated tags. \label{figure: edge 
decorations}}
\end{figure}
See Figure \ref{figure: decorated tree} for an example of a decorated graph.
\begin{figure}[ht!]
\begin{center}
\includegraphics{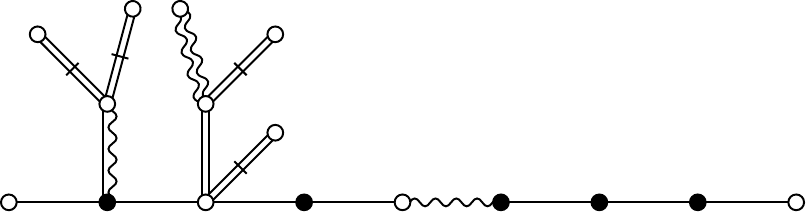}
\end{center}
\caption{A decorated graph in $\fra G$. \label{figure: decorated tree}}
\end{figure}

\subsection{Operations on graphs} \label{section: operations on graphs}
As it turns out, in order to control the graph expansion we shall have to make all stem 
vertices apart from $a(G)$ and $b(G)$ nonbacktracking. To this end, we introduce two 
operations, $\cal F_n$ and $\cal F_c$, on the set of decorated graphs $\fra G$. We shall prove 
that after a finite number of successive applications of either $\cal F_n$ or $\cal F_c$ to an 
arbitrary decorated graph, we always get a graph with a completely nonbacktracking stem. The 
index $n$ stands for ``nonbacktracking'' and $c$ for ``collapsing''.  The idea behind the 
definition of $\cal F_n$ and $\cal F_c$ is to choose the first (in the natural order of $\cal 
S(G)$) backtracking stem vertex $v_1 \in \cal V( \cal S(G)) \setminus \{a(G), b(G)\}$ and 
introduce a splitting in the definition \eqref{definition of value of graph} using
\begin{equation*}
1 \;=\; \ind{x_{v_0} \neq x_{v_2}} + \ind{x_{v_0} = x_{v_2}}\,,
\end{equation*}
where the vertices $v_0, v_2 \in \cal V(\cal S(G))$ are the neighbours of $v_1$ in the stem, 
i.e.\ they satisfy $(v_0,v_1),(v_1,v_2) \in \cal E(\cal S(G))$. 

We now define $\cal F_n$ and $\cal F_c$ more precisely. If $\cal S(G)$ has no backtracking 
vertex, set $\cal F_n(\cal G) \deq \cal G$ and $\cal F_c(\cal G) \deq \emptyset$, where 
$\emptyset$ is the empty graph satisfying $\fra V_{xy}(\emptyset) \deq 0$.

Otherwise, let $v_1$ be the first backtracking vertex in
$\cal V(\cal S(G)) \setminus \{a(G), b(G)\}$ and define $v_0$ and $v_2$ as above. Then we set 
$\cal F_n(\cal G) \deq (G, \tau_G, \wt l_G)$, where
\begin{equation*}
\wt l_G(v,w) \;\deq\; l_G(v,w) + \delta_{v v_0} \delta_{w v_2} + \delta_{w v_0} \delta_{v 
v_2}\,.
\end{equation*}
Thus, the operation $\cal F_n$ simply makes the vertex $v_1$ a nonbacktracking vertex of $\cal 
S(G)$ without changing $G$ or $\tau_G$, i.e.\ it sets $\wt l_G(v_0, v_2) = \wt l_G(v_2, v_0) = 1$ and leaves $\wt l_G$ 
unchanged for any other pair of vertices. \label{page: F_n}

Next, we define $\cal F_c$. Let $v_0,v_1,v_2$ be as above. The operation $\cal F_c$ collapses 
the two nearest stem neighbours, $v_0$ and $v_2$, of $v_1$ into one vertex and fuses the two 
edges $(v_0, v_1)$ and $(v_1, v_2)$ into one edge (see Figure \ref{figure: collapsing 
operation}). This definition is very natural in the light of Figure \ref{figure: collapsing 
operation} and our choice of conventions for drawing bough edges as double lines. Thus, a 
reader who believes his eyes when gazing at pictures like Figure \ref{figure: collapsing 
operation} may safely skip the following two paragraphs.

To define the operation $\cal F_c$ precisely, we identify $v_0$ with $v_2$, i.e.\ introduce the 
equivalence classes
\begin{equation*}
[v] \;\deq\;
\begin{cases}
\{v\} & \text{if } v \notin \{v_0,v_2\}
\\
\{v_0,v_2\} & \text{if } v \in \{v_0,v_2\}\,.
\end{cases}
\end{equation*}
Define the graph $\wt G$ through its vertex set $\cal V (\wt G) = \{[v] \,:\, v \in \cal 
V(G)\}$, and its edge set, which is obtained as follows. Each edge $(v, w) \in \cal E(G) 
\setminus \{(v_1,v_2)\}$ gives rise to the edge $([v], [w]) \in \cal E(\wt G)$. Thus, the edges 
$(v_0, v_1)$ and $(v_1, v_2)$ are fused into a single edge $([v_0], [v_1])$.
The tag $\tau_{\wt G}
\pb{([v],[w])}$ is by definition equal to the tag $\tau_G\pb{(v,w)}$ if $(v,w) \neq (v_0, 
v_1)$; the tag of the edge $([v_0], [v_1])$ is defined by the following table.
\begin{center}
\begin{tabular}{cc|c}
$\tau_G\pb{(v_0, v_1)}$ & $\tau_G\pb{(v_1,v_2)}$ & $\tau_{\wt G}\pb{([v_0],[v_1])}$
\\[0.3em]
\hline
 & &
\\[-0.7em]
$(s, 0)$ & $(s,0)$ & $(b,0)$
\\[0.3em]
$(s, 0)$ & $(s,1)$ & $(b,2)$
\\[0.3em]
$(s, 1)$ & $(s,0)$ & $(b,3)$
\\[0.3em]
$(s, 1)$ & $(s,1)$ & $(b,4)$
\end{tabular}
\end{center}
The initial and final vertices of $\wt G$ are given by $a(\wt G) \deq [a(G)]$ and $b(\wt G) 
\deq [b(G)]$.  The edges of $\wt G$ are oriented in the natural way when drawing $G$ and $\wt 
G$ in the plane; instead of giving a formal definition of the orientation, we refer to Figure 
\ref{figure: collapsing operation}.
\begin{figure}[ht!]
\begin{center}
\includegraphics{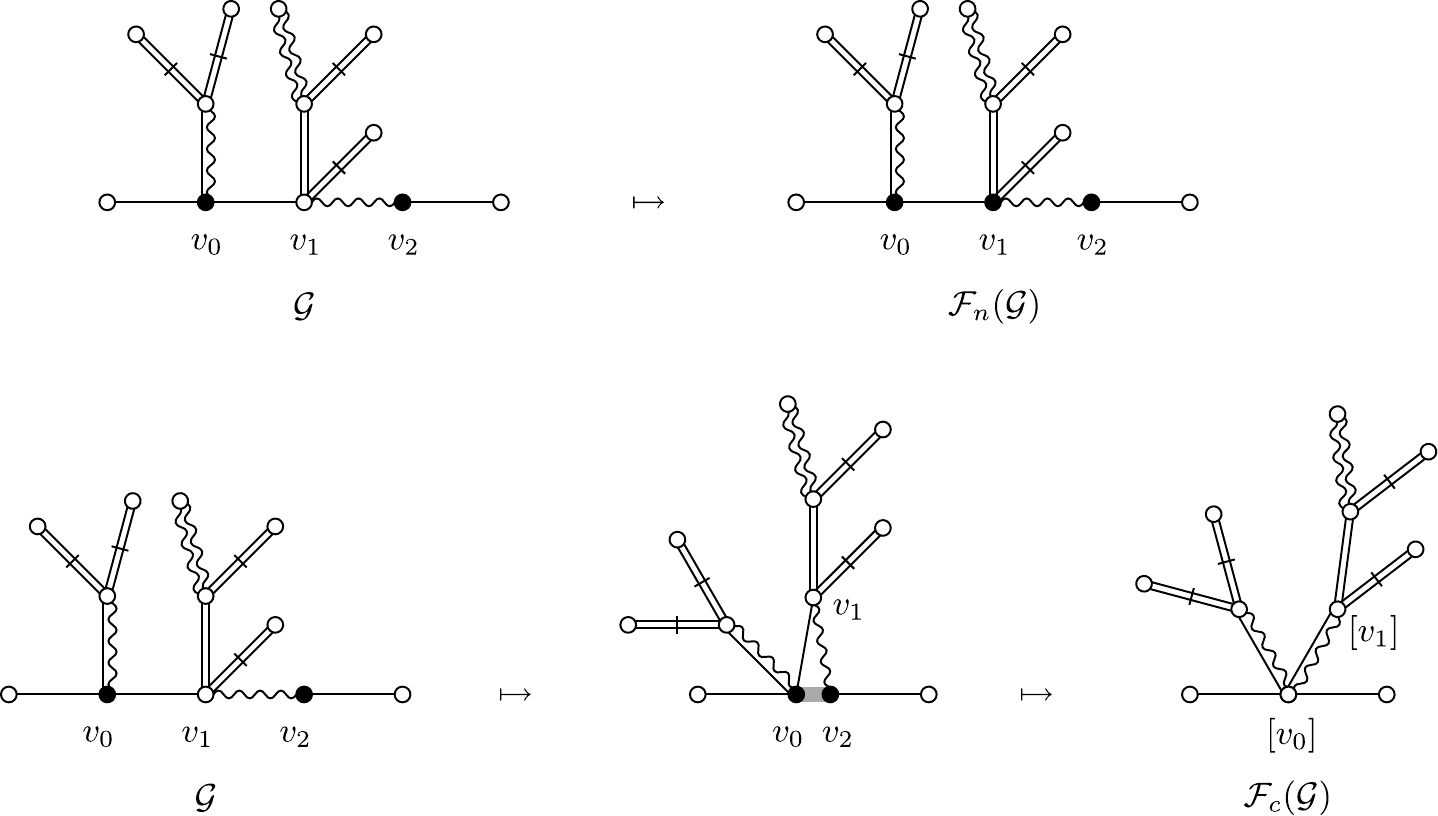}
\end{center}
\caption{Graphical representation of the operations $\cal F_n$ and $\cal F_c$. \label{figure: 
collapsing operation}}
\end{figure}

Finally, we define the map $l_{\wt G}$, which encodes the nonbacktracking information of  $\wt 
G$, through
\begin{equation*}
l_{\wt G}(\wt v, \wt w) \;\deq\;
\begin{cases}
1 & l_G(v,w) = 1 \text{ for some pair of representatives } v \in \wt v, w \in \wt w
\\
0 & \text{otherwise}\,.
\end{cases}
\end{equation*}
Thus, in the graphical representation of $\cal F_c(\cal G) \deq (\wt G, \tau_{\wt G}, l_{\wt 
G})$, the vertex $[v_0] = [v_2]$ is always white (i.e.\ backtracking). Note that if $v_0$ or
$v_2$ was nonbacktracking, this restriction remains encoded in the map $l_{\wt G}$, but is no 
longer visible in the colouring of the vertices.

We summarize the key properties of $\cal F_n$ and $\cal F_c$, which follow immediately from 
their construction.
\begin{lemma} \label{lemma: simple properties of graph operations}
Let $\cal G \in \fra G$. Then $\cal F_n(\cal G), \cal F_c(\cal G) \in \fra G$.  Moreover,
\begin{equation*}
\fra V_{xy}(\cal G) \;=\; \fra V_{xy}(\cal F_n(\cal G)) + \fra V_{xy}(\cal F_c(\cal G))\,,
\end{equation*}
and
\begin{equation*}
\deg(\cal F_n(\cal G)) \;=\; \deg(\cal F_c(\cal G)) \;=\; \deg(\cal G)\,.
\end{equation*}
\end{lemma}

\subsection{Graphs with completely nonbacktracking stem}
Next, we introduce two special subsets of decorated graphs. We define $\fra G' \subset \fra G$ 
to be the set of decorated graphs corresponding to terms in \eqref{main path expansion}. See 
Figure \ref{figure: example path} for an example. More precisely:

\begin{definition} \label{definition of G prime}
The set $\fra G'$ is the subset of $(G, \tau_G,l_G) \in \fra G$ satisfying
\begin{enumerate}
\item
All boughs of $G$ contain only one edge, whose tag is $(b,1)$;
\item
$l_G(u,w) = 1$ if and only if there is a vertex $v$ that is not the root of a bough, such that 
$(u,v), (v,w) \in \cal E(\cal S(G))$ with $\tau_G\pb{(v,w)} = (s,0)$.
\end{enumerate}
\end{definition}
Property (ii) says that all bough vertices (including the bough roots) are white, and that the 
left vertex of a wiggly edge is white. The remaining vertices (apart from $a(G)$ and $b(G)$) 
are black. It is easy to see that the graphs associated with terms on the right-hand side of 
\eqref{main path expansion} belong to $\fra G'$

Note that, unlike in the case of a general graph $\cal G \in \fra G$, the nonbacktracking 
information of a graph $\cal G \in \fra G'$ is fully encoded in the colouring of its vertices.  
Indeed, $l_G(v,w)$ can only be $1$ if $v,w \in \cal V(\cal S(G))$. Moreover, from (ii) we see 
that $l_G$ is uniquely determined by $G$ and $\tau_G$. Thus, a decorated graph $\cal G = (G, \tau_G, l_G) \in \fra G'$ 
is uniquely determined by its graph and tagging, i.e.\ the pair $(G, \tau_G)$.

The second important subset of decorated graphs is generated from $\fra G'$ by applying the 
operations $\cal F_n, \cal F_c$ to decorated graphs in $\fra G'$ until the stem is completely 
nonbacktracking, i.e.\ all stem vertices (apart from $a(G)$ and $b(G)$) are black.
\begin{definition} \label{definition of B_G}
For $\cal G \in \fra G'$ we define $\scr B_{\cal G}$ as the set of decorated graphs $\wt {\cal 
G} \in \fra G$ whose stem is completely nonbacktracking and that are obtained from $\cal G$ by 
a finite number of operations $\cal F_n$ and $\cal F_c$. Furthermore we set
\begin{equation*}
\fra G_\sharp \;\deq\; \bigcup_{\cal G \in \fra G'} \scr B_{\cal G}\,.
\end{equation*}
\end{definition}
The set $\fra G_\sharp$ is the set of ``good'' graphs that we shall work with in later 
sections. Thus, given a graph $\cal G \in \fra G'$ corresponding to a summand of \eqref{main 
path expansion}, we first transform it into the family $\scr B_{\cal G}$ of graphs in $\fra 
G_\sharp$. The contribution of $\cal G$ to the expansion \eqref{main path expansion} is given 
by the sum of the contributions of all graphs in $\scr B_{\cal G}$ (see \eqref{splitting of 
graph into nonbacktracking graphs} below). We then exploit the fact that we have good estimates 
on the contributions of graphs with completely nonbacktracking stems.

Next, we state and prove the key properties of the set $\fra G_\sharp$ and the operations $\cal 
F_n$ and $\cal F_c$.

\begin{proposition} \label{proposition: properties of graphs}
\begin{enumerate}
\item
If $\cal G = (G, \tau_G, l_G) \in \fra G_\sharp$ then $l_G$ is uniquely determined by the pair 
$(G, \tau_G)$ alone. In other words, there is a function $\ell$ such that $l_G = \ell (G, 
\tau_G)$ for all $(G, \tau_G, l_G) \in \fra G_\sharp$.
\item
If $(G, \tau_G, l_G) \in \fra G_\sharp$ then all leaves of $G$ are small (in $\tau_G$).
\item
If $(G, \tau_G, l_G) \in \fra G_\sharp$ and $e \in \cal E(G)$ has tag $\tau_G(e) = (b,1)$, then $e$ 
is a leaf of $G$.
\item
If $\cal G \neq \cal G' \in \fra G'$ then $\scr B_{\cal G} \cap \scr B_{\cal G'} = \emptyset$.
\item
For any $\cal G \in \fra G'$ and $\wt {\cal G} \in \scr B_{\cal G}$ we have $\deg (\cal G) = 
\deg(\wt {\cal G})$.
\item
For each $\cal G \in \fra G'$ we have
\begin{equation} \label{splitting of graph into nonbacktracking graphs}
\fra V_{xy}(\cal G) \;=\; \sum_{\wt {\cal G} \in \scr B_{\cal G}} \fra V_{xy}(\wt{\cal G})\,.
\end{equation}
\end{enumerate}
\end{proposition}

\begin{proof}
The key ingredient of the proof is the following \emph{ripping operation}, denoted by $\cal R$.  
It provides a link between the sets $\fra G_\sharp$ and $\fra G'$, and is essentially the 
converse of multiple applications of $\cal F_n$ and $\cal F_c$. The idea is to take hold of the 
vertices $a(G)$ and $b(G)$ of a given tagged graph $(G, \tau_G)$ and ``pull them apart'', thus 
``ripping open'' all bough edges of $\cal G$ except those of type $(b,1)$. When interpreted 
graphically, the character of each edge (straight or wiggly) is kept unchanged, whereby the 
double edge of a bough edge is split into two single edges.

When defining $\cal R$ it is convenient, in a first step, to ``rip open'' all bough edges 
(including those of type $(b,1)$) of $(G, \tau_G)$; we shall call the resulting tagged graph 
$\cal P(G, \tau_G)$.  In a second step, we undo the ripping of all bough edges of type $(b,1)$, 
which results in the tagged graph $\cal R(G, \tau_G)$.

In order to define $\cal P$, we need one additional tag $(s,2)$ for stem edges, which we draw 
with a single solid line that is slashed. Stem edges of type $(s,2)$ result from the ripping 
open of a bough edge of type $(b,1)$. By walking around $G$, we associate with the tagged graph 
$(G, \tau_G)$ a tagged bare stem $(\wt G, \tau_{\wt G}) \eqd \cal P(G, \tau_G)$.  More 
precisely, we draw $G$ in the plane, and start at the vertex $a(G)$.  At each step, we move 
along one edge of $G$ in such a way that we always remain to the left of $G$; see Figure 
\ref{figure: walk}.  Every stem edge is travelled once, and every bough edge twice. Each time 
we move along an edge $e \in \cal E(G)$, we add an edge $\wt e$ to the stem $\wt G$.  Depending 
on whether we moved along $e$ in the direction of $e$ (denoted by $+$) or against the direction 
of $e$ (denoted by $-$), we associate a tag $\tau_{\wt G}(\wt e)$ with $\wt e$ according to the 
following table.
\begin{center}
\begin{tabular}{cc|c}
$\tau_G(e)$ & direction & $\tau_{\wt G}(\wt e)$
\\[0.3em]
\hline
 & &
\\[-0.7em]
$(s, 0)$ & $+$ & $(s,0)$
\\[0.3em]
$(s, 1)$ & $+$ & $(s,1)$
\\[0.3em]
$(b,0)$ & $\pm$ & $(s,0)$
\\[0.3em]
$(b,1)$ & $\pm$ & $(s,2)$
\\[0.3em]
$(b,2)$ & $+$ & $(s,0)$
\\[0.3em]
$(b,2)$ & $-$ & $(s,1)$
\\[0.3em]
$(b,3)$ & $+$ & $(s,1)$
\\[0.3em]
$(b,3)$ & $-$ & $(s,0)$
\\[0.3em]
$(b,4)$ & $\pm$ & $(s,1)$
\end{tabular}
\hspace*{2cm}
Graphical representation of $(s,2)$: \quad \includegraphics{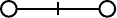}
\end{center}
\begin{figure}[ht!]
\begin{center}
\includegraphics{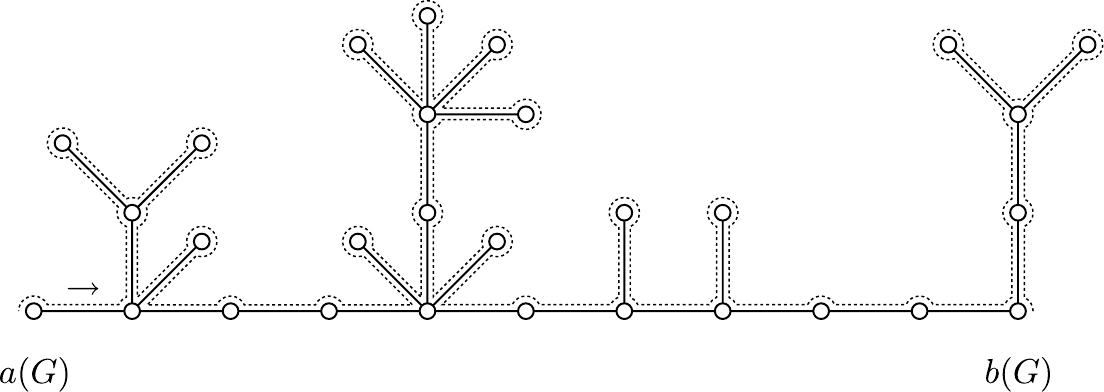}
\end{center}
\caption{The walk around $G$. \label{figure: walk}}
\end{figure}
These rules are made obvious by a glance at Figure \ref{figure: edge decorations}; indeed, a
tagged bough edge is represented with a double line which corresponds exactly to the two single 
lines resulting from ripping the bough edge open.
Figure \ref{figure: completely ripped graph} provides an example of the operation $(G, \tau_G)
\mapsto \cal P (G, \tau_G)$. The map $\cal P$ can also be interpreted as first doubling all 
bough edges according to their tags, and ripping them open successively by pulling the edges 
$a(G)$ and $b(G)$ apart; see Figure \ref{figure: progressive ripping}.

\begin{figure}[ht!]
\begin{center}
\includegraphics{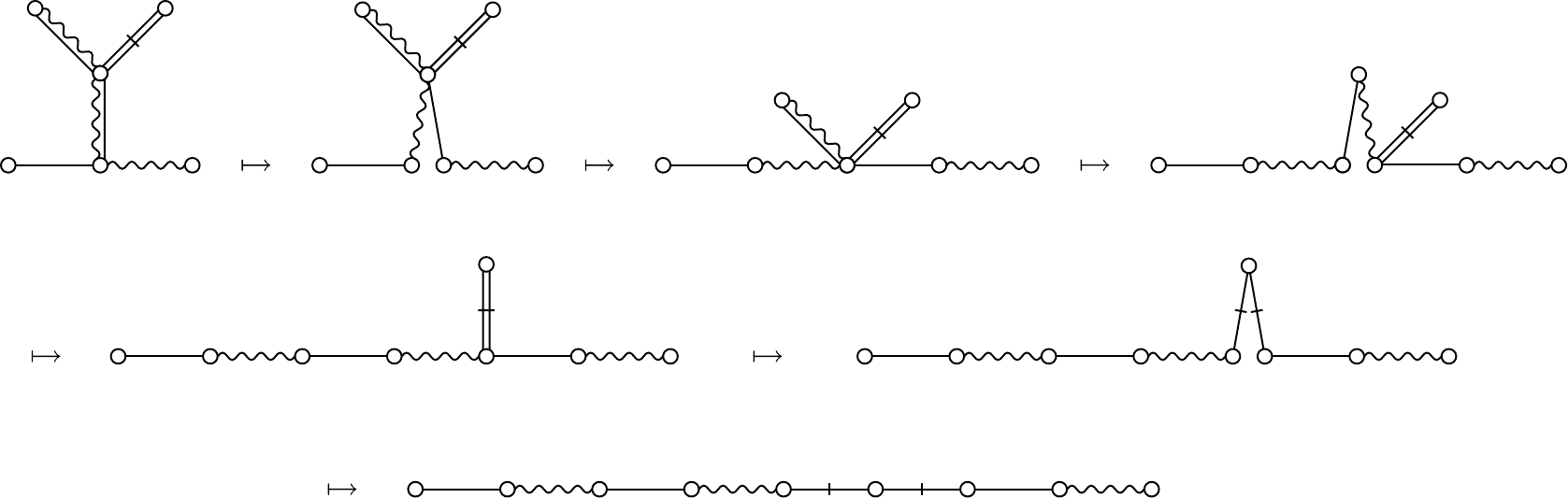}
\end{center}
\caption{The dynamical process of successively ripping open doubled bough edges. Note that the
stem edges with tag $(s,2)$ (drawn with a slashed single line) always occur in consecutive 
pairs.  \label{figure: progressive ripping}}
\end{figure}

We now define $\cal R(G, \tau_G) = \cal G'$ to be the unique decorated graph $\cal G' = (G', 
\tau_{G'}, l_{G'}) \in \fra G'$ that satisfies $\cal P(G, \tau_G) = \cal P(G', \tau_{G'})$; see 
Figure \ref{figure: completely ripped graph}. That there is exactly one such $\cal G' \in \fra 
G'$ follows immediately from the definitions of $\cal P$ and $\fra G'$, as well as the fact 
that $l_{G'}$ is uniquely determined by the pair $(G', \tau_{G'})$ through Definition 
\ref{definition of G prime} (ii). (Thus, the operation $\cal P$ plays only an auxiliary role, 
its sole purpose being to clarify the definition of $\cal R$.)

\begin{figure}[ht!]
\begin{center}
\includegraphics{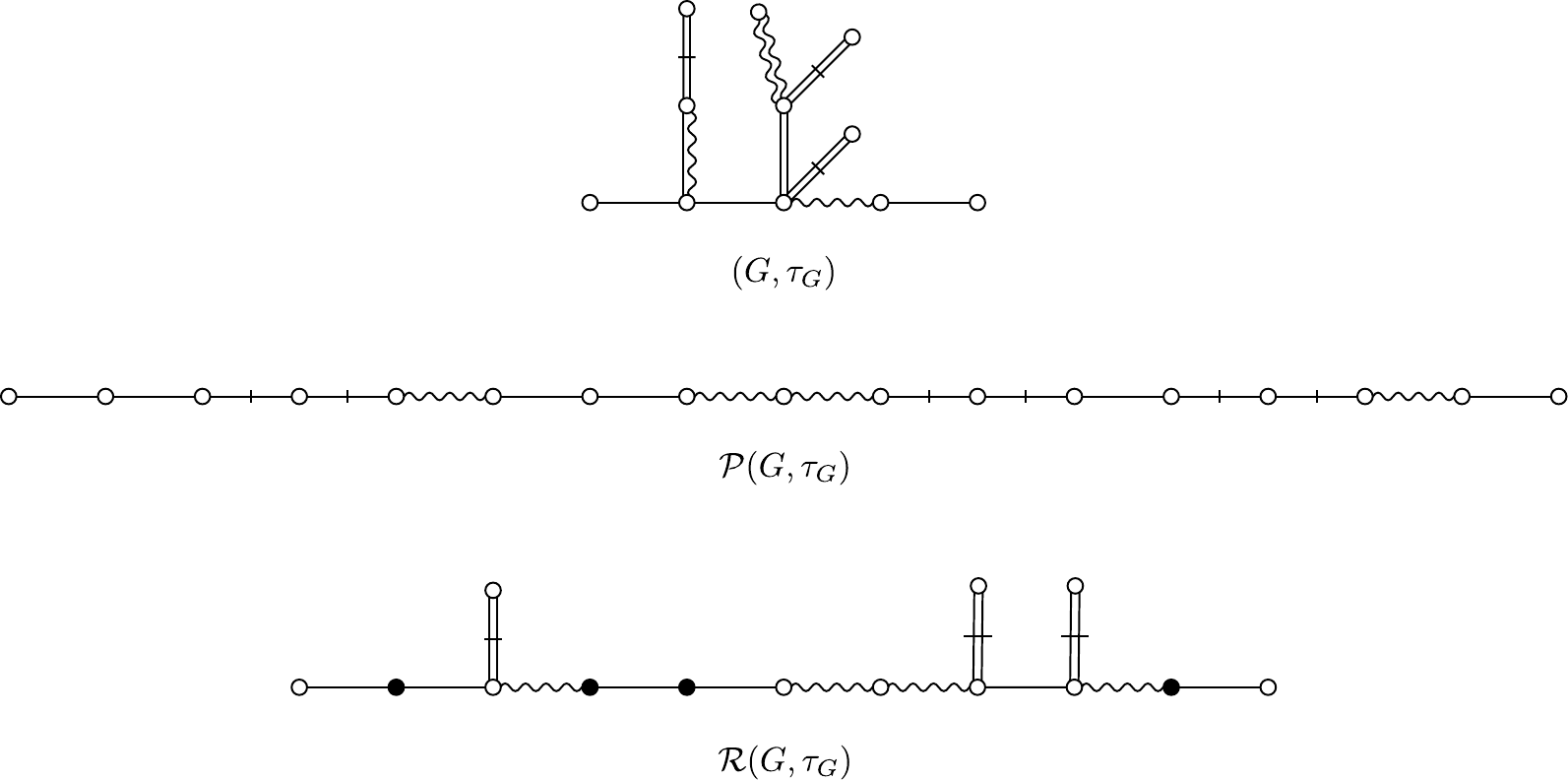}
\end{center}
\caption{The definitions of $\cal P$ and $\cal R$. \label{figure: completely ripped graph}}
\end{figure}

Having defined the ripping operation $\cal R$, we are now ready to prove Claim (i) of the 
Proposition. Before giving the full proof we outline the strategy. First, for any $\cal G  = 
(G, \tau_G, l_G) \in \fra G_\sharp$ we construct the ripped graph $\cal R(G, \tau_G) \in \fra 
G'$, which does not depend on $l_G$. Second, by definition of $\fra G'$, the ripped graph $\cal 
G' = \cal R(G, \tau_G)$ bears a unique nonbacktracing map $l_{G'}$. Third, by definition of 
$\fra G_\sharp$, there is a sequence $i_1, \dots, i_k \in \{n,c\}$ such that $\cal G = (G, 
\tau_G, l_G) = (\cal F_{i_k} \circ \cdots \circ \cal F_{i_1})(\cal R(G, \tau_G))$.  Fourth, we 
prove that this representation is unique. Thus we have expressed $l_G$ as a function of $(G, 
\tau_G)$.

Now to the proof of (i). Let $\cal G = (G, \tau_G, l_G) \in \fra G_\sharp$. By definition of 
$\fra G_\sharp$, there is a decorated graph $\cal G' = (G', \tau_{G'}, l_{G'}) \in \fra G'$ and 
a finite sequence $i_1, \dots, i_k \in \{n,c\}$ such that \begin{equation} 
\label{representation in terms of graphs operations}
\cal G \;=\; \pb{\cal F_{i_k} \circ \cdots \circ \cal F_{i_1}}(\cal G')\,.
\end{equation}
We now claim that both $\cal G'$ and the sequence $i_1, \dots, i_k$ are uniquely determined by 
$(G, \tau_G)$ (under the obvious constraint that no $\cal F_{i_j}$ is allowed to act on a 
decorated graph whose stem is completely nonbacktracking).
Indeed, we must have that $\cal G' = \cal R(G, \tau_G)$. (This follows immediately from the 
fact that $\cal R$ is left invariant under the action of $\cal F_{i}$, $i \in \{n,c\}$; i.e.\ 
$\cal R(G_2, \tau_{G_2}) = \cal R(G_1, \tau_{G_1})$ for any $(G_2, \tau_{G_2}, l_{G_2}) = \cal 
F_i(G_1, \tau_{G_1}, l_{G_1})$ where $(G_1, \tau_{G_1}, l_{G_1}) \in \fra G$.)

That different (under the above constraint) sequences applied to $\cal R(G, \tau_G)$ yield a 
different tagged graph is an immediate consequence of the following general claim.  In order to 
state it, we introduce the set $\scr D_{\wt {\cal G}}$ as the set of decorated graphs obtained 
from $\wt {\cal G} \in \fra G$ by a arbitrary applications of the operations $\cal F_n, \cal 
F_c$. (The set $\scr D_{\wt{\cal G}}$ will be used in the statement and the proof of the
following Claim $(*)$. We remark that the previously defined set $\scr B_{\wt{\cal G}}$ is a subset  of $\scr 
D_{\wt{\cal G}}$ with the additional requirement that the stem is black.)

\begin{itemize}
\item[$(*)$]
Let $\wt {\cal G} =  (\wt G, \tau_{\wt G}, l_{\wt G}) \in \fra G$ be an arbitrary decorated 
graph whose stem is not completely nonbacktracking. Then for any $\cal G_1 = (G_1, \tau_{G_1}, 
l_{G_1}) \in \scr D_{\cal F_n (\wt {\cal G})}$ and $\cal G_2 = (G_2, \tau_{G_2}, l_{G_2}) \in 
\scr D_{\cal F_c (\wt {\cal G})}$ we have $G_1 \neq G_2$.
\end{itemize}

Claim $(*)$ will be used in the following situation. We shall apply
sequences of operations $\cal F_n$ and $\cal F_c$ to a decorated graph $ \cal G' \in \fra G'$.  If two
sequences of such operations differ from each other in at least one step, then the resulting two graphs will be 
different.
In other words, if a decorated graph $\cal G$ can be written in the form \eqref{representation in terms of graphs 
operations}
and $\cal G'$ is known, then the sequence $i_1, i_2,\dots,i_k$ is uniquely determined. Together with the uniqueness of 
$\cal G'= \cal R(G,\tau_G)$
established earlier, this proves the uniqueness of the representation
\eqref{representation in terms of graphs operations}. Thus we can define the map $\ell: (G, \tau_G) \to l_G$ through
\begin{equation*}
\cal G \;=\; (G, \tau_G, l_G) \;=\; \pb{\cal F_{i_k} \circ \cdots \circ \cal F_{i_1}}\pb{\cal R(G, \tau_G)}\,,
\end{equation*}
and hence Claim (i) follows.

We now prove Claim $(*)$. For any graph $\wt G \in \fra W$ and integer $q \in Q_{\wt G} \deq 
\hb{0,1, \dots, \abs{\cal E (\cal S(\wt G))} + 2 \abs{\cal E(\cal B(\wt G))}}$, we define the 
vertex $v_{\wt G}(q) \in \cal V(\wt G)$ as the vertex reached after $q$ steps of the walk around $\wt G$ (see Figure 
\ref{figure: walk}).  For $q \in Q_{\wt G}$ we define the ``time of next return'' $r_{\wt G}(q)$ as the smallest integer 
$q' > q$ such that $v_{\wt G}(q') = v_{\wt G}(q)$; if there is no such $q'$, we set $q' \deq \infty$.

Next, let $v_1 \in \cal V(\wt G) \setminus \{a(\wt G), b(\wt G)\}$ be the first backtracking 
stem vertex of $\wt G$, and denote by $v_0$ its parent vertex (for an example see Figure 
\ref{figure: collapsing operation}). Define $q_0$ as the ``last time we walk across $v_0$'', 
i.e.\ as the largest integer in $Q_{\wt G}$ satisfying $v_{\wt G}(q_0) = v_0$. By definition of 
$q_0$, we have $r_{\wt G}(q_0) = \infty$.  Now define $\cal G_c = (G_c, \tau_{G_c}, l_{G_c}) 
\deq \cal F_{c}(\wt {\cal G})$. Clearly, we have that $r_{G_c}(q_0) < \infty$. Moreover, one 
readily sees that
\begin{equation} \label{difference in walk distance}
r_{G_2}(q_0) \;=\; r_{G_c}(q_0) \;<\; r_{G_1}(q_0)\,,
\end{equation}
for all $\cal G_1 = (G_1, \tau_{G_1}, l_{G_1}) \in \scr D_{\cal F_n (\wt {\cal G})}$ and $\cal 
G_2 = (G_2, \tau_{G_2}, l_{G_2}) \in \scr D_{\cal F_c (\wt {\cal G})}$. The equality expresses 
the fact that $v_0$ and $v_2$ have already been collapsed into one vertex in all $\cal G_2 \in 
\scr D_{\cal F_c(\wt{\cal G})}$. The inequality expresses the fact that, while $v_0$ may be 
collapsed with a stem vertex $v_j$ at some point when constructing $G_1 \in \scr D_{\cal 
F_n(\wt{\cal G})}$, the walk from $v_0$ to $v_j$ is strictly longer than from $v_0$ to $v_2$.  
Claim $(*)$ follows immediately from \eqref{difference in walk distance}.

Next, we prove Claim (ii). If $\cal G \in \fra G'$ then by definition all leaf edges have tag 
$(b,1)$, i.e.\ are small. This also holds for $\cal F_n(\cal G)$ (trivially), as well as for 
$\cal F_c(\cal G)$. In order to see this, define the property $(\mathrm P_{\cal G})$ as 
follows.
\begin{itemize}
\item[($\mathrm{P}_{\cal G}$)]
If a vertex $v \in \cal V(G)$ that is not the root of a bough satisfies $(u,v), (v,w) \in \cal 
E(\cal S(G))$ for some vertices $v,w$ and if the tags of $(u,v)$ and $(v,w)$ are both $(s,0)$, 
then the vertex $v$ is a nonbacktracking stem vertex.
\end{itemize}
Property $(\mathrm{P}_{\cal G})$ for a decorated graph $\cal G$ means that a vertex between two 
straight stem edges is black unless it is the root of a bough.
It is easy to see that the property $(\mathrm P_{\cal G})$ satisfied for all $\cal G \in \fra 
G'$ (see Definition \ref{definition of G prime} (ii)). Moreover, $(\mathrm P_{\cal G})$ is 
invariant under $\cal F_c$ and $\cal F_n$.  Recalling the definition of $\fra G_\sharp$, we see 
that Claim (ii) follows by induction.

Next, Claim (iii) clearly holds if $\cal G = (G, \tau_G, l_G) \in \fra G'$. Moreover, by 
definition of $\cal F_n$ and $\cal F_c$, Claim (iii) holds for $\cal F_n(\cal G)$ and $\cal 
F_c(\cal G)$ if it holds for $\cal G$. Hence Claim (iii) follows from the definition of $\fra 
G_\sharp$.

Claim (iv) is an immediate consequence of the fact that if $\wt {\cal G} \in \scr B_{\cal G}$ 
then $\cal G = \cal R(\wt{\cal G})$.

Claim (v) is an immediate consequence of the fact that, by definition of $\cal F_n$ and $\cal 
F_c$, we have $\deg (\cal G) = \deg(\cal F_n(\cal G)) = \deg(\cal F_c(\cal G))$.

Finally, we prove Claim (vi). Let $\cal G \in \fra G'$. Using Lemma \ref{lemma: simple 
properties of graph operations} repeatedly, we get
\begin{align*}
\fra V_{xy}(\cal G) &\;=\; \fra V_{xy}\pb{\cal F_n(\cal G)} + \fra V_{xy}\pb{\cal F_c(\cal G)}
\\
&\;=\; \fra V_{xy}\pb{\cal F_n(\cal F_n(\cal G))} + \fra V_{xy}\pb{\cal F_c(\cal F_n(\cal G))}
+ \fra V_{xy}\pb{\cal F_n(\cal F_c (\cal G))} + \fra V_{xy}\pb{\cal F_c(\cal F_c(\cal G))}
\\
&\;=\; \dots \;=\; \sum_{i \in I} \fra V_{xy}(\cal G_i)\,,
\end{align*}
where $(\cal G_i)_{i \in I}$ is a finite family of decorated graphs whose stems are completely 
nonbacktracking.  By definition of $\scr B_{\cal G}$, we have $\scr B_{\cal G} \;=\; \h{\cal 
G_i \,:\, i \in I}$.  What remains is to show that each $\wt{\cal G} \in \scr B_{\cal G}$ 
appears only once in $(\cal G_i)_{i \in I}$. But this is an immediate consequence of the 
uniqueness of the sequence $i_1, \dots, i_k$ in the representation $\wt {\cal G} = \pb{\cal 
F_{i_k} \circ \cdots \circ \cal F_{i_1}}(\cal G)$; see the proof of Claim (i) above.
\end{proof}

In view of Proposition \ref{proposition: properties of graphs} (i), we may regard the set $\fra 
G_\sharp$ as a set of tagged graphs $(G, \tau_G)$. We shall consistently adopt this point of 
view from now on.

\begin{proposition} \label{proposition: U_n as a sum over graphs}
We have
\begin{equation} \label{D_n as a sum over graphs}
(U_n)_{xy} \;=\; \sum_{\cal G \in \fra G' \,:\, \deg(\cal G) = n} \fra V_{xy}(G) \;=\; 
\sum_{\cal G \in \fra G_n} \fra V_{xy}(\cal G)\,,
\end{equation}
where $\fra V_{xy}(\cal G)$ is defined in \eqref{definition of value of graph}, and we defined the subset of graphs
\begin{equation*}
\fra G_n \;\deq\; \hb{\cal G \in \fra G_\sharp \,:\, \deg(\cal G) = n}\,.
\end{equation*}
\end{proposition}
\begin{proof}
The first equality of \eqref{D_n as a sum over graphs} follows from \eqref{main path expansion} and the definition of 
$\fra G'$ (see Definition \ref{definition of G prime}); the second from Proposition \ref{proposition: properties of 
graphs} (iv), (v), and (vi).
\end{proof}

\section{Lumping of edges} \label{section: lumping}
Recall that our aim is to compute
\begin{equation*}
\wh \varrho_{\leq}(t,x) \;=\; \sum_{n + n' \leq M^\mu}  \alpha_{n}(t) \, \ol{\alpha_{n'}(t)} \; \E 
\, (U_n)_{0x} (U_{n'})_{x0}\,.
\end{equation*}
By \eqref{D_n as a sum over graphs} we have
\begin{equation} \label{main term as a sum over graphs}
\wh \varrho_{\leq}(t,x) \;=\; \sum_{n + n' \leq M^\mu}  \alpha_{n}(t) \, \ol{\alpha_{n'}(t)} \; 
\sum_{\cal G \in \fra G_n} \sum_{\cal G' \in \fra G_{n'}}  \E \fra V_{0x}(\cal G) \fra 
V_{x0}(\cal G')\,.
\end{equation}
Computing the expectation $\E \fra V_{0x}(G) \fra V_{x0}(G')$ yields a lumping of the edges 
$\cal E(G) \cup \cal E(G')$, which we now describe.

For the following we fix $\cal G = (G, \tau_G) \in \fra G_\sharp$ and $\cal G' = (G', \tau_{G'}) \in \fra G_\sharp$.  
Thus, we also fix the maps $l_G$ and $l_{G'}$; see Proposition \ref{proposition: properties of graphs} (i).  To 
streamline notation, we introduce their union $\cal G \cup \cal G' = (G \cup G', \tau_{G \cup G'})$ defined in the 
obvious way. We also get the map $l_{G\cup G'}$ that we extend by requiring that $l_{G\cup G'}(v,w) = 0$ if $v \in \cal 
V(G)$ and $w \in \cal V(G')$.  We often abbreviate $\tau \equiv \tau_{G \cup G'}$ and $l \equiv l_{G \cup G'}$.

As in the previous section, we abbreviate the family of labels with
\begin{equation*}
\b x \;=\; \pb{x_v \,:\, v \in \cal V(G \cup G')}\,.
\end{equation*}
From \eqref{definition of value of graph} we immediately get
\begin{multline} \label{value of graph as product over edges}
\fra V_{0x}(\cal G) \fra V_{x0}(\cal G') \;=\; \sum_{\b x} \delta_{0 x_{a(G)}} \delta_{x 
x_{b(G)}} \delta_{x x_{a(G')}} \delta_{0 x_{b(G')}}
\qBB{\prod_{v,w \in \cal V(G \cup G')} \pb{1 - l(v,w) \delta_{x_v x_w}}}
\\
\times
\prod_{e \in \cal E(G \cup G')} P_{\tau(e)}\pb{\wh H_{x_{a(e)} x_{b(e)}}, \wh H_{x_{b(e)} x_{a(e)}}}\,.
\end{multline}
Next, for any fixed $\b x$ we assign to each edge $e \in \cal E(G \cup G')$ the unordered pair 
of labels
\begin{equation*}
\varrho_{\b x}(e) \;\deq\; \{x_{a(e)}, x_{b(e)}\}\,.
\end{equation*}
To each label configuration $\b x$ we assign a \emph{lumping} $\Gamma(\b x)$ of the edges $\cal E(G \cup G')$ according 
to the value of the map $\varrho_{\b x}$. We use the word lumping to mean an equivalence relation on $\cal E(G \cup 
G')$, or, equivalently, a partition of $\cal E(G \cup G')$. More precisely, the lumping $\Gamma(\b x)$ is defined as the 
equivalence relation (denoted by $\sim$)
on $\cal E(G \cup G')$ such that $e \sim e'$ if and only if $\varrho_{\b x}(e) = \varrho_{\b x}(e')$.  
We use the notation $\Gamma = \{\gamma\}_{\gamma \in \Gamma}$, where $\gamma \subset \cal E(G 
\cup G')$ is a \emph{lump}, i.e.\ an equivalence class.
Thus, taking the expectation in \eqref{value of graph as product over edges} yields
\begin{multline} \label{expectation factors}
\E \fra V_{0x}(\cal G) \fra V_{x0}(\cal G') \;=\; \sum_{\b x} \delta_{0 x_{a(G)}} \delta_{x 
x_{b(G)}} \delta_{x x_{a(G')}} \delta_{0 x_{b(G')}}
\qBB{\prod_{v,w \in \cal V(G \cup G')} \pb{1 - l(v,w) \delta_{x_v x_w}}}
\\
\times
\prod_{\gamma \in \Gamma(\b x)} \E \prod_{e \in \gamma} P_{\tau(e)}\pb{\wh H_{x_{a(e)} 
x_{b(e)}}, \wh H_{x_{b(e)} x_{a(e)}}}\,,
\end{multline}
where we used that $\wh H_{a(e) b(e)}$ and $\wh H_{a(e') b(e')}$ are independent if $\varrho_{\b 
x}(e) \neq \varrho_{\b x}(e')$.
\label{page: lumping}

Next, we define the indicator function
\begin{equation} \label{definition of Delta}
\Delta_\Gamma(\b x) \;\deq\; \ind{\Gamma(\b x) = \Gamma} \;=\; \qBB{\prod_{\gamma \neq \gamma' 
\in \Gamma} \prod_{e \in \gamma} \prod_{e' \in \gamma'} \indb{\varrho_{\b x}(e) \neq \varrho_{\b 
x}(e')}}
\qBB{\prod_{\gamma \in \Gamma} \prod_{e,e' \in \gamma} \indb{\varrho_{\b x}(e) = \varrho_{\b 
x}(e')}}\,,
\end{equation}
indicating that a labelling $\b x$ is compatible with the equivalence relation $\Gamma$, i.e.\ 
$\varrho_{\b x}(e) = \varrho_{\b x}(e')$ if and only if $e \sim e'$.

By definition, $P_{\tau(e)}$ is an even function whenever $\deg(e)$ is even and an odd function 
whenever $\deg(e)$ is odd. Moreover, the matrix elements of $H$ were truncated in such a way 
that the identity \eqref{condition on moments} remains valid for $\wh H$; see \eqref{truncation 
of matrix elements}. Thus, the expectation \eqref{expectation factors} vanishes unless all 
lumps $\gamma \in \Gamma(\b x)$ are of even degree, whereby the \emph{degree} of a lump 
$\gamma$ is defined as
\begin{equation*}
\deg(\gamma) \;\deq\; \sum_{e \in \gamma} \deg(e)\,.
\end{equation*}

Let $\scr G(G \cup G')$ denote the set of all lumpings of $\cal E(G \cup G')$ whose lumps are 
of even degree. 
Thus \eqref{expectation factors} becomes
\begin{equation} \label{expectation as a sum over lumpings}
\E \fra V_{0x}(\cal G) \fra V_{x0}(\cal G') \;=\; \sum_{\Gamma \in \scr G(G \cup G')} V_x(\cal 
G \cup \cal G', \Gamma)\,,
\end{equation}
where we defined the \emph{value} of the graph $\cal G \cup \cal G'$ with lumping $\Gamma$ as
\begin{multline} \label{value of general lumping}
V_x(\cal G \cup \cal G',\Gamma) \;\deq\; \sum_{\b x} \Delta_\Gamma(\b x)\, \delta_{0 x_{a(G)}} 
\delta_{x x_{b(G)}} \delta_{x x_{a(G')}} \delta_{0 x_{b(G')}}
\qBB{\prod_{v,w \in \cal V(G \cup G')} \pb{1 - l(v,w) \delta_{x_v x_w}}}
\\
\times
\prod_{\gamma \in \Gamma} \E \prod_{e \in \gamma} P_{\tau(e)}\pb{\wh H_{x_{a(e)} x_{b(e)}}, \wh 
H_{x_{b(e)} x_{a(e)}}}\,.
\end{multline}

\label{page: I_n}
Next, let $I_n \in \fra W$ denote the bare stem consisting of $n$ edges. Recall that a bare 
stem is a graph with no bough edges; it is uniquely determined by its number of edges.
Denote by $\cal I_n \in \fra G_n$ the decorated graph obtained from $I_n$ by assigning the tag $(s,0)$ to each edge (in 
particular, the stem $I_n$ is completely nonbacktracking in $\cal I_n$).  Define the subset \label{page: G_n^*}
\begin{equation*}
\fra G_n^* \;\deq\; \fra G_n \setminus \{\cal I_n\}\,.
\end{equation*}
From \eqref{main term as a sum over graphs} and \eqref{expectation as a sum over lumpings} we get the splitting
\begin{multline} \label{rho split into bare and garnished parts}
\wh \varrho_{\leq}(t,x)
\\
\;=\; \sum_{n + n' \leq M^\mu}  \alpha_{n}(t) \, \ol{\alpha_{n'}(t)} \sum_{\Gamma \in \scr G(I_n \cup I_{n'})} V_x(\cal 
I_n \cup \cal I_{n'}, \Gamma)
+
\sum_{n + n' \leq M^\mu}  \alpha_{n}(t) \, \ol{\alpha_{n'}(t)} \sum_{\cal G \in \fra G_n^*} \sum_{\cal G' \in \fra 
G_{n'}^*} \sum_{\Gamma \in \scr G(G \cup G')} V_x(\cal G \cup \cal G', \Gamma)
\\
+
\sum_{n + n' \leq M^\mu}  \alpha_{n}(t) \, \ol{\alpha_{n'}(t)}  \sum_{\cal G' \in \fra G_{n'}^*}
\sum_{\Gamma \in \scr G(I_n \cup G')} V_x(\cal I_n \cup \cal G', \Gamma)
+
\sum_{n + n' \leq M^\mu}  \alpha_{n}(t) \, \ol{\alpha_{n'}(t)} \sum_{\cal G \in \fra G_n^*}
\sum_{\Gamma \in \scr G(G \cup I_{n'})} V_x(\cal G \cup \cal I_{n'}, \Gamma)\,.
\end{multline}

This is our starting point for the remaining sections. The first term on the right-hand side of 
\eqref{rho split into bare and garnished parts} is the leading term, whose contribution is 
computed in Section \ref{section: bare stem}. The remaining three terms on the right-hand side 
of \eqref{rho split into bare and garnished parts} are error terms, and are estimated in 
Sections \ref{section: boughs} and \ref{section: boughs for kappa=1/3}.

\section{The bare stem} \label{section: bare stem}
In this section we analyse the first term on the right-hand side of \eqref{rho split into bare 
and garnished parts} by proving the following result.
\begin{proposition} \label{proposition: value of bare stem}
For any continuous bounded function $\varphi \in C_b(\R^d)$ and $T \geq 0$ we have
\begin{multline} \label{statement of stem convergence}
\lim_{W \to \infty} \sum_x \varphi \pbb{\frac{x}{W^{1 + d \kappa / 2}}} \sum_{n + n' \leq W^\mu} \alpha_n(W^{d \kappa} 
T) \ol{\alpha_{n'}(W^{d \kappa} T)}
\; \sum_{\Gamma \in \scr G(I_n \cup I_{n'})} V_x(\cal I_n \cup \cal I_{n'}, \Gamma)
\\
=\; \int \dd X \; L(T,X) \, \varphi(X)\,,
\end{multline}
where we recall the definition of $ L(T,X)$ from \eqref{def:L}.
\end{proposition}
The rest of this section is devoted to the proof of Proposition \ref{proposition: value of bare 
stem}.
The proof is similar to \cite{erdosknowles}, which we shall frequently refer to in this section for precise definitions 
and proofs. We therefore assume that the reader has some familiarity with \cite{erdosknowles}.

The only complication compared to \cite{erdosknowles} is that controlling higher order lumpings 
(resulting in high moments of $\wh A_{xy}$) requires more effort, since, unlike in 
\cite{erdosknowles}, the matrix elements of $\wh A$ are not bounded by $1$ (but only by 
$M^\delta$). A lump $\gamma$ containing $\abs{\gamma}$ edges carries a weight $M^{\delta 
\abs{\gamma}}$, but this factor can be compensated by the fact that large lumps impose strong 
restrictions on the labelling of the vertices. Technically, we shall deal with these higher 
order lumpings by replacing an arbitrary lumping with a pairing whose contribution is small 
enough to compensate any powers of $M$ resulting from the lumping.  In this way we can directly 
reduce the estimate of general lumpings to pairings. The appropriate pairing will be selected 
by a greedy algorithm defined in Appendix \ref{appendix: higher lumpings}.

We begin by establishing notation and recalling the relevant results from \cite{erdosknowles}.

\subsection{Pairing of edges}
The simple structure of $\cal I_n \cup \cal I_{n'}$ allows for some notational simplifications.
Following \cite{erdosknowles}, we abbreviate $\scr G_{n,n'} \deq \scr G(I_n \cup I_{n'})$ and $V_x(\Gamma) \;\deq\; 
V_x(\cal I_n \cup \cal I_{n'}, \Gamma)$. Thus the left-hand side of \eqref{statement of stem convergence} becomes
\begin{equation} \label{sum of bare trunk}
\lim_{W \to \infty} \sum_x \varphi \pbb{\frac{x}{W^{1 + d \kappa / 2}}} \sum_{n + n' \leq W^\mu} \alpha_n(W^{d \kappa} 
T) \ol{\alpha_{n'}(W^{d \kappa} T)}
\; \sum_{\Gamma \in \scr G_{n,n'}} V_x(\Gamma).
\end{equation}
As in \cite{erdosknowles}, we identify the vertices $a(I_n)$ and $b(I_{n'})$, as well as the vertices $b(I_n)$ and 
$a(I_{n'})$ (this is purely a notational simplification). We label the vertices explicitly according to
\begin{equation*}
\cal V(I_n \cup I_{n'}) \;=\; \{0, \dots, n + n'-1\}\,, \qquad a(I_n) \;=\; b(I_{n'}) \;=\; 
0\,, \qquad b(I_n) \;=\; a(I_{n'}) \;=\; n\,,
\end{equation*}
and write $\b x = (x_0, \dots, x_{n+n' -1})$. See Figure \ref{figure: identifying vertices}.
Recall that the degree of every edge of $I_n \cup I_{n'}$ is odd. Since, by definition of $\scr 
G_{n,n'}$, every lump $\gamma \in \Gamma$ has even degree, we conclude that every lump $\gamma 
\in \Gamma$ has an even number of edges. (Note that no such statement is possible for a general 
lump that also contains bough edges. Indeed, bough edges have even degree, so that the total 
degree of the lump gives no information about the its number of edges.)

\begin{figure}[ht!]
\begin{center}
\includegraphics{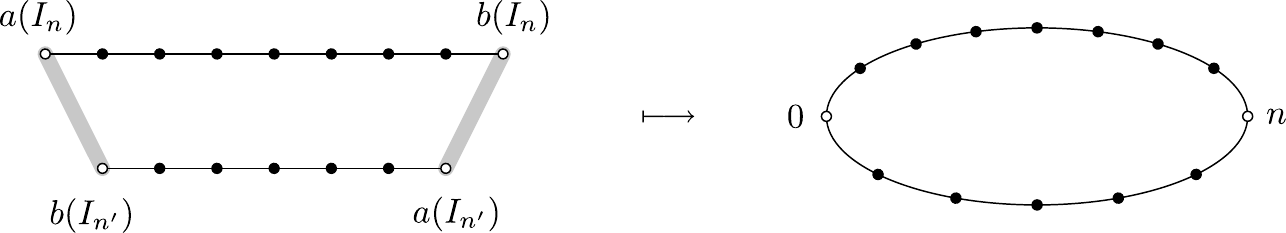}
\end{center}
\caption{Identifying the end vertices of $I_n$ and $I_{n'}$. \label{figure: identifying 
vertices}}
\end{figure}

The expression \eqref{value of general lumping} may also be simplified in the case of the bare 
stem. From \eqref{value of general lumping} we get
\begin{equation} \label{value of lumping for bare stem}
V_x(\Gamma) \;=\; \sum_{\b x} \Delta_\Gamma(\b x)\, Q_x(\b x)
\prod_{\gamma \in \Gamma} \E \prod_{e \in \gamma} \wh H_{x_{a(e)} x_{b(e)}}
\,,
\end{equation}
where we defined the indicator function
\begin{equation*}
Q_x(\b x) \;\deq\; \delta_{0 x_0} \delta_{x x_n} \qBB{\prod_{i = 0}^{n - 2} \ind{x_i \neq 
x_{i+2}}} \qBB{\prod_{i = n}^{n + n' - 2} \ind{x_i \neq x_{i + 2}}}\,.
\end{equation*}
Here we used that all edges of $\cal I_n \cup \cal I_{n'}$ have tag $(s,0)$.

Next, we make the obvious observation that, without loss of generality, we may exclude from 
$\scr G_{n,n'}$ all lumpings $\Gamma$ satisfying $\Delta_\Gamma(\b x) Q_x(\b x) = 0$ for all 
$x$ and $\b x$. In particular, if $\gamma \in \Gamma$ then $\gamma$ cannot contain two adjacent 
edges (since this would contradict the nonbacktracking condition in $Q$).

\label{page: pairing}
We call lumpings $\Gamma = \{\gamma\}$ with $\abs{\gamma} = 2$ for each $\gamma \in \Gamma$ 
\emph{pairings}, and denote the subset of pairings by $\scr P_{n,n'} \subset \scr G_{n,n'}$. We 
shall often use the notation $\Pi = \{\pi\}$ instead of $\Gamma = \{\gamma\}$ to denote a 
pairing. We represent a pair $\pi = \{e, e'\}$ graphically by drawing a line, called a 
\emph{bridge}, that joins the edges $e,e' \in \cal E(I_n \cup I_{n'})$; see Figure \ref{figure: 
pairing and ladder}.

We shall show that the leading order contribution to the left-hand side of \eqref{statement of 
stem convergence} comes from the pairings; all higher order lumpings are subleading. Moreover, 
only the contribution of the so-called \emph{ladder pairing} (see Subsection \ref{subsection: 
ladder} below) survives in the limit $W \to \infty$. In fact, only the ladder whose bridges all 
carry a \emph{straight tag} (see below for the definition of the tagging of bridges) yields a 
nonvanishing contribution to the left-hand side of \eqref{statement of stem convergence}.

\begin{figure}[ht!]
\begin{center}
\includegraphics{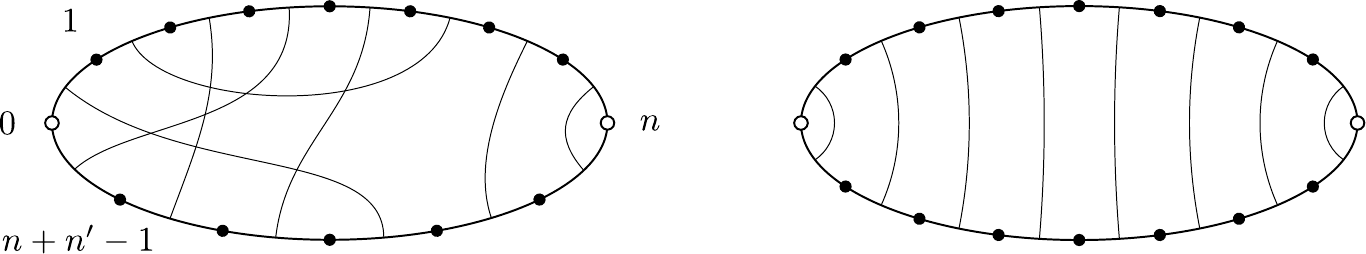}
\end{center}
\caption{A general pairing (left) and a ladder (right). \label{figure: pairing and ladder}}
\end{figure}

If $\Pi \in \scr P_{n,n'}$ is a pairing we get from \eqref{definition of Delta}
\begin{equation} \label{Delta for pairing}
\Delta_\Pi(\b x) \;=\; \qBB{\prod_{\pi \neq \pi'} \prod_{e \in \pi} \prod_{e' \in \pi'} 
\indb{\varrho_{\b x}(e) \neq \varrho_{\b x}(e')}}
\qBB{\prod_{\{e,e'\} \in \Pi} \indb{\varrho_{\b x}(e) = \varrho_{\b x}(e')}}\,.
\end{equation}
\label{page: tagged bridge}
At this point we stress that the indicator function $\indb{\varrho_{\b x}(e) = \varrho_{\b x}(e')}$ 
in \eqref{Delta for pairing} associated with the bridge $\{e,e'\}$ is different from its 
counterpart in \cite{erdosknowles} (Equation (6.3) in \cite{erdosknowles}), where bridges carry 
an orientation. In order to make the link to \cite{erdosknowles}, we \emph{tag}\footnote{To 
avoid confusion we emphasize that these bridge tags have nothing to do with the edge tags of a 
decorated graph. The use of the same word is merely a symptom of a regrettable lack of 
imagination on the authors' part.} bridges (similarly to Section 9 of \cite{erdosknowles}).  In 
other words, we choose a map $\vartheta : \Pi \to \{0,1\}$ and replace the factor 
$\indb{\varrho_{\b x}(e) = \varrho_{\b x}(e')}$ in \eqref{Delta for pairing} with $\Xi_{\b x}(\pi, 
\vartheta(\pi))$, where
\begin{align*}
\Xi_{\b x}(\{e, e'\}, 0) &\;\deq\; \indb{x_{a(e)} = x_{b(e')}} \indb{x_{b(e)} = x_{a(e')}}\,,
\\
\Xi_{\b x}(\{e, e'\}, 1) &\;\deq\; \indb{x_{a(e)} = x_{a(e')}} \indb{x_{b(e)} = x_{b(e')}} 
\indb{x_{a(e)} \neq x_{b(e)}}\,.
\end{align*}
\label{page: twisted and straight bridge}
We call a bridge $\pi$ \emph{straight} if $\vartheta(\pi) = 0$ and \emph{twisted} if 
$\vartheta(\pi) = 1$. See Figure \ref{figure: straight and twisted bridges}.
\begin{figure}[ht!]
\begin{center}
\includegraphics{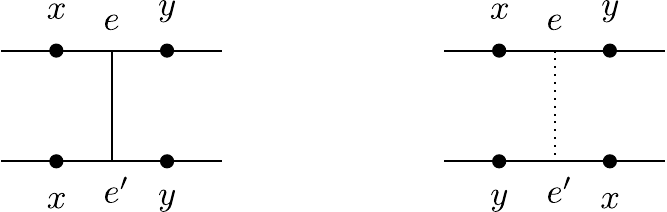}
\end{center}
\caption{A straight bridge (left, drawn with a solid line) and a twisted bridge (right, drawn with a dotted line) 
joining the edges $e$ and $e'$. In each case we indicate how the vertex labels of $x = a(e)$ and $y = b(e)$ determine 
the vertex labels of $a(e')$ and $b(e')$.  \label{figure: straight and twisted bridges}}
\end{figure}
Clearly, we have
\begin{equation} \label{splitting of untagged bridge}
\indb{\varrho_{\b x}(e) = \varrho_{\b x}(e')} \;=\; \Xi_{\b x}(\{e, e'\}, 0) + \Xi_{\b x}(\{e, e'\}, 
1)\,.
\end{equation}
Thus, each untagged bridge may be split into a straight and a twisted one. We define
\begin{equation} \label{tagged Delta}
\Delta_{\Pi, \vartheta}(\b x) \;\deq\; \qBB{\prod_{\pi \neq \pi'} \prod_{e \in \pi} \prod_{e' 
\in \pi'} \indb{\varrho_{\b x}(e) \neq \varrho_{\b x}(e')}}
\qBB{\prod_{\pi \in \Pi} \Xi_{\b x}(\pi, \vartheta(\pi))}\,,
\end{equation}
so that we have
\begin{equation} \label{tag splitting sums up to one}
\sum_{\vartheta \in \{0,1\}^\Pi} \Delta_{\Pi, \vartheta}(\b x) \;=\; \Delta_\Pi(\b x)\,.
\end{equation}
In this manner we may split
\begin{equation*}
V_x(\Pi) = \sum_{\vartheta \in \{0,1\}^\Pi} V_x(\Pi, \vartheta)\,,
\end{equation*}
where
\begin{equation} \label{value of tagged pairing for bare stem}
V_x(\Pi, \vartheta) \;\deq\; \sum_{\b x} \Delta_{\Pi, \vartheta}(\b x)\, Q_x(\b x)
\prod_{\pi \in \Pi} \E \prod_{e \in \pi} \wh H_{x_{a(e)} x_{b(e)}}
\,.
\end{equation}

\subsection{Parallel and antiparallel bridges} \label{page: parallel} \label{page: skeleton}
In \cite{erdosknowles}, the combinatorial complexity of a pairing was measured using the size 
of its \emph{skeleton pairing}. The definition of the skeleton pairing relies on the following 
notion of parallel and antiparallel bridges. We say that $\pi, \pi'$ are \emph{parallel} if 
there exist $i,j \notin \{0, n\}$ such that
\begin{equation*}
\pi \;=\; \hb{(i-1, i), (j,j+1)} \,, \qquad \pi' \;=\; \hb{(i,i+1), (j-1, j)}\,.
\end{equation*}
Similarly, $\pi, \pi'$ are \emph{antiparallel} if there exist $i,j \notin \{0, n\}$ such that
\begin{equation*}
\pi \;=\; \hb{(i-1, i), (j - 1,j)} \,, \qquad \pi' \;=\; \hb{(i,i+1), (j, j + 1)}\,.
\end{equation*}
Note that the notion (anti)parallel is independent of the bridge tags.
See Figure \ref{figure: parallel and antiparallel bridges}. A sequence of bridges $\pi_1, 
\dots, \pi_k$ is called an \emph{(anti)ladder} if $\pi_i$ and $\pi_{i+1}$ are (anti)parallel 
for all $i = 1, \dots, k - 1$.
\begin{figure}[ht!]
\begin{center}
\includegraphics{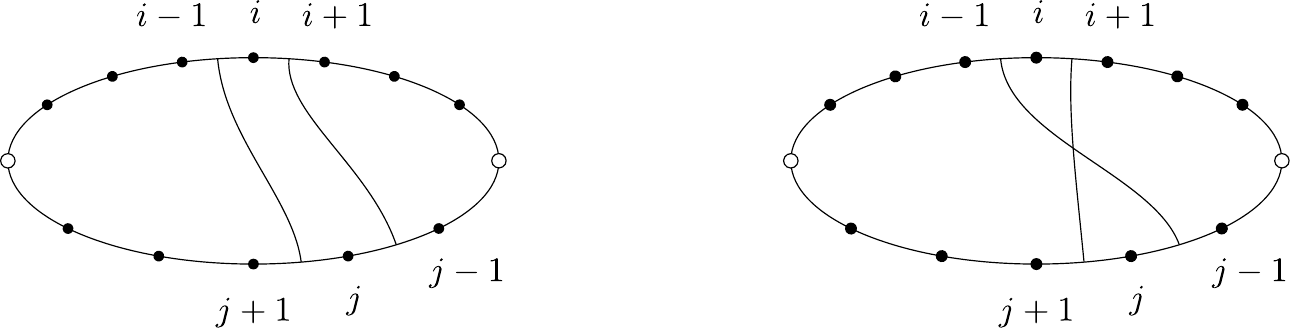}
\end{center}
\caption{Two parallel bridges (left) and two antiparallel bridges (right). \label{figure: 
parallel and antiparallel bridges}}
\end{figure}

Next, we assign to each tagged pairing $(\Pi, \vartheta)$ a skeleton $S(\Pi, \vartheta)$ 
according to the following rules. Every pair of parallel bridges that are both straight is 
replaced by a single straight bridge; every pair of antiparallel bridges that are both twisted 
is replaced by a single twisted bridge. (See \cite{erdosknowles}, Section 7.2, for a precise 
definition of this collapsing of bridges. Each collapsing step removes one bridge -- and hence 
two edges from $I_n \cup I_{n'}$ -- but always retains the vertices $a(I_n), b(I_n), a(I_{n'}), 
b(I_{n'})$.) We repeat this procedure until we reach a tagged pairing, denoted by $S(\Pi, 
\vartheta)$, which contains no parallel straight bridges and no antiparallel twisted bridges. 
The resulting skeleton is independent of the order in which pairs of bridges are collapsed.
We have that $S(\Pi, \vartheta) \in \scr P_{m,m'}$ for some $m \leq n$ and $m' \leq n'$. See 
Figure \ref{figure: creating the skeleton}, and \cite{erdosknowles}, Sections 7 and 9, for full 
details.
\begin{figure}[ht!]
\begin{center}
\includegraphics{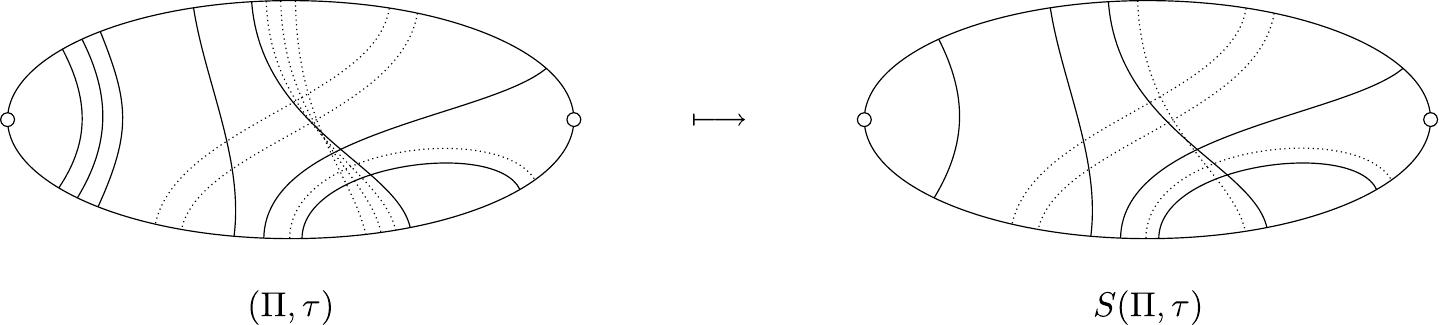}
\end{center}
\caption{A tagged pairing along with its tagged skeleton. We draw straight bridges with solid lines and twisted bridges 
with dotted lines. \label{figure: creating the skeleton}}
\end{figure}

\subsection{The ladder} \label{subsection: ladder} \label{page: ladder}
We now extract the leading order contribution to \eqref{sum of bare trunk}, the (complete) \emph{ladder}.  The ladder of 
degree $n$, denoted by $L_n$, is the pairing given by
\begin{equation}
L_n \;=\; \hB{\hb{(0,1), (2n-1, 0)}, \hb{(1,2), (2n - 2, 2n - 1)}, \dots, \hb{(n - 1, n),(n, 
n+1)}} \;\in\; \scr P_{n,n}\,,
\end{equation}
see Figure \ref{figure: pairing and ladder}. Set $\vartheta_n \equiv 0 \in \{0,1\}^{L_n}$; thus 
$(L_n, \vartheta_n)$ is the ladder whose bridges are all straight. Since all bridges of $(L_n, 
\vartheta_n)$ are straight, we find that the expectation in \eqref{value of tagged pairing for 
bare stem} is equal to $\E \absb{\wh H_{x_{a(e)} x_{b(e)}}}^2$. Now the argument of 
\cite{erdosknowles}, Section 8, applies almost verbatim, and, together with Lemma \ref{lemma: 
bounds on truncated variance}, we get
\begin{equation} \label{limit of ladders}
\lim_{W \to \infty} \sum_x \sum_{n = 0}^{M^\mu / 2} \abs{\alpha_n(W^{d \kappa} T)}^2 V_x(L_n, 
\vartheta_n) \varphi \pbb{\frac{x}{W^{1 + d \kappa / 2}}} = \int \dd X \; L(T,X) \, \varphi(X)
\end{equation}
for all $\varphi \in C_b(\R^d)$.
In fact, the only needed modification to the argument of \cite{erdosknowles}, Section 8, is 
that, in the proof of Lemma 8.4 of \cite{erdosknowles}, the i.i.d.\ random variables $(B_i)$ 
now have the law
\begin{equation} \label{law of steps}
\frac{1}{M} \frac{\sqrt{[W^{d \kappa} T]}}{\sqrt{W^{d \kappa} }} \sum_{a \in \Z^d} 
f\pbb{\frac{a}{W}} \, \frac{\delta_a}{W}
\end{equation}
instead of
\begin{equation*}
\frac{1}{M} \frac{\sqrt{[W^{d \kappa} T]}}{\sqrt{W^{d \kappa} }} \sum_{a \in \Z^d} \indb{1 \leq 
\abs{a} \leq W} \frac{\delta_a}{W}\,.
\end{equation*}
Here $[\cdot]$ denotes integer part and $\delta_a$ the point mass at $a$.
It is easy to see that the covariance matrix of the measure \eqref{law of steps} is $T \Sigma + 
o(1)$ as $W \to \infty$, where, we recall,
\begin{equation*}
\Sigma_{ij} \;=\; \int_{\R^d} \dd x \; f(x) \, x_i x_j\,.
\end{equation*}

\subsection{Bound on the non-pair lumps}
We now give a bound on the contribution of the higher-order lumpings, i.e.\ lumpings that 
contain lumps of size more than two. We start by assigning to each pairing $\Pi \in \scr 
P_{n,n'}$ its \emph{minimum skeleton size}
\begin{equation} \label{minimum skeleton size}
m(\Pi) \;\deq\; \min_{\vartheta \in \{0,1\}^\Pi} \pb{\text{number of bridges in }S(\Pi, 
\vartheta)}\,.
\end{equation}
The quantity $m(\Pi)$ is the correct measure of the combinatorial complexity of the pairing 
$\Pi$.

Let $\Gamma \in \scr G_{n,n'}$ be an arbitrary lumping and define
\begin{equation}
p(\Gamma) \;\deq\; \sum_{\gamma \in \Gamma} (\abs{\gamma} - 2)\,.
\end{equation}
We say that a lumping $\Gamma' \in \scr G_{n,n'}$ is a \emph{refinement} of a lumping $\Gamma 
\in \scr G_{n,n'}$ if for every $\gamma' \in \Gamma'$ there is a $\gamma \in \Gamma$ such that 
$\gamma' \subset \gamma$. If $\Pi \in \scr P_{n,n'}$ is a pairing that is a refinement of 
$\Gamma$, we say that $\Pi$ is a \emph{refining pairing} of $\Gamma$.

\begin{lemma} \label{lemma: existence of refining pairing}
For each $\Gamma \in \scr G_{n,n'} \setminus \scr P_{n,n'}$ there is a refining pairing $\Pi 
\in \scr P_{n,n'}$ of $\Gamma$ such that
\begin{equation} \label{lower bound on number of skeleton edges}
m(\Pi) \;\geq\; \max \pbb{\frac{p(\Gamma)}{4} \,,\, 2}\,.
\end{equation}
\end{lemma}

\begin{proof}
See Appendix \ref{appendix: higher lumpings}.
\end{proof}

Next, we define the nonnegative quantity $\wt V_x(\Gamma)$ by taking the absolute value of all 
random variables in \eqref{value of lumping for bare stem} inside the expectation, i.e.
\begin{equation} \label{definition of V tilde}
\wt V_x(\Gamma) \;=\; \sum_{\b x} \Delta_\Gamma(\b x)\, Q_x(\b x)
\prod_{\gamma \in \Gamma} \E \prod_{e \in \gamma} \absb{\wh H_{x_{a(e)} x_{b(e)}}}
\,,
\end{equation}
Clearly,
\begin{equation*}
\abs{V_x(\Gamma)} \;\leq\; \wt V_x(\Gamma)\,.
\end{equation*}
Moreover, for a pairing $\Pi \in \scr P_{n,n'}$ we define the nonnegative quantity
\begin{equation} \label{definition of R}
R_x(\Pi) \;\deq\; M^{4 \delta m(\Pi)} \sum_{\b x} Q_x(\b x) \prod_{\{e,e'\} \in \Pi} \indb{\varrho_{\b x}(e) = \varrho_{\b 
x}(e')}  \sigma_{x_{a(e)} x_{b(e)}}^2 \,,
\end{equation}
which is essentially similar to $\wt V_x(\Pi)$ except that we drop the condition that different 
lumps must have different label pairs.

We may now bound the contribution of the higher order lumpings in terms of pairings.
\begin{lemma} \label{lemma: bound of lumps in terms of pairings}
We have that
\begin{equation} \label{bound on lumps in terms of pairings}
\sum_{\Gamma \in \scr G_{n,n'} \setminus \scr P_{n,n'}} \sum_x \wt V_x(\Gamma) \;\leq\; 
\sum_{\substack{\Pi \in \scr P_{n,n'}\\ m(\Pi) \geq 2}} \sum_x R_x(\Pi)\,.
\end{equation}
\end{lemma}

\begin{proof}
We have the bound
\begin{align*}
\wt V_x(\Gamma) &\;\leq\; \sum_{\b x} Q_x(\b x) \, \Delta_\Gamma(\b x) \qBB{\prod_{\gamma \in \Gamma} M^{\delta 
(\abs{\gamma} - 2)} \prod_{e \in \gamma} \sigma_{x_{a(e)} x_{b(e)}}}
\\
&\;=\; M^{\delta p(\Gamma)} \sum_{\b x} Q_x(\b x) \Delta_\Gamma(\b x) \qBB{\prod_{\gamma \in \Gamma} \prod_{e \in 
\gamma} \sigma_{x_{a(e)} x_{b(e)}}} \,,
\end{align*}
where in the first step we used that
\begin{equation*}
\E \abs{\wh H_{xy}}^{\abs{\gamma}} \;\leq\; \sigma_{xy}^{\abs{\gamma}} M^{\delta (\abs{\gamma} - 2)}\,.
\end{equation*}

Let $\h{\Pi(\Gamma)}_{\Gamma \in \scr G_{n,n'}}$ denote a choice of refining pairings 
satisfying \eqref{lower bound on number of skeleton edges}. Then from Lemma \ref{lemma: 
existence of refining pairing} we get
\begin{align}
\sum_{\Gamma \in \scr G_{n,n'} \setminus \scr P_{n,n'}} \sum_x \wt V_x(\Gamma) &\;\leq\; 
\sum_{\Pi \in \scr P_{n,n'}} \sum_{\Gamma \in \scr G_{n,n'}} \ind{\Pi(\Gamma) = \Pi} \sum_x \wt 
V_x(\Gamma)
\notag \\ \label{first step in estimating non-pairings}
&\;\leq\;
\sum_{\Pi \in \scr P_{n,n'}} M^{4 \delta m(\Pi)} \sum_{\Gamma \in \scr G_{n,n'}} \ind{\Pi(\Gamma) = \Pi}
\sum_{x, \b x} Q_x(\b x) \Delta_\Gamma(\b x) \qBB{\prod_{\gamma \in \Gamma} \prod_{e \in 
\gamma}  \sigma_{x_{a(e)} x_{b(e)}}}\,,
\end{align}
where the sums over $\Pi$ are constrained by $m(\Pi) \geq 2$.

Next, we introduce a family $\b \varrho_\Gamma = \{\varrho_\gamma\}_{\gamma \in \Gamma}$, where 
$\varrho_\gamma$ is an unordered pair of labels. Thus we may rewrite, for fixed $\b x$,
\begin{equation*}
\Delta_\Gamma(\b x) \qBB{\prod_{\gamma \in \Gamma} \prod_{e \in \gamma}  \sigma_{x_{a(e)} 
x_{b(e)}}} \;=\; \sum_{\b \varrho_\Gamma} \qBB{\prod_{\gamma \neq \gamma'} \ind{\varrho_\gamma \neq 
\varrho_{\gamma'}}} \qBB{\prod_{\gamma \in \Gamma} \prod_{e \in \gamma} \ind{\varrho_{\b x}(e) = 
\varrho_\gamma} \sigma_{x_{a(e)} x_{b(e)}}}\,.
\end{equation*}
We now relax the condition $\Pi(\Gamma) = \Pi$ in \eqref{first step in estimating non-pairings} 
to the condition that $\Pi$ is a refinement of $\Gamma$. We may then express $\Gamma$ as 
$\Gamma = \Gamma_P$ using a partition $P = \{p\}$ of the set of bridges $\Pi$, where $\Gamma_P$ 
is defined as $\Gamma_P = \{\gamma_p\}_{p \in P}$ and $\gamma_p \deq \bigcup_{\pi \in p} \pi$, 
i.e.\ $P$ expresses which bridges of $\Pi$ need to be lumped to obtain $\Gamma$.

Thus we get for $\Gamma = \Gamma_P$
\begin{equation*}
\sum_{\b \varrho_\Gamma} \qBB{\prod_{\gamma \neq \gamma'} \ind{\varrho_\gamma \neq \varrho_{\gamma'}}} 
\qBB{\prod_{\gamma \in \Gamma} \prod_{e \in \gamma} \ind{\varrho_{\b x}(e) = \varrho_\gamma}  
\sigma_{x_{a(e)} x_{b(e)}}}
\\
=\; \sum_{\b \varrho_\Pi} I_P(\b \varrho_\Pi) \prod_{\pi \in \Pi} \prod_{e \in \pi} \ind{\varrho_{\b 
x}(e) = \varrho_\pi} \, \sigma_{x_{a(e)} x_{b(e)}}\,,
\end{equation*}
where we defined
\begin{equation*}
I_P(\b \varrho_\Pi) \;\deq\; \qBB{\prod_{p \in P} \prod_{\pi, \pi' \in p} \ind{\varrho_\pi = 
\varrho_{\pi'}}} \qBB{\prod_{p \neq p'} \prod_{\pi \in p} \prod_{\pi' \in p'} \ind{\varrho_{\pi} \neq 
\varrho_{\pi'}}}\,.
\end{equation*}
The claim \eqref{bound on lumps in terms of pairings} now follows from the identity
\begin{equation*}
1 \;=\; \sum_{P} I_P(\b \varrho_\Pi)\,,
\end{equation*}
and the fact that any lumping $\Gamma \in \scr G_{n,n'}$ of which $\Pi$ is a refinement can be 
written as $\Gamma = \Gamma_P$ for some partition $P$ of the set of bridges $\Pi$.
\end{proof}

\subsection{Bounds on all lumpings}
In this final subsection we show that the contribution to \eqref{sum of bare trunk} of all 
non-pairings, as well as all tagged pairings different from the straight ladder of Subsection 
\ref{subsection: ladder}, vanishes as $W \to \infty$.
For a pairing $\Pi \in \scr P_{n,n'}$ and tagging $\vartheta \in \{0,1\}^\Pi$, we define $\wt V_x(\Pi, \vartheta)$ in 
the obvious way (see \eqref{definition of V tilde}, \eqref{tagged Delta}, and \eqref{tag splitting sums up to one}).  
Clearly, we have that
\begin{equation*}
\wt V_x(\Pi) \;=\; \sum_{\vartheta \in \{0,1\}^\Pi} \wt V_x(\Pi, \vartheta)\,.
\end{equation*}

For $n,n' \geq 0$ we define
\begin{equation} \label{definition of h}
h_{n,n'} \;\deq\; \sum_{\Gamma \in \scr G_{n,n'}} \sum_x \wt V_x(\Gamma)
\end{equation}
and
\begin{equation}\label{definition of h*}
h^*_{n,n'} \;\deq\; h_{n,n'} - \delta_{n n'} \wt V_x(L_n, \vartheta_n)
 =  h_{n,n'} - \delta_{n n'}  V_x(L_n, \vartheta_n)
\end{equation}
is the contribution of all diagrams apart from the main term, the straight ladder,
where we used that $ V_x(L_n, \vartheta_n) = \wt V_x(L_n, \vartheta_n)$.
We remark that in \cite{erdosknowles} $h^*_{n,n'}$ was denoted by $h_{n,n'}$.

\begin{lemma} \label{lemma: bound on stem lumpings}
For any integer $1 \leq p \leq M^\mu$ we have
\begin{equation} \label{bound on non-ladder lumpings}
\sum_{n + n' = 2p} h^*_{n,n'} \;\leq\; C M^{\mu /2 -1/3 + 8 \delta}\,.
\end{equation}
as well as
\begin{equation} \label{bound on all lumpings}
\qquad \sum_{n + n' = 2p} h_{n,n'} \;\leq\; C\,.
\end{equation}
\end{lemma}

\begin{proof}
The proof of \eqref{bound on non-ladder lumpings} is almost identical to the proof of Equation (7.10) in 
\cite{erdosknowles}. We bound general lumpings in terms of non-ladder pairings, whose contribution we estimate by 
analysing vertex orbits in skeleton graphs (see Sections 7.4 -- 7.6 in \cite{erdosknowles}).

More precisely, using Lemma \ref{lemma: bound of lumps in terms of pairings} we see that the only needed modification to 
the argument of \cite{erdosknowles} arises from the additional factor $M^{4 \delta m(\Pi)}$ in \eqref{definition of R} 
compared to Equation (7.1) of \cite{erdosknowles}. Let $\bar m$ denote the number of bridges in the skeleton $S(\Pi, 
\vartheta)$; then we have $\bar m \geq m(\Pi)$ by the definition \eqref{minimum skeleton size} of $m(\Pi)$.  Thus we 
find that Equation (7.9) of \cite{erdosknowles} (in which $\Gamma$ is now a tagged pairing not equal to a straight 
ladder) remains valid provided that the factor $M^{1/3} M^{-\bar m/3}$ is replaced with $M^{1/3} M^{-(1/ 3 - 4 \delta) 
\bar m}$.  Thus we find from Equation (7.10) of \cite{erdosknowles} that, for $1 \leq p \leq M^\mu$,
\begin{equation} \label{sum over all non-ladder pairings}
\sum_{n + n' = 2p} h^*_{n,n'} \;\leq\; \frac{1}{M} + \frac{M^{1/3}}{p} \pbb{\frac{1}{p^{1/2}} + \frac{1}{M^{1/6}}} 
\pbb{\frac{M}{M - 1}}^p \sum_{r = 2}^p \pbb{\frac{Cp}{M^{1/3 - 4 \delta}}}^r \, 2^r\,,
\end{equation}
where we emphasize the additional factor of $2^r$ arising from the sum over all bridge tags of skeleton pairings, as 
described in Section 9 of \cite{erdosknowles}. The first term $1/M$ accounts for the term $p = 1$ which consists of an 
antiladder with one rung whose contribution is trivially bounded by $1/M$. As explained at the end of Section 7.5 in 
\cite{erdosknowles}, the factor $p^{-1/2} + M^{-1/6}$ results from a detailed heat kernel estimate (Lemma 7.5 in 
\cite{erdosknowles}) which follows from the band structure of $H$. If, instead of the band structure, we had imposed 
only the two conditions $\sum_{y} \sigma_{xy}^2 = 1$ and $\sigma_{xy}^2 \leq M^{-1}$, then \eqref{sum over all 
non-ladder pairings} would be valid without the factor $p^{-1/2} + M^{-1/6}$.

Now \eqref{sum over all non-ladder pairings} immediately yields
\begin{equation} \label{bound on h star}
\sum_{n + n' = 2p} h^*_{n,n'} \;\leq\; C M^{\mu / 2 - 1/3 + 8 \delta}\,,
\end{equation}
which is \eqref{bound on non-ladder lumpings}.

Moreover, \eqref{bound on all lumpings} follows from \eqref{bound on non-ladder lumpings} and 
the estimate
\begin{equation*}
\sum_x \wt V_x(L_n, \vartheta_n) \;\leq\; 1\,;
\end{equation*}
see Subsection \ref{subsection: ladder}.
\end{proof}

From \eqref{sum over alphas} and \eqref{bound on non-ladder lumpings} we get
\begin{align}
\sum_{n+ n' \leq M^\mu} \abs{\alpha_n(W^{d \kappa} T) \alpha_{n'}(W^{d \kappa} T)} \, 
h^*_{n,n'} &\;\leq\; \pBB{\sum_{n+ n' \leq M^\mu} \abs{\alpha_n(W^{d \kappa} T)}^2 
\abs{\alpha_{n'}(W^{d \kappa} T)}^2}^{1/2} \pBB{\sum_{p \leq M^\mu} \sum_{n + n' = 2p} 
(h^*_{n,n'})^2}^{1/2}
\notag \\
&\;\leq\; \pBB{\sum_{p \leq M^\mu} \pBB{\sum_{n+n' = 2p} h^*_{n,n'}}^2}^{1/2}
\notag \\
&\;\leq\; C \pb{M^\mu M^{\mu - 2/3 + 16 \delta}}^{1/2}
\notag \\ \label{summing over alphas}
&\;=\; o(1)
\end{align}
as $W \to \infty$ (see \eqref{assumption on kappa and mu}). Then Proposition \ref{proposition: 
value of bare stem} follows from  \eqref{sum of bare trunk}, \eqref{limit of ladders},
\eqref{definition of h*} and \eqref{summing over alphas}.

\section{The boughs for $\kappa < 1/5$} \label{section: boughs}
In this section we estimate the contribution of the boughs. It turns out that strengthening our 
assumption on $\kappa$ to $\kappa < 1/5$ (from $\kappa < 1/3$) greatly simplifies the estimate 
of the boughs. Thus, throughout this section we assume that $\kappa < 1/5$. The next section is 
devoted to the case $\kappa < 1/3$.

In Section \ref{section: bare stem} we computed the contribution of the first term of 
\eqref{rho split into bare and garnished parts}; see Proposition \ref{proposition: value of 
bare stem}.  We now focus our attention on the remaining three terms of \eqref{rho split into 
bare and garnished parts}, and show that their $\ell^1$-norm in $x$ vanishes.
We need to estimate
\begin{equation} \label{definition of E_1}
E_1 \;\deq\; \sum_{n + n' \leq M^\mu}  \absb{\alpha_{n}(t) \, \alpha_{n'}(t)} \sum_x \sum_{\cal G \in \fra G_n^*} 
\sum_{\cal G' \in \fra G_{n'}^*} \sum_{\Gamma \in \scr G(G \cup G')} \absb{V_x(\cal G \cup \cal G', \Gamma)}
\end{equation}
and
\begin{equation} \label{definition of E_2}
E_2 \;\deq\;
\sum_{n + n' \leq M^\mu}  \absb{\alpha_{n}(t) \, \alpha_{n'}(t)} \sum_x \sum_{\cal G \in \fra G_{n}^*}
\sum_{\Gamma \in \scr G(G \cup I_{n'})} \absb{V_x(\cal G \cup \cal I_{n'}, \Gamma)}
\end{equation}
(It is easy to check that $E_2$ estimates both terms on the second line of \eqref{rho split 
into bare and garnished parts} since $\wh H$ is Hermitian).

\begin{proposition} \label{proposition: boughs vanish for small kappa}
Choose $\mu$ and $\delta$ so that
\begin{equation*}
\kappa + 4 \delta \;<\; \mu \;<\; 1/5 - 4 \delta\,.
\end{equation*}
Then
\begin{equation*}
\lim_{W \to \infty} E_1 \;=\; \lim_{W \to \infty} E_2 \;=\; 0\,.
\end{equation*}
\end{proposition}

The rest of this section is devoted to the proof of Proposition \ref{proposition: boughs vanish 
for small kappa}.
We expound our main argument for $E_1$. The estimate of $E_2$ is very similar, and we shall 
describe the required minor modifications in
 Subsection~\ref{subsection: bound on E_2}.

From \eqref{value of general lumping} we get
\begin{multline*}
E_1 \;\leq\; \sum_{n + n' \leq M^\mu} \abs{\alpha_n(t) \alpha_{n'}(t)} \sum_x \sum_{\cal G \in \fra G_n^*} \sum_{\cal G' 
\in \fra G_{n'}^*}
\sum_{\Gamma \in \scr G(G \cup G')} \; \sum_{\b x \,:\, \Gamma(\b x) = \Gamma}
\\
\times \prod_{\gamma \in \Gamma} \absbb{\E \prod_{e \in \gamma} P_{\tau(e)}\pb{\wh H_{x_{a(e)} x_{b(e)}}, \wh 
H_{x_{b(e)} x_{a(e)}}}}
\; \delta_{0 x_{a(G)}} \delta_{x x_{b(G)}} \delta_{x x_{a(G')}} \delta_{0 x_{b(G')}}
\prod_{v,w \in \cal V(G \cup G')} \pb{1 - l(v,w) \delta_{x_v x_w}}
\,,
\end{multline*}
where we abbreviated $\tau \equiv \tau_{G \cup G'}$ and $l \equiv l_{G \cup G'}$. 

Next, we relax all nonbacktracking conditions in $l$ pertaining to bough vertices. This gives
\begin{equation} \label{estimate on E_n}
E_1 \;\leq\; \sum_{n + n' \leq M^\mu} \abs{\alpha_n(t) \alpha_{n'}(t)} \sum_{\cal G \in \fra 
G_n^*} \sum_{\cal G' \in \fra G_{n'}^*}
\sum_{\Gamma \in \scr G(G \cup G')} \; \sum_{\b x \,:\, \Gamma(\b x) = \Gamma} Q(\b x)
\prod_{\gamma \in \Gamma} \absbb{ \E \prod_{e \in \gamma} P_{\tau(e)}\pb{\wh H_{x_{a(e)} 
x_{b(e)}}, \wh H_{x_{b(e)} x_{a(e)}}}}\,,
\end{equation}
where
\begin{multline} \label{definition of Q}
Q(\b x) \;\deq\; \delta_{0 x_{a(G)}} \delta_{0 x_{b(G')}} \delta_{x_{a(G')} x_{b(G)}} 
\qBB{\prod_{v,w \in \cal V(\cal S(G))} \ind{d(v,w) = 2} \ind{x_v \neq x_w}}
\\
\times
\qBB{\prod_{v,w \in \cal V(\cal S(G'))} \ind{d(v,w) = 2} \ind{x_v \neq x_w}}
\end{multline}
implements the nonbacktracking condition on the stems $\cal S(G)$ and $\cal S(G')$. The
estimate \eqref{estimate on E_n}
follows from
\begin{equation*}
\sum_x \delta_{0 x_{a(G)}} \delta_{x x_{b(G)}} \delta_{x x_{a(G')}} \delta_{0 
x_{b(G')}}\prod_{v,w \in \cal V(G \cup G')} \pb{1 - l(v,w) \delta_{x_v x_w}} \;\leq\; Q(\b 
x)\,,
\end{equation*}
since, by definition of $\fra G_n \subset \fra G_\sharp$, the stems $\cal S(G)$ and $\cal S(G')$ are 
completely nonbacktracking in $l$ (i.e.\ $l(v,w) = 1$ if $d(v,w) = 2$ and $v,w \in \cal V(\cal 
S(G) \cup \cal S(G))$).  Note that $Q$ depends only on the labels of stem vertices.

\subsection{Sketch of the argument} \label{subsection: sketch of simple boughs}
Before embarking on the estimate of $E_1$, we outline our strategy. We first fix the graph 
$\cal G \cup \cal G'$ and the lumping $\Gamma$. We assume that $\cal G \cup \cal G' \neq \cal 
I_n \cup \cal I_{n'}$ for all $n,n'$, i.e.\ we are not dealing with the bare stem. Starting 
from the bough leaves of $G \cup G'$, we sum successively over all vertex labels that do not 
belong to the stem. The order of summation is such that we sum over the label of a bough vertex 
only after we have summed over the labels of all of its children.

Our estimate uses two crucial facts. First, each leaf is a small edge (this is an immediate 
consequence of the growth process that generates boughs; see Proposition \ref{proposition: 
properties of graphs} (ii)).  This means that, if a leaf is not lumped with any other edge, its 
contribution is small.  Second, edges that are lumped together yield a small contribution owing 
to fixing of labels, which reduces the entropy factor associated with the summation of the 
labels.

Any large bough edge yields a contribution bounded by $1$, as follows from
\begin{equation} \label{l1-bound}
\sum_{x_{b(e)}} \E P_{(b,0)}\pb{\wh H_{x_{a(e)} x_{b(e)}}, \wh H_{x_{b(e)} x_{a(e)}}} \;=\;
\sum_{x_{b(e)}} \E \absb{\wh H_{x_{a(e)} x_{b(e)}}}^2
\;\leq\; 1\,.
\end{equation}
Ideally, we would hope that each leaf, being a small edge, yield a factor of essentially $M^{-1}$.  For example, if 
$\tau(e) = (b,2)$, summation over the label of the final vertex of $e$ yields
\begin{equation} \label{linfty-bound}
\absbb{\sum_{x_{b(e)}} \E P_{\tau(e)}\pb{\wh H_{x_{a(e)} x_{b(e)}}, \wh H_{x_{b(e)} x_{a(e)}}} }
\;=\; \sum_{x_{b(e)}} \E \absb{\wh H_{x_{a(e)} x_{b(e)}}}^4 \;\leq\; C \frac{M^{2 
\delta}}{M}\,.
\end{equation}
In this case the order of a bough with $l$ leaves would be $M^{-l}$ (up to an irrelevant factor 
$M^{2 \delta l}$). It is easy to see that a similar estimate holds for any leaf that is not 
lumped with another edge. This smallness fights against the combinatorics of the number of 
rooted, oriented trees with $k$ edges and $l$ leaves, which is of the order $k^{2l - 2} / 
l^{2l}$ (see \eqref{Naranya number} below). Thus we would find that the sum over the 
contributions of all rooted oriented trees with $k$ edges is
\begin{equation*}
\frac{1}{k^2}\sum_{l \geq 1} \frac{1}{l^{2l}}\pbb{\frac{k^2}{M}}^l \;\leq\; \frac{C}{M}\,,
\end{equation*}
since $k \leq M^\mu \leq M^{1/2}$. (The requirement $l \geq 1$ is simply a statement that there 
is at least one bough edge.)
It would then be a relatively straightforward matter to bound the contribution of all families of boughs growing from 
the stem, and to show that it vanishes as $W \to \infty$.

Unfortunately, this simple approach breaks down because two leaves of type $(b,1)$ lumped together yield a contribution
\begin{equation} \label{second moment of small edge}
\sum_y \E \pb{\abs{\wh H_{xy}}^2 - \sigma^2_{xy}}^2 \;\approx\; \sum_y \frac{CM^{2 \delta}}{M} \sigma_{xy}^2 \;=\; 
\frac{C M^{2 \delta}}{M}\,,
\end{equation}
which is much larger than the desired factor $M^{-2}$. We emphasize that this problem only 
occurs when a lump consists solely of leaves of type $(b,1)$. Indeed, lumping leaves with tags 
$(b,2)$, $(b,3)$ or $(b,4)$ yields a sufficiently high negative power of $M$ to keep the simple 
power counting mentioned above valid. For example, if two leaves of type $(b,2)$ are lumped, 
their contribution is
\begin{equation*}
\sum_y \E \pb{\abs{\wh H_{xy}}^4}^2 \;\leq\; \pbb{\frac{M^{2 \delta}}{M}}^3\,.
\end{equation*}
In fact, it would suffice that every lump had a single edge whose tag is not $(b,1)$ to ensure that each leaf yield a 
factor $1/M$.

In this section, we develop a method that extracts a factor $1/\sqrt{M}$ from each leaf (or, more precisely, a factor 
$1/M$ from pairs of leaves)
instead of the optimal factor $1/M$, thus allowing us to reach time scales of order $M^{1/5}$.  In order to reach time 
scales of order $M^{1/3}$, we need a decay of order $1/M$ from each leaf. This requires more effort and is done in 
Section \ref{section: boughs for kappa=1/3}.

Notice that the estimates of the type \eqref{l1-bound}-\eqref{second moment of small edge} rely on 
$\ell^1$-$\ell^\infty$-bounds on the variances, $\sum_y \sigma_{xy}^2 = 1$ and $\max_{y} \sigma_{x y}^2 \leq M^{-1}$ for 
each $x \in \Lambda_N$. In fact, all the estimates in Sections \ref{section: boughs} and \ref{section: boughs for 
kappa=1/3} rely on such power counting estimates.

\subsection{Ordering of edges and parametrization of lumpings}
There are two natural structures governing the vertex labels in the bound \eqref{estimate on 
E_n}: the tree graph $G \cup G'$ and the lumping $\Gamma$. In the case of the bare stem 
(Section \ref{section: bare stem}), we chose to sum over all vertex labels simultaneously, 
under the constraints imposed by $\Gamma$. This was possible because the tree graph $I_n \cup 
I_{n'}$ of the bare stem was very simple. For a general tree graph $G \cup G'$, however, this 
approach breaks down. Instead, we have to sum over the vertex labels in a manner dictated by 
the structure of the tree graph $G \cup G'$, i.e.\ successively over each individual vertex 
label, starting from the leaves. If all bough edges were in their own single-edge lumps, this 
strategy would be easy to implement. For a general lumping, however, we have additional 
constraints on the bough vertex labels arising from the lumping, which are completely nonlocal 
and in this sense conflicting with the constraints resulting from the tree graph structure $G 
\cup G'$. We overcome this difficulty by introducing a special parametrization for lumpings 
(denoted by $(\wt \Gamma, A) \mapsto \Gamma$ below) that is suited to a successive summation 
along the bough branches. This parametrization is also needed for controlling the summation 
over all lumpings $\Gamma$.

Let us fix $n,n'$ as well as $\cal G = (G, \tau_G) \in \fra G_n^*$ and $\cal G' = (G', \tau_{G'}) \in \fra G_n^*$ in the 
summation \eqref{estimate on E_n}. We abbreviate $\cal E_B \deq \cal E(\cal B(G) \cup \cal B(G'))$ for the set of bough 
edges.
\label{page: total order 1}
Recall that a \emph{leaf} is an edge $e \in \cal E_B$ such that $b(e)$ has degree one.  We now 
introduce a total order $\preceq$ on the set of all edges $\cal E(G \cup G')$.  This order will 
govern the order of the summation of the vertex labels. We use the notation $e \prec e'$ to 
mean $e \preceq e'$ and $e \neq e'$. We impose the following conditions of $\preceq$.
\begin{enumerate}
\item
If $e$ and $e'$ are both bough edges and $e'$ is the parent of $e$ (i.e.\ $a(e) = b(e')$) then 
$e \prec e'$.
\item
We start the ordering from the leaves: If $e$ is a leaf and $e'$ is not a leaf then $e \prec e'$.
\item
Bough edges are smaller than stem edges: If $e$ is a bough edge and $e'$ a stem edge then $e 
\prec
e'$.
\end{enumerate}
It is easy to see that such an order $\preceq$ exists. We choose one and consider it fixed in 
the sequel. Once $\preceq$ is given, each edge $e \in \cal E(G \cup G')$ (except the last edge) 
has a \emph{successor}, denoted by $\sigma(e)$ and defined as the smallest edge strictly 
greater than $e$. Note that the order $\preceq$ is not the same as the (partial) order induced 
by the directedness of the graph. Similarly, the concepts of successor and child are unrelated.

\label{page: A}
We shall sum over the vertex labels of the boughs, starting from the degree one vertices of the leaves.  To this end, we 
need a parametrization of the lumping $\Gamma  \in \scr G(G \cup G')$ that is suited for such a successive summation. 
The parametrization will be given by a map $e \mapsto A_e$ on the set $\cal E_B$, and by $\wt \Gamma$, defined as the 
restriction of $\Gamma$ to the stem edges. The idea behind the construction of $A$ is to set $A_e$ to be the smallest 
edge in the lump containing $e$ with the property that $A_e \succ e$; if there is no such edge, we set $A_e = e$.

\begin{definition} \label{definition of A}
Denote by $\scr A(G \cup G')$ the set of mappings
\begin{equation*}
A : \cal E_B \to \cal E(G \cup G')\,, \quad e \;\mapsto\; A_e\,,
\end{equation*}
with the following two properties. First, $A_e \succeq e$ for all $e$. Second, if $e', e'' 
\prec e$ satisfy $A_{e'} = A_{e''} = e$ then $e' = e''$.
\end{definition}

The following definition will be used to reconstruct $\Gamma$ from the pair $(\wt \Gamma, A)$.

\begin{definition} \label{definition of lump parametrization}
Let $\wt \Gamma$ be a lumping of the stem edges $\cal E(\cal S(G) \cup \cal S(G'))$, and $A \in \scr A(G \cup G')$. Then 
we define $\Gamma(\wt \Gamma, A)$ as the finest equivalence relation on $\cal E(G \cup G')$ (denoted by $\sim$) for 
which $A_e \sim e$ for all $e$ and $e \sim e'$ whenever $e$ and $e'$ belong to the same lump of $\wt \Gamma$.
\end{definition}

Next, let $u$ and $u'$ denote the number of edges in $\cal S(G)$ and $\cal S(G')$ respectively. Note that $u + u'$ is 
even. This is easy to see from the facts that stem edges have odd degree, bough edges have even degree, and the total 
degree $n + n' = \deg(\cal G \cup \cal G')$ is even. We have the following result which shows that any lumping $\Gamma$ 
can be encoded using a lumping $\wt \Gamma$ of the stem and a map $A \in \scr A(G \cup G')$.

\begin{lemma} \label{lemma: lumpings parametrized}
For each $\Gamma \in \scr G(G \cup G')$ there is a pair $(\wt \Gamma, A) \in \scr G_{u,u'} 
\times \scr A(G \cup G')$ such that $\Gamma = \Gamma(\wt \Gamma, A)$.
\end{lemma}

\begin{proof}
Let $\Gamma \in \scr G_{G \cup G'}$ be given. We define $\wt \Gamma$ to be the restriction of 
$\Gamma$ to the set $\cal E_S \deq \cal E(\cal S(G) \cup \cal S(G'))$, i.e.\ $\wt \Gamma = 
\{\gamma \cap \cal E_S\}_{\gamma \in \Gamma}$. We now claim that $\wt \Gamma \in \scr 
G_{u,u'}$. Indeed, by definition of $\scr G(G \cup G')$, each $\gamma \in \Gamma$ contains an 
even number of stem edges, which implies that the lumps of $\wt \Gamma$ are of even size.

In order to define $A$, we assign to each bough edge $e \in \cal E_B$ the smallest edge $e' 
\succ e$, $e' \in \cal E(G \cup G')$ in the same lump as $e$. If no such edge exists, we set 
$A_e \deq e$; otherwise we set $A_e \deq e'$. It is now immediate that $\Gamma = \Gamma(\wt 
\Gamma, A)$. In fact this is even a one-to-one map (a fact we shall not need however).
\end{proof}

We now make use of Lemma \ref{lemma: lumpings parametrized} to sum labels $x_v$ of bough 
vertices $v$ in \eqref{estimate on E_n}, starting from the leaves. Let us write
\begin{equation} \label{bound of E_1 in terms of E(G)}
E_1 \;\leq\; \sum_{n + n' \leq M^\mu} \abs{\alpha_n(t) \alpha_{n'}(t)}  \sum_{\cal G \in \fra 
G_n^*} \sum_{\cal G' \in \fra G_{n'}^*} E_{\cal G \cup \cal G'}\,,
\end{equation}
where we defined
\begin{align*}
E_{\cal G \cup \cal G'} &\;\deq\; \sum_{\Gamma \in \scr G(G \cup G')} \; \sum_{\b x \,:\, 
\Gamma(\b x) = \Gamma} Q(\b x)
\prod_{\gamma \in \Gamma} \absbb{ \E \prod_{e \in \gamma} P_{\tau(e)}\pb{\wh H_{x_{a(e)} 
x_{b(e)}}, \wh H_{x_{b(e)} x_{a(e)}}}}
\\
&\;\leq\; \sum_{\wt \Gamma \in \scr G_{u,u'}} \sum_{A} \; \sum_{\b x \,:\, \Gamma(\b x) = 
\Gamma(\wt \Gamma, A)} Q(\b x) \prod_{\gamma \in \Gamma(\wt \Gamma, A)} \absbb{\E \prod_{e \in 
\gamma} P_{\tau(e)}\pb{\wh H_{x_{a(e)} x_{b(e)}}, \wh H_{x_{b(e)}
x_{a(e)}}}}\,.
\end{align*}
The inequality follows from Lemma \ref{lemma: lumpings parametrized}. Here $u = \absb{\cal 
E(\cal S(G))}$ and $u' = \absb{\cal E(\cal S(G'))}$. Moreover, the summation over $A$ is 
understood to mean summation over all $A \in \scr A(G \cup G')$.

Let us partition the vertex labels $\b x$ into bough labels $\b x_B$ and stem labels $\b x_S$, 
i.e.
\begin{equation} \label{splitting of labels}
\b x \;=\; \pb{x_v \,:\, v \in \cal V(G \cup G')} \;=\; (\b x_B, \b x_S)\,,
\end{equation}
where
\begin{equation} \label{definition of split labels}
\b x_B \;\deq\; \pb{x_{b(e)} \,:\, e \in \cal E_B}\,, \qquad \b x_S \;\deq\; \pb{x_v \,:\, v 
\in \cal V\pb{\cal S(G) \cup \cal S(G')}}\,.
\end{equation}
Recall that $Q(\b x) = Q(\b x_S)$; see \eqref{definition of Q}.
Thus we get
\begin{multline} \label{error bound for fixed graphs}
E_{\cal G \cup \cal G'} \;\leq\; \sum_{\wt \Gamma \in \scr G_{u,u'}} \; \sum_{\b x_S \,:\, 
\Gamma(\b x_S) = \wt \Gamma} Q(\b x_S)
\\
\times
\sum_{A} \; \sum_{\b x_B} \qBB{\prod_{e \in \cal E_B} \indb{\varrho_{\b x}(e) = \varrho_{\b x}(A_e)}}
\prod_{\gamma \in \Gamma(\wt \Gamma, A)} \absbb{\E \prod_{e \in \gamma} P_{\tau(e)}\pb{\wh 
H_{x_{a(e)} x_{b(e)}}, \wh H_{x_{b(e)}
x_{a(e)}}}}\,.
\end{multline}
Equation \eqref{error bound for fixed graphs} is our starting point for estimating the 
contribution of the boughs. 

The roadmap for the following subsections is as follows. We start by fixing all summation 
variables in \eqref{error bound for fixed graphs}. In a first step, we sum over the bough 
labels $\b x_B$ (Subsection \ref{section: sum over bough labels}). In a second step, we sum 
over the bough lumpings $A$ (Subsection \ref{section: sum over easy bough lumpings}). The 
result of these summations is the bound \eqref{final bound on E(G)} on $E_{\cal G \cup \cal 
G'}$. In a third step, we sum over all stem labels (i.e.\ $\b x_S$ and $\wt \Gamma$) which 
yields a factor $h_{u,u'}$ (Subsection \ref{section: decoupling of the easy boughs}).  In a 
fourth step, we plug the estimate \eqref{final bound on E(G)} back into \eqref{bound of E_1 in terms of E(G)} and sum 
over the tagging $\tau$ (Subsections \ref{section: decoupling of the easy boughs} and \ref{subsection: sum over easy 
bough tags}). Finally, we sum over the bough graphs $G, G'$ (Subsection \ref{subsection: sum over bough graphs}).

\subsection{Sum over bough labels} \label{section: sum over bough labels}
In this subsection we fix $\cal G, \cal G'$, $A$ as well as an order $\preceq$, and sum over 
$\b x_B$ in \eqref{error bound for fixed graphs}.
The following definitions will prove helpful.

\begin{definition} \label{definition of inverse of A}
On $\cal E_B$ we define the \emph{inverse} $A^{-1}$ of $A$ by setting $A^{-1}_e \deq e'$ if 
there exists a (necessarily unique) $e' \prec e$ such that $A_{e'} = e$; otherwise we set 
$A^{-1}_e = e$. Obviously, $A^{-1}_e \preceq e$.
We say that a bough edge $e$ is \emph{lonely} (with respect to $A$) if $e = A_e = A^{-1}_e$.  
\end{definition}
Note that $e$ is lonely with respect to $A$ if and only if $e$ is the only edge in its lump of 
$\Gamma(\wt \Gamma, A)$ (this property is independent of $\wt \Gamma$).

For now we assume that all nonleaf bough edges have tag $(b,0)$; dealing with different nonleaf bough tags is very easy 
and is done at the end of this subsection. Define the new tagging $\wt \tau \equiv \wt \tau_A$  through
\begin{equation*}
\wt \tau(e) \;\deq\;
\begin{cases}
\tau(e) & \text{if $e$ is not a bough leaf}
\\
(b, 2) & \text{if $e$ is a lonely bough leaf}
\\
(b,5) & \text{if $e$ is a nonlonely bough leaf}\,,
\end{cases}
\end{equation*}
where we introduced the new bough tag $(b,5)$ whose associated polynomial (see table on page 
\pageref{polynomial table}) reads
\begin{equation*}
P_{(b,5)}(\wh H_{xy}, \wh H_{yx}) \;\deq\; 2 M^{2 \delta} \sigma_{xy}^2\,.
\end{equation*}
The motivation behind this definition is the following. If $e$ is a lonely leaf, its 
contribution to \eqref{error bound for fixed graphs} can be bounded by
\begin{equation} \label{bound on a lonely leaf}
\absB{\E P_{\tau(e)}\pb{\wh H_{x_{a(e)} x_{b(e)}}, \wh H_{x_{b(e)}
x_{a(e)}}}}
\;\leq\;
\E \absb{P_{(b,2)}\pb{\wh H_{x_{a(e)} x_{b(e)}}, \wh H_{x_{b(e)}
x_{a(e)}}}} \;\leq\; \frac{M^{2 \delta}}{M} \sigma_{x_{a(e)} x_{b(e)}}^2\,,
\end{equation}
as can be easily seen from Proposition \ref{proposition: properties of graphs} (ii) and Lemma 
\ref{lemma: bounds on truncated variance}. If $e$ is a leaf that is not lonely, its 
contribution in the worst case is of the same order as if its tag were $(b,0)$. Here the worst 
case is given by $\tau(e) = (b,1)$. The best we can do is use the trivial bound
\begin{equation} \label{bound on a nonlonely leaf}
\absb{P_{(b,i)}(\wh H_{xy}, \wh H_{yx})} \;\leq\; \absb{P_{(b, 5)}(\wh H_{xy}, \wh H_{yx})}
\end{equation}
for all $i$.
From \eqref{bound on a lonely leaf} and \eqref{bound on a nonlonely leaf} we see that the 
smallness of a leaf of type $(b,1)$ is only useful if it is lonely; otherwise, its contribution 
is the same as if it were an edge of type $(b,0)$.
For instance, we have
\begin{equation*}
\absB{\E P_{(b,1)}\pb{\wh H_{x y}, \wh H_{y x}}^2} \;=\; \E \pb{\abs{\wh H_{xy}}^2 - 
\sigma_{xy}^2}^2 \;\approx\; M^{2 \delta} \sigma_{xy}^4 \;\approx\;
\absB{\E P_{(b,0)}\pb{\wh H_{x y}, \wh H_{y x}}^2}
\,.
\end{equation*}

Now we claim that
\begin{equation} \label{tau bounded in terms of tau tilde}
\prod_{\gamma \in \Gamma(\wt \Gamma, A)} \absbb{\E \prod_{e \in \gamma} P_{\tau(e)}\pb{\wh 
H_{x_{a(e)} x_{b(e)}}, \wh H_{x_{b(e)}
x_{a(e)}}}}
\;\leq\; \prod_{\gamma \in \Gamma(\wt \Gamma, A)} \E \prod_{e \in \gamma} \absB{P_{\wt 
\tau(e)}\pb{\wh H_{x_{a(e)} x_{b(e)}}, \wh H_{x_{b(e)}
x_{a(e)}}}}\,.
\end{equation}
Indeed, this follows immediately from \eqref{bound on a lonely leaf}, \eqref{bound on a 
nonlonely leaf}, and the definition of $\wt \tau$. In fact, the definition of $\wt \tau$ was 
chosen so as to satisfy \eqref{tau bounded in terms of tau tilde}.

Next, we sum over the bough labels $\b x_B$ in the formula, obtained from \eqref{error bound 
for fixed graphs} and \eqref{tau bounded in terms of tau tilde},
\begin{equation} \label{starting point for sum over labels}
\sum_{\b x_B} \qBB{\; \prod_{e \in \cal E_B} \indb{\varrho_{\b x}(e) = \varrho_{\b x}(A_e)}}
\prod_{\gamma \in \Gamma(\wt \Gamma, A)} \E \prod_{e \in \gamma} \absb{P_{\wt \tau(e)}\pb{\wh 
H_{x_{a(e)} x_{b(e)}}, \wh H_{x_{b(e)}
x_{a(e)}}}}
\end{equation}
by successively summing up all labels in $\b x_B$, in the order defined by $\preceq$.
We denote the current summation edge by $\bar e \in \cal E_B$ (meaning that in the current
 step we sum over the label 
$x_{b(\bar e)}$), and call $\bar e$ the \emph{running edge}. When we tackle the edge $\bar e$,
 we shall \emph{sum it 
out}, by which we mean that we sum over the label $x_{b(\bar e)}$ of the final vertex of $\bar e$,
 and think of $\bar e$ 
as being struck from the graph $G \cup G'$. Thus, if the running edge is $\bar e$, then 
all edges $e \prec \bar e$ have 
already been summed out, and hence struck from $G \cup G'$. In this manner we shall successively 
sum out all bough edges 
and strike them all from $G \cup G'$.

For a running edge $\bar e \in \cal E_B$ define the subset of bough edges
\begin{equation} \label{definition of B^e}
B^{(\bar e)} \;\deq\; \hb{e \in \cal E_B \,:\, e \succeq \bar e}\,.
\end{equation}
The set $B^{(\bar e)}$ represents the bough edges that have not yet been
 summed out when $\bar e$ is the running edge.
We also abbreviate
\begin{equation} \label{definition of x^e and A^e}
\b x^{(\bar e)} \;\deq\; \pb{x_{b(e)} \,:\, e \in B^{(\bar e)}}\,,
\qquad
A^{(\bar e)} \;\deq\; \pb{A_e \,:\, e \in B^{(\bar e)}}\,.
\end{equation}
If $\bar e$ is a bough edge, we define
\begin{equation} \label{definition of recursive bound}
R^{(\bar e)} \;\deq\; \sum_{\b x^{(\bar e)}} \qBB{\; \prod_{e \in B^{(\bar e)}} \indb{\varrho_{\b 
x}(e) = \varrho_{\b x}(A_e)}}
\prod_{\gamma \in \Gamma(\wt \Gamma, A^{(\bar e)})} \E \prod_{e \in \gamma} \absb{P_{\wt 
\tau(e)}\pb{\wh H_{x_{a(e)} x_{b(e)}}, \wh H_{x_{b(e)}
x_{a(e)}}}}\,.
\end{equation}
If $\bar e$ is not a bough edge, we set $R^{(\bar e)} \deq 1$.

Let $e_0$ be the first edge of $\cal E(G \cup G')$. Moreover, \eqref{tau bounded in terms of 
tau tilde} yields
\begin{equation*}
\sum_{\b x_B} \qBB{\; \prod_{e \in \cal E_B} \indb{\varrho_{\b x}(e) = \varrho_{\b x}(A_e)}}
\prod_{\gamma \in \Gamma(\wt \Gamma, A)} \absbb{\E \prod_{e \in \gamma} P_{\tau(e)}\pb{\wh 
H_{x_{a(e)} x_{b(e)}}, \wh H_{x_{b(e)}
x_{a(e)}}}} \;\leq\; R^{(e_0)}\,.
\end{equation*}
We now proceed recursively, starting with $\bar e = e_0$, summing over $x_{b(\bar e)}$, then 
setting $\bar e$ to be the next edge (with respect to $\preceq$), summing over $x_{b(\bar e)}$, and so on until $\bar e$ 
is the first stem edge. In other words, we successively
sum out all bough edges in the order specified by $\preceq$.
 At each step, we get a bound of the form
\begin{equation*}
R^{(\bar e)} \;\leq\; \xi(\bar e, A) \, R^{(\sigma(\bar e))}\,,
\end{equation*}
where $\xi(\bar e, A) > 0$ is the factor resulting from the summation over $x_{b(\bar e)}$.  
Recall that $\sigma(\bar e)$ is the successor (with respect to $\preceq$) of $\bar e$. The 
following lemma gives an expression for $\xi(\bar e, A)$. It also identifies the ``bad 
leaves'', i.e.\ the leaves whose contribution to the right-hand side of \eqref{tau bounded in 
terms of tau tilde} is of order one, as the leaves $e$ that satisfy $A_e^{-1} \prec e = A_e$.  Our approach will 
eventually work because the number of bad leaves cannot be too large (see Lemma \ref{bound on the number of bad leaves} 
below).

\begin{lemma} \label{lemma: definition of easy xi}
For each $\bar e \in \cal E_B$ we have the bound $R^{(\bar e)} \;\leq\; \xi(\bar e, A) 
R^{(\sigma(\bar e))}$, where
\begin{equation*}
\xi(\bar e, A) \;\deq\;
\begin{cases}
\frac{2M^{2 \delta}}{M} + \ind{A_{\bar e} = \bar e} & \text{if $\bar e$ is not a leaf}
\\
\frac{2 M^{2 \delta}}{M} + \indb{A^{-1}_{\bar e} \prec \bar e = A_{\bar e}} 2 M^{2 \delta} & 
\text{if $\bar e$ is a leaf}.
\end{cases}
\end{equation*}
\end{lemma}

\begin{proof}
Assume first that $\bar e \in \cal E_B$ is not a leaf.
Then we have $\wt \tau(\bar e) = (b,0)$ (recall that we assumed that all nonleaf bough edges 
have tag $(b,0)$).
If $A_{\bar e} = \bar e$, we get
\begin{align}
R^{(\bar e)} &\;\leq\;
\sum_{\b x^{(\sigma(\bar e))}} \qBB{\; \prod_{e \in B^{(\sigma(\bar e))}} \indb{\varrho_{\b x}(e) 
= \varrho_{\b x}(A_e)}}
\prod_{\gamma \in \Gamma(\wt \Gamma, A^{(\sigma(\bar e))})} \E \prod_{e \in \gamma} 
\absb{P_{\wt \tau(e)}\pb{\wh H_{x_{a(e)} x_{b(e)}}, \wh H_{x_{b(e)}
x_{a(e)}}}}
\notag \\
&\qquad \times \sum_{x_{b(\bar e)}}\E \absb{P_{\wt \tau(\bar e)}\pb{\wh H_{x_{a(\bar e)} 
x_{b(\bar e)}}, \wh H_{x_{b(\bar e)}
x_{a(\bar e)}}}}
\notag \\
&\;\leq\; R^{(\sigma(\bar e))}\,, \label{recursive estimate for nonleaf}
\end{align}
where in the second step we used that $\wt \tau(\bar e) = (b,0)$, and consequently
\begin{equation*}
\sum_{x_{b(\bar e)}}\E \absb{P_{\wt \tau(\bar e)}\pb{\wh H_{x_{a(\bar e)} x_{b(\bar e)}}, \wh H_{x_{b(\bar e)} x_{a(\bar 
e)}}}} \;=\; \sum_{x_{b(\bar e)}} \E \absb{\wh H_{x_{a(\bar e)} x_{b(\bar e)}}}^2 \;\leq\; 1
\end{equation*}
to sum over the label $x_{b(\bar e)}$.

If $A_{\bar e} \succ \bar e$, we get
\begin{align}
R^{(\bar e)} &\;\leq\;
\sum_{\b x^{(\sigma(\bar e))}} \qBB{\; \prod_{e \in B^{(\sigma(\bar e))}} \indb{\varrho_{\b x}(e) 
= \varrho_{\b x}(A_e)}}
\prod_{\gamma \in \Gamma(\wt \Gamma, A^{(\sigma(\bar e))})} \E \prod_{e \in \gamma} 
\absb{P_{\wt \tau(e)}\pb{\wh H_{x_{a(e)} x_{b(e)}}, \wh H_{x_{b(e)}
x_{a(e)}}}}
\notag \\
&\qquad \times \sum_{x_{b(\bar e)}} \indb{\varrho_{\b x}(\bar e) = \varrho_{\b x}(A_{\bar e})} \;
\normb{P_{\wt \tau(\bar e)}\pb{\wh H_{x_{a(\bar e)} x_{b(\bar e)}}, \wh H_{x_{b(\bar e)}
x_{a(\bar e)}}}}_\infty
\notag \\ \label{recursive bound for lumped edge}
&\;\leq\; \frac{2M^{2 \delta}}{M} R^{(\sigma(\bar e))}\,,
\end{align}
where in the last step we used the bound
\begin{equation} \label{bound on Linfty norm of polynomial}
\normb{P_{(b,0)}\pb{\wh H_{x y}, \wh H_{y
x}}}_\infty \;\leq\; M^{2 \delta} \sigma_{x y}^2 \;\leq\; \frac{M^{2 \delta}}{M}
\end{equation}
as well as
\begin{equation*}
\indb{\varrho_{\b x}(\bar e) = \varrho_{\b x}(A_{\bar e})} \;\leq\; \indb{x_{b(\bar e)} = 
x_{b(A_{\bar e})}} +\indb{x_{b(\bar e)} = x_{a(A_{\bar e})}}\,.
\end{equation*}
Note that the summation over $x_{b(\bar e)}$ in \eqref{recursive bound for lumped edge} is 
restricted to the two values $x_{a(A_{\bar e})}$ and $x_{b(A_{\bar e})}$, which are fixed as 
they belong to $\b x^{(\sigma(\bar e))}$. This concludes the proof of Lemma \ref{lemma: 
definition of easy xi} in the case that $\bar e$ is not a leaf.

Next, let $\bar e \in \cal E_B$ be a leaf.
If $A_{\bar e}^{-1} = \bar e = A_{\bar e}$ then $\bar e$ is lonely and $\wt \tau(\bar e) = 
(b,2)$. Thus we get, exactly as in \eqref{recursive estimate for nonleaf} and using 
\eqref{bound on a lonely leaf},
\begin{equation*}
R^{(\bar e)} \;\leq\; \frac{M^{2 \delta}}{M} R^{(\sigma(\bar e))}\,.
\end{equation*}
If $A^{-1}_{\bar e} \prec \bar e = A_{\bar e}$ then $\wt \tau(\bar e) = (b,5)$. Therefore, using \begin{equation*}
\normb{P_{(b,5)}\pb{\wh H_{x y}, \wh H_{y x}}}_\infty \;\leq\; 2 M^{2 \delta} \sigma_{xy}^2
\end{equation*}
we get, as in \eqref{recursive estimate for nonleaf},
\begin{equation*}
R^{(\bar e)} \;\leq\; 2 M^{2 \delta} R^{(\sigma(\bar e))}\,.
\end{equation*}
If $\bar e \prec A_{\bar e}$ then $\wt \tau(\bar e) = (b,5)$ and we find, as in \eqref{recursive bound for lumped edge},
\begin{equation*}
R^{(\bar e)} \;\leq\; \frac{2 M^{2 \delta}}{M} R^{(\sigma(\bar e))}\,.
\end{equation*}
This concludes the proof.
\end{proof}

Putting everything together we get by iteration, for a fixed $\b x_S$,
\begin{multline} \label{bound on sum over bough labels}
\sum_{\b x_B} \qBB{\; \prod_{e \in \cal E_B} \indb{\varrho_{\b x}(e) = \varrho_{\b x}(A_e)}}
\prod_{\gamma \in \Gamma(\wt \Gamma, A)} \absbb{\E \prod_{e \in \gamma} P_{\tau(e)}\pb{\wh 
H_{x_{a(e)} x_{b(e)}}, \wh H_{x_{b(e)}
x_{a(e)}}}}
\\
\leq\; F(A) \, \prod_{\gamma \in \wt \Gamma} \E \prod_{e \in \gamma} \absB{P_{\tau(e)}\pb{\wh 
H_{x_{a(e)} x_{b(e)}}, \wh H_{x_{b(e)}
x_{a(e)}}}}\,,
\end{multline}
where
\begin{equation} \label{definition of F(A)}
F(A) \;\deq\; \prod_{e \in \cal E_B \text{ nonleaf}} \pbb{\frac{2M^{2 \delta}}{M} + \ind{A_{e} 
= e}}
\prod_{e \in \cal E_B \text{ leaf}} \pbb{\frac{2 M^{2 \delta}}{M} + \indb{A^{-1}_{e} \prec e = 
A_{e}} \, 2 M^{2 \delta}}\,.
\end{equation}

So far we assumed that all nonleaf bough tags were $(b,0)$.
Now we deal with arbitrary taggings. We split the tagging $\tau = (\tau_B, \tau_S)$ into a bough and stem tagging, where
\begin{equation*}
\tau_B \;\deq\; \pb{\tau(e) \,:\, e \in \cal E_B}\,,
\qquad
\tau_S \;\deq\; \pb{\tau(e) \,:\, e \in \cal E(\cal S(G) \cup \cal S(G'))}\,.
\end{equation*}
We now define $F(A, \tau_B)$ in such a way that \eqref{bound on sum over bough labels}, with $F(A)$ replaced by $F(A, 
\tau_B)$, holds for an arbitrary tagging $\tau$.

Let $e$ be a nonleaf bough edge. If $\tau(e) = (b,i)$ for $i > 0$, Proposition 
\ref{proposition: properties of graphs} (iii) implies that $i \geq 2$. Therefore the bound
\begin{equation*}
\absb{P_{(b,i)}(\wh H_{xy}, \wh H_{yx})} \;\leq\; \frac{M^{2 \delta}}{M} \abs{\wh H_{xy}}^2\,,
\end{equation*}
valid for all $i \geq 2$, implies that each nonleaf bough edge whose tag is not $(b,0)$ 
contributes
an additional factor $M^{-1 + 2 \delta}$ to the right-hand side of \eqref{bound on sum over bough labels} compared to if 
its tag were $(b,0)$. Thus we have that, for an arbitrary tagging $\tau = (\tau_B, \tau_S)$, the estimate \eqref{bound 
on sum over bough labels} is valid with $F(A)$ replaced by
\begin{multline} \label{definition of F(A,tau)}
F(A, \tau_B) \;\deq\; \qBB{\prod_{e \in \cal E_B \text{ nonleaf}} \pbb{\frac{M^{2 
\delta}}{M}}^{\ind{\tau(e) \neq (b,0)}} \pbb{\frac{2M^{2 \delta}}{M} + \ind{A_{e} = e}}}
\\ \times
\prod_{e \in \cal E_B \text{ leaf}} \pbb{\frac{2 M^{2 \delta}}{M} + \indb{A^{-1}_{e} \prec e = 
A_{e}} \, 2 M^{2 \delta}}\,.
\end{multline}
Thus we get from \eqref{error bound for fixed graphs}
\begin{equation} \label{result over summing over bough labels}
E_{\cal G \cup \cal G'} \;\leq\; \sum_{\wt \Gamma \in \scr G_{u,u'}} \; \sum_{\b x_S \,:\, 
\Gamma(\b x_S) = \wt \Gamma} Q(\b x_S)
\sum_{A} F(A, \tau_B) \, \prod_{\gamma \in \wt \Gamma} \E \prod_{e \in \gamma} 
\absB{P_{\tau(e)}\pb{\wh H_{x_{a(e)} x_{b(e)}}, \wh H_{x_{b(e)}
x_{a(e)}}}}\,.
\end{equation}

\subsection{Sum over bough lumpings} \label{section: sum over easy bough lumpings}
In this subsection we estimate $\sum_A F(A, \tau_B)$. Let $\cal E_l \subset \cal E_B$ denote 
the subset of bough leaves. Multiplying out the product over leaves in \eqref{definition of 
F(A,tau)} yields
\begin{multline} \label{multiplying out leaf product}
F(A, \tau_B) \;\leq\; \qBB{\prod_{e \in \cal E_B \text{ nonleaf}} \pbb{\frac{M^{2 
\delta}}{M}}^{\ind{\tau(e) \neq (b,0)}} \pbb{\frac{2M^{2 \delta}}{M} + \ind{A_{e} = e}}}
\\
\times
\sum_{a \in \{0,1\}^{\cal E_l}} \; \qBB{\prod_{e \in \cal E_B \text{ leaf}} \pbb{\frac{2 M^{2 
\delta}}{M}}^{1 - a_e} \pbb{\indb{A^{-1}_{e} \prec e = A_{e}} \, 2 M^{2 \delta}}^{a_e}}\,.
\end{multline}
Let $L \deq \abs{\cal E_l}$ denote the number of bough leaves.
We claim that the right-hand side of \eqref{multiplying out leaf product} vanishes unless 
$\abs{a} \deq \sum_e a_e \leq [L/2]$, where $[\cdot]$ denotes integer part. This is an 
immediate consequence of the following Lemma.

\begin{lemma} \label{bound on the number of bad leaves}
The set of bad leaves $\cal L \deq \{e \in \cal E_l \,:\, A^{-1}_e \prec e = A_e\}$ contains at 
most $[L/2]$ elements.
\end{lemma}
\begin{proof}
If $e \in \cal L$ then it follows from the definition of $\preceq$ that $A^{-1}_e \in \cal E_l 
\setminus \cal L$. In words: A bad leaf always comes with a unique companion that is not bad.
\end{proof}

Abbreviating $\sum_{a \in \{0,1\}^{\cal E_l} \,:\, \abs{a} \leq [L/2]}$ by $\sum_{\abs{a} \leq 
[L/2]}$, we get from \eqref{multiplying out leaf product}
\begin{align} \label{longsum}
\sum_A F(A, \tau_B) &\;\leq\; \sum_A \qBB{\prod_{e \in \cal E_B \text{ nonleaf}} 
\pbb{\frac{M^{2 \delta}}{M}}^{\ind{\tau(e) \neq (b,0)}} \pbb{\frac{2M^{2 \delta}}{M} + 
\ind{A_{e} = e}}}
\\
&\qquad \times
\sum_{\abs{a} \leq [L/2]} \; \qBB{\prod_{e \in \cal E_B \text{ leaf}} \pbb{\frac{2 M^{2 
\delta}}{M}}^{1 - a_e} \pbb{\indb{A^{-1}_{e} \prec e = A_{e}} \, 2 M^{2 \delta}}^{a_e}}
\nonumber
\\
&\;\leq\;
\qBB{\prod_{e \in \cal E_B \text{ nonleaf}} \pbb{\frac{M^{2 \delta}}{M}}^{\ind{\tau(e) \neq 
(b,0)}} \pbb{\frac{2M^{2 \delta} M^\mu}{M} + 1}}
\nonumber\\
&\qquad \times
\sum_{\abs{a} \leq [L/2]} \; \qBB{\prod_{e \in \cal E_B \text{ leaf}} \pbb{\frac{2 M^{2 \delta} 
M^\mu}{M}}^{1 - a_e} \pb{2 M^{2 \delta}}^{a_e}}
\nonumber\\
&\;\leq\; C 2^L \qBB{\prod_{e \in \cal E_B \text{ nonleaf}} \pbb{\frac{M^{2 
\delta}}{M}}^{\ind{\tau(e) \neq (b,0)}}} \pbb{1 + \frac{2 M^{2 \delta + \mu}}{M}}^{M^\mu} 
\pbb{\frac{2 M^{2 \delta} M^\mu}{M}}^{L - [L/2]} \pb{2 M^{2 \delta}}^{[L/2]}
\nonumber\\
&\;\leq\; C \pbb{\frac{C M^{4 \delta} M^\mu}{M}}^{L / 2} \prod_{e \in \cal E_B \text{ nonleaf}} 
\pbb{\frac{M^{2 \delta}}{M}}^{\ind{\tau(e) \neq (b,0)}}\,,
\nonumber
\end{align}
where we used that $2 \delta + 2 \mu < 1$, and performed the sum over $A$ trivially using the 
fact that, for each $e$, $A_e$ takes values in a set of size at most $M^\mu$.

Summarizing, we get from \eqref{result over summing over bough labels}
\begin{multline} \label{final bound on E(G)}
E_{\cal G \cup \cal G'} \;\leq\;
C \pbb{\frac{C M^{4 \delta} M^\mu}{M}}^{L/2} \prod_{e \in \cal E_B \text{ nonleaf}} 
\pbb{\frac{M^{2 \delta}}{M}}^{\ind{\tau(e) \neq (b,0)}}
\\
\times \sum_{\wt \Gamma \in \scr G_{u,u'}} \; \sum_{\b x_S \,:\, \Gamma(\b x_S) = \wt \Gamma} 
Q(\b x_S)
\prod_{\gamma \in \wt \Gamma} \E \prod_{e \in \gamma} \absB{P_{\tau(e)}\pb{\wh H_{x_{a(e)} 
x_{b(e)}}, \wh H_{x_{b(e)}
x_{a(e)}}}}\,.
\end{multline}
We can understand the first factor
in \eqref{final bound on E(G)} as follows. Each leaf
carries a factor $M^{-1}$ due to its smallness.
We estimated the combinatorial factor arising from the sum over lumpings by $M^\mu$ per leaf (which is near optimal in 
the case when  most bough edges are
leaves). Therefore, ideally, each leaf should contribute a factor
$M^{-1+\mu}$ (up to an irrelevant $M^{4\delta}$). The above argument is only able to
exploit this factor for half of the leaves; this is why we have the exponent $L/2$ instead of the desired $L$ in 
\eqref{final bound on E(G)}.
This deficiency is the main reason why the exponent of the time scale $\kappa$ is restricted
to $\kappa<1/5$ in this section.
If $L/2$ were replaced with $L$ at this point, the whole argument of Section \ref{section: boughs} would be valid up to 
time scales of order $M^{1/3}$.

\subsection{Decoupling of the graphs and the tags} \label{section: decoupling of the easy 
boughs}
The summation in \eqref{bound of E_1 in terms of E(G)} over the decorated graphs $\cal G$ 
involves summing over $G$ and $\tau_G$ under the constraint
\begin{equation*}
\sum_{e \in \cal E(G)} \deg_{\tau_G}(e)  = n\,,
\end{equation*}
and similarly for $\cal G'$. In order to sum over $G$ and $\tau_G$ separately, it is convenient 
to decouple them.
To this end, we define the degree of the boughs and the stem separately,
\begin{equation*}
\deg\pb{\cal B(G), \tau_B} \;\deq\; \sum_{e \in \cal E(\cal B(G))} \deg(e)\,,
\qquad
\deg\pb{\cal S(G), \tau_S} \;\deq\; \sum_{e \in \cal E(\cal S(G))} \deg(e)\,.
\end{equation*}
As above, we use the variable $u$ to denote $\abs{\cal E(\cal S(G))}$. Moreover, we introduce 
the variable $r = 1,2,3, \dots$ through
\begin{equation*}
\deg \pb{\cal S(G), \tau_S} \;=\; u + 2r\,.
\end{equation*}
That $r$ is an integer follows from the fact that all stem edges have odd degree (since a stem 
edge has degree 1 or 3). The variable $r$ is equal to the number of small edges (i.e.\ edges of 
type $(s,1)$ which have degree 3) in the stem $\cal E(\cal S(G))$.
The primed variables $u',r'$ are defined similarly in terms of $\cal G'$.

Let us denote by $l(G)$ and $l(G')$ the number of bough leaves in $G$ and $G'$ respectively.
Now we may write, using first \eqref{bound of E_1 in terms of E(G)} and then \eqref{final bound 
on E(G)},
\begin{align}
E_1 &\;\leq\; \sum_{n + n' \leq M^\mu} \abs{\alpha_n(t) \alpha_{n'}(t)}  \sum_{\cal G, \cal G' 
\in \fra G_\sharp} \sum_{u = 0}^{n-1} \sum_{u' = 0}^{n' - 1} \sum_{r,r' \geq 0}  E_{\cal G \cup \cal 
G'}
\notag \\
&\qquad \times
\qB{\indb{\abs{\cal E(\cal S(G))} = u} \indb{\deg(\cal S(G), \tau_S) = u + 2r} \indb{\deg 
\pb{\cal B(G), \tau_B} = n - u - 2r}} \qB{\text{primed}}
\notag \\
&\;\leq\;
\sum_{n + n' \leq M^\mu} \abs{\alpha_n(t) \alpha_{n'}(t)} \sum_{\cal G, \cal G' \in \fra G_\sharp} 
\sum_{u = 0}^{n-1} \sum_{u' = 0}^{n' - 1} \sum_{r,r' \geq 0}
\sum_{\wt \Gamma \in \scr G_{u,u'}} \; \sum_{\b x_S \,:\, \Gamma(\b x_S) = \wt \Gamma} Q(\b 
x_S)
\prod_{\gamma \in \wt \Gamma} \E \prod_{e \in \gamma} \absB{P_{\tau(e)}\pb{\wh H_{x_{a(e)} 
x_{b(e)}}, \wh H_{x_{b(e)}
x_{a(e)}}}}
\notag \\
&\qquad \times
\qB{\indb{\abs{\cal E(\cal S(G))} = u} \indb{\deg(\cal S(G), \tau_S) = u + 2r} \indb{\deg 
\pb{\cal B(G), \tau_B} = n - u - 2r}} \qB{\text{primed}}
\notag \\ \label{beginning of estimate of E_1}
&\qquad \times
C \pbb{\frac{C M^{4 \delta} M^\mu}{M}}^{\frac{l(G) + l(G')}{2}} \prod_{e \in \cal E_B \text{ 
nonleaf}} \pbb{\frac{M^{2 \delta}}{M}}^{\ind{\tau(e) \neq (b,0)}}\,,
\end{align}
where $[\text{primed}]$ means the preceding product of indicator functions with primed 
variables. The condition $u < n$ is equivalent to requiring that $\cal G \neq \cal I_n$.

Next, in \eqref{beginning of estimate of E_1} we bound
\begin{equation*}
\sum_{\wt \Gamma \in \scr G_{u,u'}} \; \sum_{\b x_S \,:\, \Gamma(\b x_S) = \wt \Gamma} Q(\b 
x_S)
\prod_{\gamma \in \wt \Gamma} \E \prod_{e \in \gamma} \absB{P_{\tau(e)}\pb{\wh H_{x_{a(e)} 
x_{b(e)}}, \wh H_{x_{b(e)}
x_{a(e)}}}} \;\leq\; \pbb{\frac{M^{2 \delta}}{M}}^{r + r'} h_{u,u'}\,.
\end{equation*}
This follows immediately from \eqref{definition of h}, \eqref{definition of V tilde}, the bound
\begin{equation*}
\absB{P_{\tau(e)}\pb{\wh H_{x_{a(e)} x_{b(e)}}, \wh H_{x_{b(e)}
x_{a(e)}}}} \;\leq\;
\begin{cases}
\abs{\wh H_{x_{a(e)} x_{b(e)}}} & \text{if } \tau(e) = (s,0)
\\
\frac{M^{2 \delta}}{M} \, \abs{\wh H_{x_{a(e)} x_{b(e)}}} & \text{if } \tau(e) = (s,1)\,,
\end{cases}
\end{equation*}
and the fact that precisely $r + r'$ stem edges have tag $(s,1)$.
Thus we get
\begin{multline} \label{bound on E_1 before decoupling}
E_1 \;\leq\;
\sum_{n + n' \leq M^\mu} \abs{\alpha_n(t) \alpha_{n'}(t)}  \sum_{r,r' \geq 0}
\pbb{\frac{M^{2 \delta}}{M}}^{r + r'} \sum_{u = 0}^{n-1} \sum_{u' = 0}^{n' - 1} h_{u,u'}
\\
\times
\sum_{\cal G, \cal G' \in \fra G_\sharp}
\qB{\indb{\abs{\cal E(\cal S(G))} = u} \indb{\deg(\cal S(G), \tau_S) = u + 2r} \indb{\deg 
\pb{\cal B(G), \tau_B} = n - u - 2r}} \qB{\text{primed}}
\\
\times
C \pbb{\frac{C M^{4 \delta} M^\mu}{M}}^{\frac{l(G) + l(G')}{2}} \prod_{e \in \cal E_B \text{ 
nonleaf}} \pbb{\frac{M^{2 \delta}}{M}}^{\ind{\tau(e) \neq (b,0)}}\,.
\end{multline}

In the next lemma we show that we can replace the condition
\begin{equation*}
\deg \pb{\cal B(G), \tau_B} = n - u - 2r \qquad \text{with} \qquad
2 \absb{\cal E(\cal B(G))} = n - u - 2r
\end{equation*}
to obtain an upper bound. Thus we decouple the dependence of the indicator function on $G$ from its dependence on the 
tagging $\tau_B$. We do this by adding bough edges of type $(b,0)$ to $G$, and by ensuring that this procedure does not 
decrease the estimate of the graph contributing to \eqref{bound on E_1 before decoupling}.

\begin{lemma} \label{lemma: decoupling}
We have that
\begin{multline} \label{decoupled estimate on E_1}
E_1 \;\leq\;
\sum_{n + n' \leq M^\mu} \abs{\alpha_n(t) \alpha_{n'}(t)}  \sum_{r,r' \geq 0}
\pbb{\frac{M^{2 \delta}}{M}}^{r + r'} \sum_{u = 0}^{n-1} \sum_{u' = 0}^{n' - 1} h_{u,u'}
\\
\times
\sum_{\cal G, \cal G' \in \fra G_\sharp}
\qB{\indb{\abs{\cal E(\cal S(G))} = u} \indb{\deg(\cal S(G), \tau_S) = u + 2r} \indb{2 
\absb{\cal E(\cal B(G))} = n - u - 2r}} \qB{\text{\rm primed}}
\\
\times
C \pbb{\frac{C M^{4 \delta} M^\mu}{M}}^{\frac{l(G) + l(G')}{2}} \prod_{e \in \cal E_B \text{\rm 
nonleaf}} \pbb{\frac{M^{2 \delta}}{M}}^{\ind{\tau(e) \neq (b,0)}}\,.
\end{multline}
\end{lemma}

\begin{proof}
Fix $n,n',r,r',u,u'$.
Note first that
\begin{equation} \label{definition of edge deficiency}
2 D \;\deq\; \deg \pb{\cal B(G), \tau_B} - 2 \absb{\cal E(\cal B(G))}\,,
\end{equation}
is a nonnegative even number. It is nonnegative because every bough edge has degree at least 
two, and even because both terms of the right-hand side of \eqref{definition of edge 
deficiency} are even.

Let $\cal G = (G, \tau_G)$ satisfy $\deg \pb{\cal B(G), \tau_G} = n - u - 2r$. We construct a 
tagged graph $\wt {\cal G} = (\wt G, \tau_{\wt G})$ as follows. If $D = 0$ then we set $\wt 
{\cal G} = \cal G$. If $D > 0$ then we denote by $v$ the stem vertex that is closest to $a(G)$ 
such that $v$ is the root of a bough. (Because $D > 0$ there is such a $v$.) We then define 
$\wt{\cal G}$ to be $\cal G$ but with the vertex $v$ replaced with a path consisting of $D$ 
bough edges, each carrying the tag $(b,0)$. (More precisely, if $e$ denotes the bough edge 
incident to $v$, we separate the vertices $a(e)$ and $v$ and join them with path of length $D$ 
carrying tags $(b,0)$). Thus, we simply lengthen a leaf by adding $D$ additional large edges.

We claim that $\wt {\cal G}$ has the following properties.
\begin{enumerate}
\item
The map $\cal G \mapsto \wt{\cal G}$ is injective.
\item
$G$ and $\wt G$ have the same number of bough leaves.
\item
$2 \absb{\cal E(\cal B(\wt G))} = n - u - 2r$.
\item
The number of small nonleaf bough edges is the same in $\cal G$ and $\wt {\cal G}$, i.e.
\begin{equation*}
\prod_{e \in \cal E(\cal B(G)) \text{ nonleaf}} \pbb{\frac{M^{2 \delta}}{M}}^{\ind{\tau_G(e) 
\neq (b,0)}} \;=\; \prod_{e \in \cal E(\cal B(\wt G)) \text{ nonleaf}} \pbb{\frac{M^{2 
\delta}}{M}}^{\ind{\tau_{\wt G}(e) \neq (b,0)}}\,.
\end{equation*}
\item
$\cal G$ and $\wt {\cal G}$ have the same tagged stem.
\end{enumerate}
Properties (ii) -- (v) are immediate from the definition of $\wt {\cal G}$. Property (i) follows from the fact that 
$\cal G$ can be reconstructed from $\wt {\cal G}$ as follows. Set $\cal G_* = (G_*, \tau_{G_*}) \deq \wt {\cal G}$. Let 
$e$ be the first bough edge of $G_*$ reached along the walk (see Figure \ref{figure: walk}) around $G_*$. If the total 
degree of the boughs of $\cal G_*$ is greater than $2 \abs{\cal E (\cal B(\wt G))}$, remove the edge $e$ from  $\cal 
G_*$.  (Note that in this case $\cal G_* \neq \cal G$, and the edge $e \in \cal E(\cal B(G_*))$ was added to $\wt G$ in 
the above construction.)
Repeat this process until the total degree of the boughs of $\cal G_*$ is equal to $2 \abs{\cal E (\cal B(\wt G))}$.  
Then $\cal G_* = \cal G$.

Constructing a tagged graph $\wt {\cal G}'$ in the same way from $\cal G'$, we bound the term 
indexed by $\cal G, \cal G'$ on the right-hand side of \eqref{bound on E_1 before decoupling} 
by the term corresponding to $\wt {\cal G}, \wt {\cal G}'$. Using the fact that the map
$\cal G \mapsto \wt {\cal G}$ is injective we may therefore bound the right-hand side of 
\eqref{bound on E_1 before decoupling} by the right-hand side of
\eqref{decoupled estimate on E_1}, writing $\cal G$ and $\cal G'$ instead of $\wt {\cal G}$ and $\wt{\cal G}'$.
\end{proof}

\subsection{Sum over taggings} \label{subsection: sum over easy bough tags}
Thanks to Lemma \ref{lemma: decoupling}, we may perform the sums over $G, G', \tau_B$, and 
$\tau_S$ separately in \eqref{decoupled estimate on E_1}. We start with the sum over $\tau_B$.
From Lemma \ref{lemma: decoupling} we get
\begin{multline} \label{sum over bough tags}
E_1 \;\leq\;
\sum_{n + n' \leq M^\mu} \abs{\alpha_n(t) \alpha_{n'}(t)}  \sum_{r,r' \geq 0}
\pbb{\frac{M^{2 \delta}}{M}}^{r + r'} \sum_{u = 0}^{n - 1} \sum_{u' = 0}^{n'-1} h_{u,u'}
\\
\times
\sum_{G, G' \in \fra W} \sum_{\tau_S}
\qB{\indb{\abs{\cal E(\cal S(G))} = u} \indb{\deg(\cal S(G), \tau_S) = u + 2r} \indb{2 
\absb{\cal E(\cal B(G))} = n - u - 2r}} \qB{\text{\rm primed}}
\\
\times
C\sum_{\tau_B} \pbb{\frac{C M^{4 \delta} M^\mu}{M}}^{\frac{l(G) + l(G')}{2}}  \prod_{e \in \cal 
E_B \text{\rm nonleaf}} \pbb{\frac{M^{2 \delta}}{M}}^{\ind{\tau_B(e) \neq (b,0)}}\,.
\end{multline}
The last line of \eqref{sum over bough tags} is bounded by
\begin{equation} \label{detail in sum over bough tags}
C \pbb{\frac{C M^{4 \delta} M^\mu}{M}}^{\frac{l(G) + l(G')}{2}} C^{l(G) + l(G')} \pbb{1 + 
\frac{C M^{2 \delta}}{M}}^{M^\mu} \;\leq\;
C \pbb{\frac{C M^{4 \delta} M^\mu}{M}}^{\frac{l(G) + l(G')}{2}}\,.
\end{equation}

Next, we sum over the stem taggings $\tau_S$ in \eqref{sum over bough tags}. The constraint $\deg(\cal S(G), \tau_S) = 
u + 2r$ means that the stem $\cal S(G) = I_u$ has $u - r$ edges with tag $(s,0)$ and $r$ edges with tag $(s,1)$. Thus we 
get in \eqref{sum over bough tags}
\begin{equation} \label{bound on sum over stem tags}
\sum_{\tau_S} \indb{\deg(\cal S(G), \tau_S) = u + 2r}
\indb{\deg(\cal S(G'), \tau_S) = u' + 2r'} \;=\;
\binom{u}{r} \binom{u'}{r'} \;\leq\; M^{\mu(r + r')}\,.
\end{equation}
Plugging \eqref{detail in sum over bough tags} and \eqref{bound on sum over stem tags} into 
\eqref{sum over bough tags} yields
\begin{multline} \label{sum over all tags}
E_1 \;\leq\;
C \sum_{n + n' \leq M^\mu} \abs{\alpha_n(t) \alpha_{n'}(t)}  \sum_{r,r' \geq 0}
\pbb{\frac{M^{2 \delta} M^\mu}{M}}^{r + r'} \sum_{u = 0}^{n - 1} \sum_{u' = 0}^{n'-1} h_{u,u'}
\\
\times
\sum_{G, G' \in \fra W}
\qB{\indb{\abs{\cal E(\cal S(G))} = u} \indb{2 \absb{\cal E(\cal B(G))} = n - u - 2r}} 
\qB{\text{\rm primed}}
\pbb{\frac{C M^{4 \delta} M^\mu}{M}}^{\frac{l(G) + l(G')}{2}}\,.
\end{multline}

\subsection{Sum over the bough graphs} \label{subsection: sum over bough graphs}
We now sum over $G, G' \in \fra W$ and complete the estimate of $E_1$.
From \eqref{sum over all tags} we get
\begin{equation} \label{easy E_1 estimated with Z}
E_1 \;\leq\; C \sum_{n + n' \leq M^\mu} \abs{\alpha_n(t) \alpha_{n'}(t)}  \sum_{u = 0}^{n - 1} 
\sum_{u' = 0}^{n'-1} h_{u,u'} \, Z_{n,u} Z_{n',u'}\,,
\end{equation}
where we defined
\begin{equation} \label{definition of easy Z}
Z_{n,u} \;\deq\;  \sum_{r \geq 0}
\pbb{\frac{M^{2 \delta} M^\mu}{M}}^r \sum_{G \in \fra W}
\indb{\abs{\cal E(\cal S(G))} = u} \indb{2 \absb{\cal E(\cal B(G))} = n - u - 2r} \pbb{\frac{C M^{4 \delta} 
M^\mu}{M}}^{\frac{l(G)}{2}}\,.
\end{equation}
The graph $G$ has a stem $\cal S(G) = I_u$ of size $u$, to which are attached boughs consisting 
together of
\begin{equation*}
k_{n,u}(r) \;\equiv\; k(r) \;\deq\; \frac{n - u - 2r}{2}
\end{equation*}
edges. Note that, because $u < n$, we always have $k(r) + r > 0$.

Next, let $s \geq 0$ be the number of boughs in $G$. We order the $s$ boughs of $G$ in some 
arbitrary manner and index them using $i = 1, \dots, s$.  Let $k_i$ be the number of edges in 
the $i$-th bough, and $l_i$ the number of leaves in the $i$-th bough.  Denote by $S_{k, l}$ the 
number of oriented, unlabelled, rooted trees with $k$ edges and $l$ leaves.  Thus we get from 
\eqref{definition of easy Z}, splitting the contributions $s = 0$ and $s \geq 1$,
\begin{equation*}
Z_{n,u} \;\leq\; \sum_{r \geq 0} \pbb{\frac{M^{2 \delta} M^\mu}{M}}^r \qBB{\ind{k(r) = 0} + 
\sum_{s \geq 1} \binom{u + 1}{s} \sum_{k_1 + \cdots + k_s = k(r)} \; \sum_{l_1 = 1}^{k_1} 
\cdots \sum_{l_s = 1}^{k_s} \prod_{i = 1}^s S_{k_i, l_i} \pbb{\frac{C M^{4 \delta} 
M^\mu}{M}}^{l_i/2}}\,,
\end{equation*}
where we sum over $k_i \geq 1$ for all $i$.
The binomial factor accounts for the locations of the roots of the boughs, which may be located 
at any of the $u+1$ stem vertices.

The number $S_{k, l}$ is known as the \emph{Naranya number}. For the convenience of the reader, we outline its key 
properties in the following short combinatorial digression. For full details see e.g.\ \cite{stanley}, p.\ 237. Denote 
by $X_{k,l}$ the set of sequences $(w_1, w_2, \dots, w_{2k})$ with $k$ elements $+1$ and $k$ elements $-1$, such that 
all partial sums are nonnegative and
\begin{equation*}
l \;=\; \absb{\h{j \st w_j = 1\,,\, w_{j+1} = -1}}\,.
\end{equation*}
The set $X_{k,l}$ parametrizes the set of oriented, unlabelled, rooted trees with $k$ edges and $l$ leaves. This
identification is the well-known bijection between such trees and Dick paths. It is constructed by walking around the 
tree, as in Figure \ref{figure: walk}, whereby at each step we add the element $+1$ to the sequence if we move away from 
the root and the element $-1$ if we move towards the root. See e.g.\ \cite{AGZ}, Chapter 1, for further details. Thus we 
have $S_{k,l} = \abs{X_{k,l}}$. In \cite{stanley}, p.\ 237, it is proved that
\begin{equation} \label{Naranya number}
S_{k,l} \;=\; \frac{1}{l} \binom{k - 1}{l - 1} \binom{k}{l - 1} \;\leq\; k^{2 l - 2}\,.
\end{equation}

Having found the expression \eqref{Naranya number} for $S_{k,l}$, we may continue our estimate of $Z_{n,u}$. We get
\begin{align*}
Z_{n,u} &\;\leq\;
\sum_{r \geq 0} \pbb{\frac{M^{2 \delta} M^\mu}{M}}^r \qBB{\ind{k(r) = 0} + \sum_{s \geq 1} 
\binom{u + 1}{s} \sum_{k_1 + \cdots + k_s = k(r)} \; \sum_{l_1 = 1}^{k_1} \cdots \sum_{l_s = 
1}^{k_s} \prod_{i = 1}^s k_i^{2 l_i - 2} \pbb{\frac{C M^{4 \delta} M^\mu}{M}}^{l_i/2}}
\\
&\;\leq\;
\sum_{r \geq 0} \pbb{\frac{M^{2 \delta} M^\mu}{M}}^r \qBB{\ind{k(r) = 0} + \sum_{s \geq 1} 
\binom{u + 1}{s} \sum_{k_1 + \cdots + k_s = k(r)} \; \prod_{i = 1}^s \frac{1}{k_i^2} \sum_{l_1 
= 1}^{k_i} \pbb{\frac{C k_i^4 M^{4 \delta} M^\mu}{M}}^{l_i/2}}
\\
&\;\leq\; \frac{C M^{2 \delta} M^\mu}{M} + \sum_{r \geq 0} \pbb{\frac{M^{2 \delta} M^\mu}{M}}^r 
\sum_{s \geq 1} \binom{u + 1}{s} \sum_{k_1 + \cdots + k_s = k(r)} \; \prod_{i = 1}^s 
\pbb{\frac{C M^{4 \delta} M^\mu}{M}}^{1/2}\,,
\end{align*}
where we used that $k_i \leq M^\mu$, $\mu + 4 \delta < 1/5$, and the fact that $r \geq 1$ if 
$k(r) = 0$.  Thus we get
\begin{align}
Z_{n,u} &\;\leq\; \frac{C M^{2 \delta} M^\mu}{M} + \sum_{r \geq 0} \pbb{\frac{M^{2 \delta} 
M^\mu}{M}}^r \sum_{s \geq 1} \binom{u + 1}{s} \binom{k(r) - 1}{s - 1} \;
\pbb{\frac{C M^{4 \delta} M^\mu}{M}}^{s/2}
\notag \\
&\;\leq\; \frac{C M^{2 \delta} M^\mu}{M} + \sum_{r \geq 0} \pbb{\frac{M^{2 \delta} M^\mu}{M}}^r 
\sum_{s \geq 1} \binom{M^\mu}{s} \binom{M^\mu}{s - 1} \;
\pbb{\frac{C M^{4 \delta} M^\mu}{M}}^{s/2}
\notag \\
&\;\leq\; \frac{C M^{2 \delta} M^\mu}{M} + C \sum_{s \geq 1} \frac{1}{M^\mu} \pbb{\frac{C M^{4 
\delta} M^{5 \mu}}{M}}^{s/2}
\notag \\ \label{bound on Z}
&\;\leq\; \frac{o(1)}{M^\mu}\,.
\end{align}

From \eqref{easy E_1 estimated with Z} and \eqref{bound on Z} we may finally conclude
\begin{equation} \label{final estimate for easy E_1}
E_1 \;\leq\; o(1) \sum_{n + n' \leq M^\mu} \abs{\alpha_n(t) \alpha_{n'}(t)}  \sum_{u = 0}^{n - 
1} \sum_{u' = 0}^{n'-1} h_{u,u'} \frac{1}{M^{2 \mu}}
\;\leq\; o(1) \sum_{n + n' \leq M^\mu}
\abs{\alpha_n(t) \alpha_{n'}(t)} \frac{M^\mu}{M^{2 \mu}}
\;=\;
o(1)\,,
\end{equation}
where we used \eqref{bound on all lumpings}, Cauchy-Schwarz, and \eqref{sum over alphas}.

\subsection{Bound on $E_2$} \label{subsection: bound on E_2}
In this final subsection, we show that $E_2$ vanishes as $W \to \infty$. Recall from 
\eqref{definition of E_2} that
\begin{equation*}
E_2 \;=\;
\sum_{n + n' \leq M^\mu}  \absb{\alpha_{n}(t) \, \alpha_{n'}(t)} \sum_x \sum_{\cal G \in \fra G_{n}^*}
\sum_{\Gamma \in \scr G(G \cup I_{n'})} \absb{V_x(\cal G \cup \cal I_{n'}, \Gamma)}
\end{equation*}
Now the preceding discussion, after setting $\cal G' = \cal I_{n'}$ and $u' = n'$ carries over 
verbatim.  The analogue of \eqref{final estimate for easy E_1} yields
\begin{equation*}
E_2 \;\leq\; o(1) \sum_{n + n' \leq M^\mu} \abs{\alpha_n(t) \alpha_{n'}(t)} \sum_{u = 0}^{n - 
1} h_{u,n'} \frac{1}{M^\mu} \;\leq\; o(1) \pBB{\sum_{n = 0}^{M^\mu} \frac{1}{M^{\mu / 2}} 
\abs{\alpha_n(t)}} \pBB{\sum_{u, n' = 0}^{M^\mu} \abs{\alpha_{n'}(t)} \frac{1}{M^{\mu/2}} 
h_{u,n'}}\,.
\end{equation*}
The first parenthesis is bounded by a constant (using Cauchy-Schwarz and \eqref{sum over alphas}).
We bound the second parenthesis using Lemma \ref{lemma: bound on stem lumpings} and \eqref{sum over alphas}:
\begin{align*}
\sum_{u, n' = 0}^{M^\mu} \abs{\alpha_{n'}(t)} \frac{1}{M^{\mu/2}} h_{u,n'} &\;\leq\; \sum_{u = 
0}^{M^\mu} \abs{\alpha_u(t)} \frac{1}{M^{\mu /2}} h_{u,u} +
\sum_{u, n' = 0}^{M^\mu} \ind{u \neq n'} \abs{\alpha_{n'}(t)} \frac{1}{M^{\mu/2}} h^*_{u,n'}
\\
&\;\leq\; C + \frac{1}{M^{\mu / 2}}\sum_{p \leq M^\mu} \sum_{u + n' = 2 p} h^*_{u,n'}
\\
&\;\leq\; C + \frac{1}{M^{\mu / 2}} M^\mu M^{\mu / 2 - 1/3 + 8 \delta}
\\
&\;\leq\; C\,.
\end{align*}
This completes the proof of Proposition \ref{proposition: boughs vanish for small kappa}.

\section{The boughs for $\kappa < 1/3$} \label{section: boughs for kappa=1/3}

In this section we extend the result of Section \ref{section: boughs} (i.e.\ Proposition \ref{proposition: boughs vanish 
for small kappa}) from $\kappa < 1/5$ to $\kappa < 1/3$.
The goal of this section is to prove the following result.

\begin{proposition} \label{proposition: boughs vanish for large kappa}
Choose $\mu$ and $\delta$ so that
\begin{equation*}
\kappa + 4 \delta \;<\; \mu \;<\; 1/3 - 8 \delta\,.
\end{equation*}
Then
\begin{equation*}
\lim_{W \to \infty} E_1 \;=\; \lim_{W \to \infty} E_2 \;=\; 0\,,
\end{equation*}
where $E_1$ and $E_2$ are defined in \eqref{definition of E_1} and \eqref{definition of E_2} 
respectively.
\end{proposition}

\subsection{Sketch of the argument} \label{sect:longsketch}
In Section \ref{section: boughs} we estimated the contribution of the boughs by summing
successively, starting from the leaves, over the label of the final vertex $x_{b(\bar e)}$
 of each bough edge $\bar e$.  
We called this process \emph{summing out} the \emph{running edge} $\bar e$ and interpreted 
it as striking $\bar e$ from 
the graph $G \cup G'$.
This summation was done for a fixed lumping which induces constraints on the values of the
 labels. In particular, we used 
the simple fact that, if the running edge $\bar e$ is lumped with
 another edge that has not yet been summed out, then
 the label of final vertex $x_{b(\bar e)}$ of $\bar e$ is fixed.  This reduces the entropy factor associated with the 
summation over $x_{b(\bar e)}$ from $M^\mu$ to 2. In general,  bigger lumps typically have
smaller contributions and this effect counterbalances the fact that the combinatorics
of the lumpings consisting of bigger lumps is larger. If, on the other hand,
a leaf $\bar e$ is not lumped with any other edge (and therefore its end-label
 $x_{b(\bar e)}$ can be summed up
without restriction), then the factor resulting from summing out $\bar e$ is small; 
see \eqref{bound on a lonely leaf}.

It turns out that the summations over all bough labels and bough lumpings (i.e.\ over
 $\b x_B$ and $A$ in the notation 
of Section \ref{section: boughs}) are not critical on the time scales we are concerned with,
 $t\sim M^{\kappa}$ where 
$\kappa <1/3$.  Hence these summations we can be done generously (see \eqref{multiplying out leaf product} where
the main contribution comes from $|a|= [L/2]$). The main reason for the restriction $\kappa<1/3$ is that the exponent 
$\kappa = 1/3$ is critical when estimating the summation over all
 stem lumpings; see \cite{erdosknowles}, Section 11.
With the method presented in this section, $\kappa <1/3$  is also critical for
the summation over the bough graphs (i.e.\ over $G \cup G'$).

As outlined in Subsection \ref{subsection: sketch of simple boughs}, the key
 difficulty when estimating the contribution 
of the boughs is to extract a sufficiently high negative power of $M$ from the
 summing out of each bough leaf. This power is needed to control the combinatorics resulting from summing over all bough 
graphs.
Ideally, each bough leaf should give a factor $M^{-1}$ (up to factors of $M^\delta$),
but in Section \ref{section: boughs} we saw that this is not true for leaves of
 type $(b,1)$. Accordingly, we were only 
able to extract a factor $M^{-1/2}$ from each bough leaf; see 
Lemma~\ref{bound on the number of bad leaves}
and \eqref{final bound on E(G)}.
More precisely, the only obstacle to extracting the full factor $M^{-1}$ from every leaf, and thus
 reaching time scales of order $M^{1/3}$, was lumps consisting exclusively of
 leaves of type $(b,1)$; see \eqref{second 
moment of small edge}.

In this section we overcome this obstacle by exploiting the fact that,
 if the running edge $\bar e$ is a leaf that is 
lumped with another edge that has not been summed out, then \emph{both} of its vertex labels,
 $x_{a(\bar e)}$ and 
$x_{b(\bar e)}$, are fixed.  In order to make use of the reduction of the entropy factor
 resulting from the fixing of 
$x_{a(\bar e)}$, we need to sum over both $x_{b(\bar e)}$ and $x_{a(\bar e)}$ when our algorithm 
tackles the leaf $\bar e$.  
The sum over $x_{a(\bar e)}$ corresponds to summing out the parent edge of $\bar e$.
 This additional summation over 
$x_{a(\bar e)}$ is clearly not possible for every leaf since several leaves 
may have a common parent or the parent of the  leaf
may be on the stem whose labels are summed over separately.
Thus the simultaneous summation over both labels of a leaf
 can only be applied once for each group of adjacent leaves (namely, to the
 \emph{free leaf} of the group; see below), 
and is not applicable at all for leaves incident to the stem (called
 \emph{degenerate leaves}; see below).
However, this deficiency is counteracted by the fact that the number of
 boughs with large groups of adjacent leaves, as well as many leaves incident to the stem, is considerably smaller than 
the number
 of arbitrary boughs (see Lemma \ref{lemma: fine 
bound on tree graphs}). This gain in the graph combinatorics is sufficient to
 compensate for the large contribution of 
groups of adjacent leaves and of leaves incident to the stem.

Roughly speaking, we gain a factor $M^{-1/2}$ from summing out each degenerate
leaf, essentially as in Section \ref{section: boughs}. Additionally, with the
double summation procedure for the free leaves, we gain the optimal factor
 $M^{-1}$ from summing out a free leaf 
together with its parent. Actually, we get the somewhat larger factor
 $ M^{-1 + \mu + 5 \delta}$, where the additional $M^\mu$ represents the entropy factor from summing over bough lumpings 
$A$ as in Section~\ref{section: boughs}. (Recall that the combinatorics
of the bough lumping, encoded in the function $e \mapsto A_e$, is overestimated
by allowing $A_e$ to by any of the $M^\mu$ edges.)
These gains have to be compared with the combinatorics of the graphs.
The number of bough graphs with a given number of free and degenerate leaves
can be easily estimated; this (with a slightly different parametrization) is
 the content of Lemma~\ref{lemma: fine bound 
on tree graphs} below. Then it would be a fairly straightforward
enumeration to sum up the contribution of all boughs; this will eventually be
done in the second part of Subsection~\ref{section: sum over bough graphs}.

Unfortunately, this simple-minded procedure is substantially complicated by a technical
hurdle. In Section~\ref{section: boughs}
the graph structure of the boughs and a simple ordering of the lumps
determined a natural order of summation over the bough edges
in such a way that the necessary size factor could be extracted from
each edge at the time it was summed out. This idea was implemented
by recursive relations of  the type $R^{(\bar e)} \;\leq\; \xi \, R^{(\sigma(\bar e))}$
 in Section~\ref{section: boughs},
where recall that $\sigma(\bar e)$ denotes the successor of $\bar e$.
 In the current situation, we have to extract a factor $M^{-1}$ from \emph{each}
 free leaf. If a free leaf $\bar e$ is bad (i.e.\ it is lumped with an edge preceding it in the order $\preceq$ but with 
no edge
 following it, written $A^{-1}_{\bar e}\prec \bar e = A_{\bar e}$;
see  Lemma \ref{bound on the number of bad leaves}), then
 the simple-minded approach 
of Section \ref{section: boughs} yields a factor of order 1 from summing out $\bar e$. 
(In fact, a key step in Section 
\ref{section: boughs} was to bound the number of such bad leaves.)

The solution is to reallocate dynamically, along the summation procedure, the weight
 factors from the running edge to edges that will
be summed out at a later stage. In other words we make sure that, if $\bar e$ is a leaf, when summing out the edge
$e \deq A^{-1}_{\bar e} \prec \bar e$ we transfer a part of the smallness resulting from summing out $e$ to the leaf 
$\bar e$. If $e$ itself is not a free leaf this is easy, because we can afford to transfer all of the smallness 
resulting from summing out $e$ (i.e.\ $M^{-1}$) to $\bar e$. If $e$ itself is a free leaf then this approach does not 
work, because the combined summing out of $e$ and $\bar e$ yields a smallness factor $M^{-1}$, which is not small enough 
to be shared among two free leaves. We solve this problem by summing out $e$ and $\sigma(e)$ in one step, as explained 
above. By choosing the order $\preceq$ appropriately, we shall ensure that
the successor $\sigma(e)$ of any free leaf $e$
 is its parent (i.e.\ $a_e= b_{\sigma(e)}$), so this double summation
amounts to summing up the labels of both vertices of $e$ at the same time. This yields a total smallness
factor $M^{-2}$, half of which is used to sum out $e$ (and hence get a small contribution),
 and the other half transferred to $\bar e$.

Thus, when summing out a free leaf $e$, we always also sum out its successor $\sigma(e)$.
 In practice, 
we need to consider all possible cases 
for $A_{e}$ and $A_{\sigma(e)}$,
 but only the cases 
where $A_{e} \notin \{e, \sigma(e)\}$
 are interesting (since otherwise $A_{e}$ is cannot be a leaf larger than $e$).
The various cases are summarized in Proposition \ref{proposition: step for hard leaves} (iii)
 (note that in the notation of Proposition \ref{proposition: step for hard leaves} the running edge is $\bar e$, which 
was denoted by $e$ in the above discussion.).

In the preceding paragraphs we neglected the role of the
entropy factor $M^\mu$ arising from the summation over the
bough lumpings $A$.  The reallocation of weights (implemented in
Proposition~\ref{proposition: step for hard leaves}) is designed in a way
that ensures that every free leaf yields a small contribution of order $M^{-1+\mu}$ (up to irrelevant $M^{O(\delta)}$ factors) \emph{after} the summation over  bough
lumpings.  Thus, we not only shift weights arising from
the summation over labels, but also entropy factors associated
with summing over lumpings. More precisely, if $e\prec A_e$ and both $e$ and $A_e$ are leaves, then we shall shift a factor $M^{-1 + \mu + 3 \delta}$ from $e$ to $A_e$ (the tiny power $3 \delta$ is unimportant, and needed only to compensate the various powers of $M^{\delta}$ that
arise in our estimates). We have seen above that a factor $M^{-1}$
is available for transfer irrespective whether $e$ is bound (its total
gain can be transferred) or $e$ is free (its total gain, $M^{-2}$,
can be shared between $e$ and $A_e$). To see why transferring the
 entropy factor $M^\mu$ is necessary, consider for instance the case where $e$ is a bound leaf and $A_e$ is a leaf.
 After the smallness $M^{-1}$ has been transferred from 
$e$ to $A_e$, we shall have to sum over $A_e$, which yields an entropy factor $M^\mu$.
 To ensure that the contribution 
of $e$ after the transfer and the summation over $A_e$ is not $O(M^\mu)$ but $O(1)$,
 we transfer only a factor $M^{-1 + \mu + O(\delta)}$ instead of $M^{-1}$  from $e$ to $A_e$.
In this way,  sufficient smallness (i.e.\ a factor $M^{-\mu}$)  remains with $e$ to compensate the entropy factor 
$M^\mu$ associated with the summation over $A_e$.

Summarizing, we transfer $M^{-1+\mu + O(\delta)}$  from $e$ to $A_e$, thus moving the
combined contribution (weight times entropy) from $e$, where it was
obtained, to $A_e$, where it is used. This transfer makes the iterative
argument cleaner; it provides a simple way of making sure that every free leaf yields
a factor $M^{-1+\mu +O(\delta)}$. Without this procedure the total weight of the leaves would be the same, but we would 
need a more complicated bookkeeping of the small factors $M^{-1+\mu +O(\delta)}$ to ensure that they arise precisely as 
often as free leaves. For instance, if $e$ is a bound leaf and $A_e$ a free leaf,
the contribution of $e$ would be $M^{-1+\mu +O(\delta)}$
(summing out $e$ and summing up for the possible lumpings $A_e$)
and the contribution of $A_e$ would only be $O(M^\delta)$.

We shall bookkeep the weights using tags as before,
but now we shall work with taggings  $\tau^{(\bar e)}$ depending on the running edge $\bar e$
that will express this reallocation process.
We shall also introduce a new tag, $(b,6)$, to record the smallness needed from each lonely leaf; it bears the weight 
that we shift around, i.e.\ $M^{-1 + \mu + 3 \delta} M^{2 \delta}$; see \eqref{definition of (p,6)}. Using it, we may 
transfer smallness from the running edge to another leaf that has become
 lonely only after all other edges in its lump 
have been summed out. In other words, the
concept of loneliness, and the smallness factor associated with it, becomes dynamical
(see Definition \ref{def:relative lonely}).
 The goal is to organize the summation over the
bough labels in such a way that one or at most two edges are summed out in one step
and, as before, recursive relations of the type  $R^{(\bar e)} \;\leq\; \xi \, R^{(\sigma(\bar e))}$
or $R^{(\bar e)} \;\leq\; \xi \, R^{(\sigma^2(\bar e))}$ keep track of the
result. The running quantity $R^{(\bar e)}$, which expresses the
size of the terms not yet summed out, will depend on the dynamical tagging,
$\tau^{(\bar e)}$. Much of the heavy notation of the following subsections
 is due to the meticulous bookkeeping of this 
dynamical process.

\subsection{Classification and ordering of leaves}
As in Section \ref{section: boughs}, we first concentrate on the term $E_1$. Our starting point 
for the proof are the bounds \eqref{bound of E_1 in terms of E(G)} and \eqref{error bound for 
fixed graphs}, where we split the vertex labels according to \eqref{splitting of labels} and 
\eqref{definition of split labels}.

Let all summation variables in \eqref{error bound for fixed graphs} be fixed. We begin by 
classifying all leaves in $\cal E_B = \cal E(\cal B(G) \cup \cal B(G'))$. To this end, we 
recall that $\cal B(G) \cup \cal B(G') = \bigcup_i T_i$ consists of disjoint boughs (rooted oriented trees) $T_i$ whose 
roots are distinct stem vertices. If all edges of a bough $T_i$ are leaves, we call $T_i$ \emph{degenerate}.  Otherwise 
we call $T_i$ \emph{nondegenerate}. Thus, all edges of a degenerate bough are incident to its root. We call the edges of 
a (non)degenerate bough \emph{(non)degenerate edges}.

Next, we assign each leaf of $\cal E_B$ to one of three categories: \emph{degenerate}, \emph{free}, or \emph{bound}. See 
Figure \ref{figure: classification of leaves}.
\begin{figure}[ht!]
\begin{center}
\includegraphics{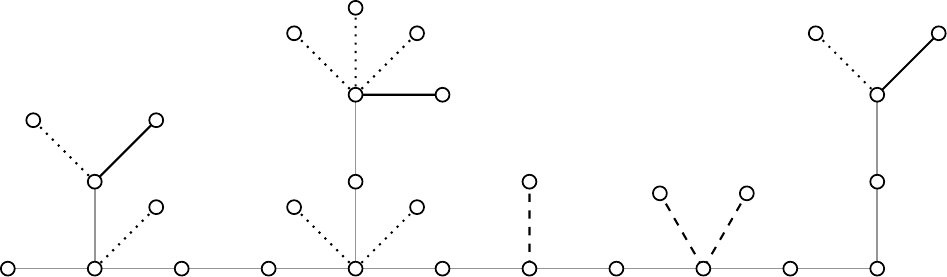}
\end{center}
\caption{A graph $G$ with five boughs, two of which are degenerate. We draw one possible choice of free leaves.  Free 
leaves are drawn with solid lines, bound leaves with dotted lines, and degenerate leaves with dashed lines.  
\label{figure: classification of leaves}}
\end{figure}

\begin{definition} \label{definition: types of leaves}
A leaf is \emph{degenerate} if it belongs to a degenerate bough. For each nondegenerate bough 
$T$, we choose a maximal subset of leaves $\cal L_T \subset \cal E(T)$ with the properties that
\begin{enumerate}
\item
no leaf in $\cal L_T$ is incident to the root of $T$,
\item
no two leaves of $\cal L_T$ are adjacent.
\end{enumerate}
We call the leaves in $\cal L_T$ \emph{free}, and the remaining leaves of $\cal E(T)$ 
\emph{bound}.
\end{definition}

Note that, by definition, each bound leaf is adjacent to a free leaf. Thus, in a nondegenerate 
bough, each group of adjacent leaves contains precisely one free leaf.

\label{page: total order 2}
Next, we introduce a total order $\preceq$ on $\cal E(G \cup G')$, as in Section \ref{section: boughs}, which will 
dictate the order of summation over the labels of the bough vertices. As in Section \ref{section: boughs}, we denote by 
$\sigma(e)$ the immediate successor of $e$ with respect to $\preceq$ (provided that $e$ is not the last edge of $\cal 
E(G \cup G')$). We impose the following conditions of $\preceq$.
\begin{enumerate}
\item
If $e$ and $e'$ are both bough edges and $e'$ is the parent of $e$ then $e \prec e'$.
\item
A free leaf immediately precedes its parent: If $e$ is a free leaf then $\sigma(e)$ is the 
parent of $e$.
\item
Nondegenerate boughs precede degenerate boughs which precede the stem: If $e$ belongs to a 
nondegenerate bough, $e'$ to a degenerate bough, and $e''$ to the stem, then $e \prec e' \prec 
e''$.
\end{enumerate}
These properties encode the plan that children will be summed up before their
parents (as in Section \ref{section: boughs}) and, additionally, in case of
the free leaves, their parents will be summed up immediately after them. Furthermore, nondegenerate edges will be summed 
up before
the degenerate ones.

It is easy to see that an order satisfying (i)--(iii) exists\footnote{Such an order can for instance be constructed as 
follows, by successively removing edges from $\cal E(G \cup G')$. Pick any nondegenerate bough $T$ (if one exists) and 
remove all of its bound leaves in an arbitrary order. Then remove from $T$ either a free leaf followed by its parent or 
a nonleaf edge, in such a way as to keep the resulting tree connected (in other words, respect the condition (i) 
above).  When all edges of $T$ have been thus removed, repeat this procedure on another nondegenerate bough.  When all 
nondegenerate boughs have been thus removed, remove all degenerate leaves in an arbitrary order.  Finally, remove all 
stem edges in an arbitrary order.} on $\cal E(G \cup G')$. We choose one such order and consider it fixed in the sequel. 

As in Section \ref{section: boughs}, we parametrize a general lumping $\Gamma \in \scr G(G \cup 
G')$ with a pair $(\wt \Gamma, A) \in \scr G_{u,u'} \times \scr A(G \cup G')$, where $u$ and 
$u'$ denote the number of edges in $\cal E(\cal S(G))$ and $\cal E(\cal S(G'))$ respectively.  
See Definitions \ref{definition of A} and \ref{definition of lump parametrization}, as well as 
Lemma \ref{lemma: lumpings parametrized}.

\subsection{Sum over nondegenerate bough labels}
We start by summing over the bough labels $\b x_B$ in \eqref{error bound for fixed graphs}, 
which we recall here for convenience:
\begin{multline} \label{error bound for fixed graphs bis}
E_{\cal G \cup \cal G'} \;\leq\; \sum_{\wt \Gamma \in \scr G_{u,u'}} \; \sum_{\b x_S \,:\, 
\Gamma(\b x_S) = \wt \Gamma} Q(\b x_S)
\\
\times
\sum_{A} \; \sum_{\b x_B} \qBB{\prod_{e \in \cal E_B} \indb{\varrho_{\b x}(e) = \varrho_{\b x}(A_e)}}
\prod_{\gamma \in \Gamma(\wt \Gamma, A)} \absbb{\E \prod_{e \in \gamma} P_{\tau(e)}\pb{\wh 
H_{x_{a(e)} x_{b(e)}}, \wh H_{x_{b(e)}
x_{a(e)}}}}\,.
\end{multline}
In this subsection we sum over the vertex labels of nondegenerate boughs.
\begin{definition}
Denote by $e_0$ the first (with respect to $\preceq$) edge of $\cal E(G \cup G')$, by $e_d$ the 
first edge of the degenerate boughs, and by $e_s$ the first stem edge. If there are no 
degenerate boughs, set $e_d = e_s$.
\end{definition}
Note that, by definition of $\preceq$, we have $e_0 \preceq e_d \preceq e_s$ (where equality is 
possible).

As in Section \ref{section: boughs}, we sum over the bough labels successively using a running 
edge $\bar e$. We start with $\bar e = e_0$, the first edge in $\cal E(G \cup G')$, and after 
each step redefine $\bar e$ to be $\sigma(\bar e)$; we stop when $\bar e = e_d$.  We use the 
abbreviations $B^{(\bar e)}$, $\b x^{(\bar e)}$, and $A^{(\bar e)}$ from Section \ref{section: 
boughs}; see \eqref{definition of B^e} and \eqref{definition of x^e and A^e}.

We shall need one additional bough tag, $(b,6)$. Its associated polynomial is defined by (we 
also recall the definition of $P_{(b,5)}$ from Section \ref{section: boughs})
\begin{equation} \label{definition of (p,6)}
P_{(b,5)}(\wh H_{xy}, \wh H_{yx}) \;=\; 2 M^{2 \delta} \sigma^2_{xy} \,, \qquad
P_{(b,6)}(\wh H_{xy}, \wh H_{yx}) \;\deq\; 2  M^{-1 + \mu + 5 \delta} \sigma^2_{xy}\,.
\end{equation}
As in Section \ref{section: boughs}, we consider first the special case that all nonleaf bough 
edges are large, i.e.\ have tag $(b,0)$.

We recall that a bough edge $e \in \cal E_B$ is \emph{lonely} whenever $A_e = e$ and there is 
no $e' \prec e$ satisfying $A_{e'} = e$. This is equivalent to saying that $e$ is the only edge 
in its lump of $\Gamma(\wt \Gamma, A)$.
We define a new tagging $\wt \tau$ through
\begin{equation*}
\wt \tau(e) \;\deq\;
\begin{cases}
\tau(e) & \text{if $e$ is not a bough leaf}
\\
(b,6) & \text{if $e$ is a lonely bough leaf}
\\
(b,5) & \text{if $e$ is a nonlonely leaf}\,.
\end{cases}
\end{equation*}
Note that this tagging $\wt\tau$ is  different from the one used in
 Section~\ref{section: boughs}. The role of the new 
tag $(b,6)$ is to encode the gain from a lonely leaf, similarly to \eqref{bound on a lonely leaf},
but the estimate will be somewhat weaker.

Using Proposition \ref{proposition: properties of graphs} (ii) and Lemma \ref{lemma: bounds on 
truncated variance} it is now easy to see that in \eqref{error bound for fixed graphs bis} we 
have the bound
\begin{equation} \label{bound using tau tilde}
\prod_{\gamma \in \Gamma(\wt \Gamma, A)} \absbb{\E \prod_{e \in \gamma} P_{\tau(e)}\pb{\wh 
H_{x_{a(e)} x_{b(e)}}, \wh H_{x_{b(e)}
x_{a(e)}}}} \;\leq\;
\prod_{\gamma \in \Gamma(\wt \Gamma, A)} \E \prod_{e \in \gamma} \absB{P_{\wt \tau(e)}\pb{\wh 
H_{x_{a(e)} x_{b(e)}}, \wh H_{x_{b(e)}
x_{a(e)}}}}\,.
\end{equation}
At this point the definition \eqref{definition of (p,6)} deserves a comment. It seems that the bound \eqref{bound using 
tau tilde} with the definition of $P_{(b,6)}$ is wasteful; indeed, \eqref{bound using tau tilde} would be correct even 
if on the right-hand side of the second equation of \eqref{definition of (p,6)} we replaced $M^{-1 + \mu + 5 \delta}$ 
with $M^{-1 + 2 \delta}$. This difference is of no consequence for our estimates, however, as the critical contribution 
to \eqref{bound using tau tilde} comes not from lonely leaves, but from leaves lumped with other leaves, as explained 
around \eqref{second moment of small edge} and \eqref{bound on a lonely leaf}. Hence the wasteful additional factor 
$M^{\mu + 3 \delta}$ is of no consequence. The factor $M^{-1 + \mu + 5 \delta}$ is designed with the transferring of 
entropy factors $M^\mu$ in mind, as explained in Section \ref{sect:longsketch}. This turns out to be the correct choice 
for the algorithm contained in Proposition \ref{proposition: step for hard leaves} below.

Next, we introduce a generalization of the notion of loneliness that is relative to the running edge.
\begin{definition}\label{def:relative lonely}
Let $\bar e \in \cal E_B$ be the running edge. We say that an edge $e \in \cal E_B$ is 
\emph{lonely in} $A^{(\bar e)}$ whenever $A_e = e$ and there is no $e' \prec e$ satisfying $e' 
\in B^{(\bar e)}$ and $A_{e'} = e$.
\end{definition}
Clearly, $e$ is lonely if and only if $e$ is lonely in $A^{(e_0)}=A$ (since $e_0$
is the first edge of $\cal E_B$).

In order to get an adequate estimate from the successive summation over the nondegenerate bough 
labels, we shall need that, at each step indexed by $\bar e$ of the recursion, every leaf that 
is lonely in $A^{(\bar e)}$ is small in the sense that it has tag $(b,6)$. Thus, we shall have to dynamically modify the 
bough tagging. To this end, let
\begin{equation*}
\tau^{(\bar e)} \;=\; \pb{\tau^{(\bar e)}(e) \,:\, e \succeq \bar e}
\end{equation*}
denote a tagging on the set of edges $\hb{e \in \cal E(G \cup G') \,:\, e \succeq \bar e}$.  The tag $\tau^{(\bar 
e)}(e)$ will indicate the actual weight of the edge $e$
after summing out all edges before $\bar e$, i.e.\ it will take into account the reallocation of
the weights.

In analogy to \eqref{definition of recursive bound}, we define
\begin{equation} \label{define of R with dynamical tag}
R^{(\bar e)}(\tau^{(\bar e)}) \;\deq\; \sum_{\b x^{(\bar e)}} \qBB{\; \prod_{e \in B^{(\bar 
e)}} \indb{\varrho_{\b x}(e) = \varrho_{\b x}(A_e)}}
\prod_{\gamma \in \Gamma(\wt \Gamma, A^{(\bar e)})} \E \prod_{e \in \gamma} 
\absb{P_{\tau^{(\bar e)}(e)}\pb{\wh H_{x_{a(e)} x_{b(e)}}, \wh H_{x_{b(e)}
x_{a(e)}}}}\,.
\end{equation}
If $\bar e$ is not a bough edge, we set $R^{(\bar e)}(\tau^{(\bar e)}) \deq 1$.

We define the initial tagging through $\tau^{(e_0)} \deq \wt \tau$.
Now \eqref{bound using tau tilde} immediately implies that in \eqref{error bound for fixed 
graphs bis} we may bound
\begin{equation} \label{starting point for hard recursion}
\sum_{\b x_B} \qBB{\prod_{e \in \cal E_B} \indb{\varrho_{\b x}(e) = \varrho_{\b x}(A_e)}}
\prod_{\gamma \in \Gamma(\wt \Gamma, A)} \absbb{\E \prod_{e \in \gamma} P_{\tau(e)}\pb{\wh 
H_{x_{a(e)} x_{b(e)}}, \wh H_{x_{b(e)}
x_{a(e)}}}} \;\leq\; R^{(e_0)}(\tau^{(e_0)})\,.
\end{equation}
We shall construct a sequence of taggings $\tau^{(e_0)}, \tau^{(\sigma(e_0))}, \tau^{(\sigma^2(e_0))}, \dots$ that 
satisfies the following property at each step $\bar e$.
\begin{itemize}
\item[$(\rr L_{\bar e})$]
If $e \in B^{(\bar e)}$ is a leaf then $\tau^{(\bar e)}(e) \in \{(b,5), (b,6)\}$. If in 
addition $e$ is lonely in $A^{(\bar e)}$ then $\tau^{(\bar e)} (e)= (b,6)$.
\end{itemize}

To start the iteration in $\bar e$, we note
that by definition of $\wt \tau$, the initial property $(\rr L_{e_0})$ holds.

The next proposition allows us to sum out a nondegenerate bough edge
 $\bar e$, i.e.\ perform one step of the iteration. 
Before stating it, we give its main idea.
If $\bar e$ is a nonleaf edge or a bound leaf, we sum over the label $x_{b(\bar e)}$ in 
\eqref{define of R with dynamical tag}. The resulting bound is given by a number $\xi$ 
resulting from the summation over $x_{b(\bar e)}$ multiplied with $R^{(\sigma(\bar e))}$. We 
use the fact that, by $(\rr L_{\bar e})$, a lonely leaf  in $A^{(\bar e)}$ is
 small in the sense that it carries a 
factor $ M^{-1 + \mu + 5 \delta}$. In order to be able to iterate this procedure, 
we need that $(\rr L_{\sigma(\bar e)})$ 
be satisfied. This is not true if $A_{\bar e}$ is a bough leaf that is lonely in
 $A^{(\sigma(\bar e))}$. We remedy this 
by introducing a factor $1 =  M^{1 - \mu - 3 \delta} \,  M^{-1 + \mu + 3 \delta}$,
 and absorbing the first term into $\xi$ 
and the second into $R^{(\sigma(\bar e))}$. This will be implemented by changing
 the tag 
$\tau^{(\bar e)}(A_{\bar e}) = (b,5)$ to $\tau^{(\sigma(\bar e))}(A_{\bar e}) = (b,6)$.

In the case that $\bar e$ is a free leaf, we have to extract a factor $ M^{-1 + \mu + 5 \delta}$
 from the summation over 
$x_{b(\bar e)}$. This is relatively easy if $\bar e$ is
lonely in $A^{(\bar e)}$; otherwise we need to distinguish further cases.
The most involved case, in particular, occurs when $A_{\bar e} \succ \bar e$ and $A_{\bar e}$
is a leaf that is lonely in  $A^{(\sigma(\bar e))}$ (i.e.\ the lump of $\bar e$ in $A^{(\bar e)}$
consists only of two elements, $\bar e$ and $A_{\bar e}$). In this case we need to change the
 tag of $A_{\bar e}$ in $\tau^{(\sigma(\bar e))}$ to $(b,6)$, as described above. 
The resulting factor $\xi$ is of order 
$M^{-1 + 2 \delta}  M^{1 - \mu - 3 \delta}$, which is much too large (here $M^{-1 + 2\delta}$
 comes from \eqref{bound on a lonely leaf}
exactly as in Lemma \ref{lemma: definition of easy xi}, while $ M^{1 - \mu - 3 \delta}$ comes
from the reallocation of this factor into $\xi$ explained in the previous paragraph).
We remedy this by exploiting the fact that the label of $a(\bar 
e)$ is also fixed by the lumping. Thus, we perform two summations in one step:
 over $x_{b(\bar e)}$ as well as over 
$x_{a(\bar e)} = x_{b(\sigma(\bar e))}$. In other words, we sum out both $\bar e$
 and its successor $\sigma(\bar e)$ at 
the same time.
Together these summations yield a factor of order $ M^{-1 + \mu + 5 \delta}$, which is small enough.

We shall need to distinguish several different cases, which leads to a somewhat
 lengthy statement of the iteration 
step.  The reason is that the worst-case estimates, which arise in the case
 $A_{\bar e} = \bar e$ (or, if $\bar e$ is a free leaf, in the cases where we do not have $A_{\bar e} \notin \{\bar e, 
\sigma(\bar e)\}$ and $A_{\sigma(\bar e)} \succ \sigma(\bar e)$) are not good enough for the following step of summing 
over lumpings (done in Subsection \ref{subsection: sum over hard lumpings}). In the case $A_{\bar e} = \bar e$,
 we shall need to compensate the poor 
estimate by the smallness of the entropy factor associated with the summation
 over $A_{\bar e}$ (namely $1$). In the case 
$A_{\bar e} \succ \bar e$, this same entropy factor is much larger
 (of the order $M^\mu$), but the estimate on $\xi$ is 
sufficiently strong to compensate for this. 
  The following Proposition
collects the estimates for the various cases.

\begin{proposition} \label{proposition: step for hard leaves}
Assume that $\bar e \in \cal E_B$ is a nondegenerate bough edge and that $(\rr L_{\bar e})$ 
holds.
\begin{enumerate}
\item
If $\bar e$ is not a leaf, then there exists a tagging $\tau^{(\sigma(\bar e))}$ satisfying $(\rr L_{\sigma(\bar e)})$ 
as well as the estimate
\begin{equation*}
R^{(\bar e)}(\tau^{(\bar e)}) \;\leq\; \xi \, R^{(\sigma(\bar e))}(\tau^{(\sigma(\bar e))})\,,
\end{equation*}
where
\begin{equation}\label{xi for nonleaf}
\xi \;\deq\;
\begin{cases}
1 & \text{if } A_{\bar e} = \bar e
\\
2 M^{-1 + 2 \delta} & \text{if } A_{\bar e} \succ \bar e \text{ and $A_{\bar e}$ is not a leaf}
\\
2 M^{-\mu - \delta} & \text{if } A_{\bar e} \succ \bar e \text{ and $A_{\bar e}$ is a leaf}\,.
\end{cases}
\end{equation}
\item
If $\bar e$ is a bound leaf, then there exists a tagging $\tau^{(\sigma(\bar e))}$ satisfying $(\rr L_{\sigma(\bar e)})$ 
as well as the estimate
\begin{equation*}
R^{(\bar e)}(\tau^{(\bar e)}) \;\leq\; \xi \, R^{(\sigma(\bar e))}(\tau^{(\sigma(\bar e))})\,,
\end{equation*}
where
\begin{equation}\label{xi for bound leaf}
\xi \;\deq\; 2 M^{- \mu - \delta}\,.
\end{equation}
\item
If $\bar e$ is a free leaf, then there exists a tagging $\tau^{(\sigma^2(\bar e))}$ satisfying $(\rr L_{\sigma^2(\bar 
e)})$ as well as the estimate
\begin{equation*}
R^{(\bar e)}(\tau^{(\bar e)}) \;\leq\; \xi \, R^{(\sigma^2(\bar e))}\pb{\tau^{(\sigma^2(\bar 
e))}}\,,
\end{equation*}
where
\begin{equation} \label{definition of xi for free leaf}
\xi \;\deq\;
\begin{cases}
2  M^{-1 + \mu + 5 \delta} & \text{if } A_{\bar e} \in \{\bar e, \sigma(\bar e)\} \text{ and } A_{\sigma(\bar e)} = 
\sigma(\bar e)
\\
2 M^{-1 + 4 \delta} & \text{if } A_{\bar e} \in \{\bar e, \sigma(\bar e)\} \text{ and } A_{\sigma(\bar e)} \succ 
\sigma(\bar e)
\\
2 M^{-1 - \mu + \delta} & \text{if } A_{\bar e} \notin \{\bar e, \sigma(\bar e)\}\,.
\end{cases}
\end{equation}
\end{enumerate}
\end{proposition}

The form of \eqref{definition of xi for free leaf} is crucial for the later summation
 over the lumpings $A_{\bar e}$ and 
$A_{\sigma(\bar e)}$. The entropy factor from each such summation is $O(1)$ if we
 have a ``hard constraint'' (i.e.\ that constrains $A_{\bar e}$ (or $A_{\sigma(\bar e)}$) to one or two edges), and 
$M^\mu$ if we have no hard constraint.
Thus, summing over the lumpings $A_{\bar e}$ and $A_{\sigma(\bar e)}$ yields an entropy factor $M^{\mu (2 - i)}$, where 
$i = 0,1,2$ is the number of hard constraints. It is easy to see from \eqref{definition of xi for free leaf} that $\xi 
M^{\mu (2 - i)}$ is always
 bounded by $M^{-1 + \mu + O(\delta)}$.

\begin{proof}[Proof of Proposition \ref{proposition: step for hard leaves}]
In order to avoid needless special cases throughout the proof, we shall always assume that $A_{\bar e}$ and 
$A_{\sigma(\bar e)}$ are leaves, unless otherwise stated. This assumption always covers the worst case scenario.

We begin with Case (i). The cases $A_{\bar e} = \bar e$ and $A_{\bar e} \succ \bar e$, $A_{\bar 
e}$ not a leaf, are dealt with exactly as in the proof of Lemma \ref{lemma: definition of easy 
xi}; see \eqref{recursive estimate for nonleaf} and \eqref{recursive bound for lumped edge}. In 
both cases we set $\tau^{(\sigma(\bar e))}(e) \deq \tau^{\bar e}(e)$ for all $e \succeq 
\sigma(\bar e)$.

If $A_{\bar e} \succ \bar e$ is a leaf we get from \eqref{recursive bound for lumped edge}
\begin{equation*}
R^{(\bar e)}(\tau^{(\bar e)}) \;\leq\; 2 M^{-1 + 2 \delta} R^{(\sigma(\bar e))}(\tau^{(\bar e)}) \;=\; 2 M^{-\mu - 
\delta} M^{-1 + \mu + 3\delta} R^{(\sigma(\bar e))}(\tau^{(\bar e)}) \;\leq\; 2 M^{- \mu - \delta} R^{(\sigma(\bar 
e))}(\tau^{(\sigma(\bar e))})\,,
\end{equation*}
where $\tau^{(\sigma(\bar e))}$ is defined as
\begin{equation} \label{redefining tau}
\tau^{(\sigma(\bar e))}(e) \;\deq\;
\begin{cases}
\tau^{(\bar e)}(e) & \text{if } e \neq A_{\bar e}
\\
(b,6) & \text{if } e = A_{\bar e}\,,
\end{cases}
\end{equation}
i.e.\ the gain of size $M^{-1 + 2\delta}$ from the summation over $x_{b(\bar e)}$ is not
exploited immediately in $\xi$, but a part of size $M^{-1+\mu + 3\delta}$ is reallocated to the tag of $A_{\bar e}$.
Here we used the bound
\begin{equation*}
M^{-1 + \mu + 3 \delta} \absb{P_{(b,5)}(\wh H_{xy}, \wh H_{yx})} \;\leq\; \absb{P_{(b,6)}(\wh H_{xy}, \wh H_{yx})}\,,
\end{equation*}
which we tacitly make use of in the rest of the proof. Note that the
second line of \eqref{redefining tau} guarantees that $(L_{\sigma(\bar e)})$ holds.

Next, we consider Case (ii). If $A_{\bar e} = \bar e$, i.e.\ $\bar e$ is lonely in
$A^{(\bar e)}$, then we use \eqref{variances sum to one} and
the fact that the property $(\rr L_{\bar e})$ implies $\tau^{(\bar e)}(\bar e) = (b,6)$. Hence, by \eqref{definition of 
(p,6)}, we get
\begin{equation*}
R^{(\bar e)}(\tau^{(\bar e)}) \;\leq\; 2 M^{-1 + \mu + 5 \delta} R^{(\sigma(\bar e))}(\tau^{(\bar e)})
 \; \le \;  2 M^{-\mu - \delta} R^{(\sigma(\bar e))}(\tau^{(\sigma(\bar e))})\,,
\end{equation*}
where we set $\tau^{(\sigma(\bar e))}(e) \deq \tau^{(\bar e)}(e)$ for $e \succeq \sigma(\bar e)$. If $A_{\bar e} \succ 
\bar e$ we define $\tau^{(\sigma(\bar e))}$ through \eqref{redefining tau} and get, as in the proof of Case (i),
\begin{equation*}
R^{(\bar e)}(\tau^{(\bar e)}) \;\leq\; 2 M^{-\mu - \delta} R^{(\sigma(\bar e))}(\tau^{(\sigma(\bar e))})\,.
\end{equation*}
Again, one can easily check that  $(L_{\sigma(\bar e)})$ holds.

Now consider Case (iii). By property (ii) of
the order $\prec$, we know that $\sigma(\bar e)$ is the parent of $\bar e$, i.e.\ $b(\sigma(\bar e))=a(\bar e)$. Note 
that in this case we sum out the two edges $\bar e$ and $\sigma(\bar e)$ in one step.

Consider first the case $A_{\bar e} = \bar e$ and $A_{\sigma(\bar e)} = \sigma(\bar e)$. Then $\tau^{(\bar e)}(\bar e) = 
(b,6)$ and $\tau^{(\bar e)}(\sigma(\bar e)) = (b,0)$ by assumption. Therefore summing over $x_{b(\bar e)}$ and 
$x_{a(\bar e)}=x_{b(\sigma(\bar e))}$ using \eqref{definition of (p,6)} and \eqref{variances sum to one}
yields
\begin{equation*}
R^{(\bar e)}(\tau^{(\bar e)}) \;\leq\;
 2 M^{-1 + \mu + 5 \delta} R^{(\sigma^2(\bar e))}\pb{\tau^{(\sigma^2(\bar e))}}\,,
\end{equation*}
where $\tau^{(\sigma^2(\bar e))}(e) \deq \tau^{(\bar e)}(e)$.
In the case $A_{\bar e} = A_{\sigma(\bar e)} = \sigma(\bar e)$ we have that $\tau^{(\bar e)}(\bar e)$ is either $(b,5)$ 
or $(b,6)$, and $\tau^{(\bar e)}(\sigma(\bar e)) = (b,0)$. Thus \eqref{definition of (p,6)} and \eqref{variances sum to 
one} imply
\begin{equation*}
R^{(\bar e)}(\tau^{(\bar e)}) \;\leq\; 2 M^{-1 + 2 \delta}
 R^{(\sigma^2(\bar e))}\pb{\tau^{(\sigma^2(\bar e))}}\,,
\end{equation*}
where $\tau^{(\sigma^2(\bar e))}(e) \deq \tau^{(\bar e)}(e)$. It is
 easy to see that $(\rr L_{\sigma^2(\bar e)})$ holds. We have covered the first line of \eqref{definition of xi for free 
leaf}.

Next, consider the case $A_{\bar e} = \bar e$ and $A_{\sigma(\bar e)} \succ \sigma(\bar e)$. Then, using that $\bar e$ 
is lonely in $A^{\bar e}$ (and hence carries a
tag $(b,6)$ by $(\rr L_{\bar e})$) and that $x_{b(\sigma(\bar e))}$ is fixed, we get the bound
\begin{multline*}
R^{(\bar e)}(\tau^{(\bar e)}) \;\leq\; 2 M^{-1 + \mu + 5 \delta} M^{-1 + 2 \delta} R^{(\sigma^2(\bar e))}(\tau^{(\bar 
e)})
\;=\; 2 M^{-1 + 4 \delta} M^{-1 + \mu + 3 \delta} R^{(\sigma^2(\bar e))}(\tau^{(\bar e)})
\\
\leq\; 2 M^{-1 + 4 \delta} R^{(\sigma^2(\bar e))}(\tau^{(\sigma^2(\bar e))})\,,
\end{multline*}
where we set
\begin{equation} \label{definition of tau 2}
\tau^{(\sigma^2(\bar e))}(e) \;\deq\;
\begin{cases}
\tau^{(\bar e)}(e) & \text{if } e \neq A_{\sigma(\bar e)}
\\
(b,6) & \text{if } e = A_{\sigma(\bar e)}\,;
\end{cases}
\end{equation}
thus, $(\rr L_{\sigma^2(\bar e)})$ holds.
Similarly, if $A_{\bar e} = \sigma(\bar e)$ and $A_{\sigma(\bar e)} \succ \sigma(\bar e)$ then  $x_{b(\bar e)}$ and 
$x_{b(\sigma(\bar e))}$ are fixed by $A_{\bar e}$ and
we get
\begin{multline*}
R^{(\bar e)}(\tau^{(\bar e)}) \;\leq\; 2 M^{-1 + 2 \delta} M^{-1 + 2 \delta} R^{(\sigma^2(\bar e))}(\tau^{(\bar e)})
\;=\; 2 M^{-1 - \mu + \delta} M^{-1 + \mu + 3 \delta} R^{(\sigma^2(\bar e))}(\tau^{(\bar e)})
\\
\leq\; 2 M^{-1 - \mu + \delta} R^{(\sigma^2(\bar e))}(\tau^{(\sigma^2(\bar e))})\,,
\end{multline*}
where we define $\tau^{(\sigma^2(\bar e))}$ through \eqref{definition of tau 2}. This covers the second line of 
\eqref{definition of xi for free leaf}.

We now turn to the last line of \eqref{definition of xi for free leaf}.
Consider the case $A_{\bar e} \notin \{\bar e, \sigma(\bar e)\}$ and $A_{\sigma(\bar e)} = \sigma(\bar e)$. Thus 
$x_{b(\bar e)}$ and $x_{b(\sigma(\bar e))}$ are fixed by $A_{\bar e}$, and we have
\begin{multline*}
R^{(\bar e)}(\tau^{(\bar e)}) \;\leq\; 2 M^{-1 + 2 \delta} M^{-1 + 2 \delta} R^{(\sigma^2(\bar e))}(\tau^{(\bar e)}) 
\;=\; 2 M^{-1 - \mu + \delta} M^{-1 + \mu + 3 \delta} R^{(\sigma^2(\bar e))}(\tau^{(\bar e)})
\\
\leq\; 2 M^{-1 - \mu + \delta} R^{(\sigma^2(\bar e))}(\tau^{(\sigma^2(\bar e))})\,,
\end{multline*}
where we set
\begin{equation*}
\tau^{(\sigma^2(\bar e))}(e) \;\deq\;
\begin{cases}
\tau^{(\bar e)}(e) & \text{if } e \neq A_{\bar e}
\\
(b,6) & \text{if } e = A_{\bar e}\,,
\end{cases}
\end{equation*}
passing part of the total gain from the double summation to $A_{\bar e}$. In particular, $(\rr L_{\sigma^2(\bar e)})$ 
holds. 

Next, consider the case $A_{\bar e} \notin \{\bar e, \sigma(\bar e)\}$ and  $A_{\sigma(\bar e)} \succ \sigma(\bar e)$.  
Assume first that $A_{\bar e}$ and $A_{\sigma(\bar e)}$ are not both bough leaves that are lonely in $A^{(\sigma^2(\bar 
e))}$.  Then we get, using again that  $x_{b(\bar e)}$ and $x_{b(\sigma(\bar e))}$ are fixed by $A_{\bar e}$, that
\begin{multline*}
R^{(\bar e)}(\tau^{(\bar e)}) \;\leq\; 2 M^{-1 + 2 \delta} M^{-1 + 2 \delta} R^{(\sigma^2(\bar e))}(\tau^{(\bar e)}) 
\;=\; 2 M^{-1 - \mu + \delta} M^{-1 + \mu + 3 \delta} R^{(\sigma^2(\bar e))}(\tau^{(\bar e)})
\\
\leq\; 2 M^{-1 - \mu + \delta} R^{(\sigma^2(\bar e))}(\tau^{(\sigma^2(\bar e))})\,,
\end{multline*}
where we set
\begin{equation*}
\tau^{(\sigma^2(\bar e))}(e) \;\deq\;
\begin{cases}
(b,6) & \text{if } e \in \h{A_{\bar e}, A_{\sigma(\bar e)}} \text{ is a leaf lonely in } 
A^{(\sigma^2(\bar e))}
\\
\tau^{(\bar e)}(e) & \text{otherwise}\,.
\end{cases}
\end{equation*}
Here we used that, in order for $(\rr L_{\sigma^2(\bar e)})$ to hold, we need at most one of 
the bough edges $A_{\bar e}$ and $A_{\sigma(\bar e)}$ to receive the tag $(b,6)$ in 
$\tau^{(\sigma^2(\bar e))}$.

Finally, we consider the case where both $A_{\bar e} \eqd e' \notin \{\bar e, \sigma(\bar e)\}$ and $A_{\sigma(\bar e)} 
\eqd e'' \succ \sigma(\bar e)$
 are bough leaves that 
are lonely in $A^{(\sigma^2(\bar e))}$. Although our goal is to sum out only the edges
 $\bar e$ and $\sigma(\bar e)$, it 
will prove necessary to first sum out all four edges $\bar e, \sigma(\bar e), e', e''$ in
 order to get a sufficiently 
strong reduction of the entropy factor. Having done this, we put back the sum over the
 end-labels of $e'$ and $e''$ (thus undoing their ``striking'' out of the graph that resulted from summing them out) to 
get the needed factor $R^{(\sigma^2(\bar e))}(\tau^{(\bar e)})$.

Thus, we sum over all the labels of the four vertices $b(\bar e), b(\sigma(\bar e)), b(e')$,
 and $b(e'')$ in the expression for $R^{(\bar e)}(\tau^{(\bar e)})$; we fix all other labels. Now it is easy to see
 that  the label $x_{b(e')}$ uniquely determines the other three labels, so
 the total entropy factor for these summations is $M$.
 Hence summing over the above four labels in the expression for $R^{(\bar e)}(\tau^{(\bar e)})$ yields the bound
\begin{equation*}
R^{(\bar e)}(\tau^{(\bar e)}) \;\leq\; (2 M^{2 \delta})^3 M^{-3} \, \wt R\,,
\end{equation*}
where $\wt R$ is the expression obtained from $R^{(\bar e)}(\tau^{(\bar e)})$ by summing
 out the edges $\bar e, \sigma(\bar e), e', e''$.
(In the estimate we used the worst case scenario, in which the edges $\bar e, e', e''$ are of type $(b,5)$ and the edge 
$\sigma(\bar e)$ of type $(b,0)$.)
 Next, it is easy
 to see that summing out the two edges $e'$ and $e''$ in the expression for $R^{(\sigma^2(\bar e))}(\tau^{(\bar e)})$
 gives the equality
\begin{equation*}
R^{(\sigma^2(\bar e))}(\tau^{(\bar e)}) \;=\; (2 M^{2 \delta})^2 \wt R\,
\end{equation*}
since at the moment when $\bar e$ is summed out, both $e'$ and $e''$
are nonlonely bough leaves, thus $\tau^{(\bar e)}(e')=\tau^{(\bar e)}(e'')=(b,5)$. Thus we find
\begin{multline*}
R^{(\bar e)}(\tau^{(\bar e)}) \;\leq\; 2 M^{-3 + 2 \delta} R^{(\sigma^2(\bar e))}(\tau^{(\bar e)})
 \;=\; 2 M^{-1 - 2 \mu - 4\delta} \pb{M^{-1 + \mu + 3\delta}}^2 R^{(\sigma^2(\bar e))}(\tau^{(\bar e)})
\\
\leq\; 2 M^{-1 - 2 \mu - 4 \delta} R^{(\sigma^2(\bar e))}(\tau^{(\sigma^2(\bar e))})\,,
\end{multline*}
where we set
\begin{equation*}
\tau^{(\sigma^2(\bar e))}(e) \;\deq\;
\begin{cases}
(b,6) & \text{if } e \in \h{A_{\bar e}, A_{\sigma(\bar e)}}
\\
\tau^{(\bar e)}(e) & \text{otherwise}\,.
\end{cases}
\end{equation*}
The factor $\pb{M^{-1 + \mu + 3\delta}}^2$ is absorbed into $R^{(\bar e)}(\tau^{(\bar e)})$ to get $R^{(\bar 
e)}(\tau^{(\sigma^2(\bar e))})$ and at the same time ensure that $(\rr L_{\sigma^2(\bar e)})$ is satisfied.
This covers the third line of \eqref{definition of xi for free leaf}, and hence concludes the proof.
\end{proof}

\subsection{Sum over degenerate bough labels} \label{section: sum over deg}
In this subsection we sum over all labels associated with degenerate bough edges, having already summed out the 
nondegenerate boughs in the previous section. We estimate $R^{(e_d)}(\tau^{(e_d)})$, where $e_d$ is the first degenerate 
bough edge of $\cal E_B$, and $\tau^{(e_d)}$ is a tagging satisfying $(\rr L_{e_d})$.  Our strategy is very similar to 
Subsection \ref{section: sum over bough labels}.  For $e \in B^{(e_d)}$ (i.e.\ $e_d \preceq e \prec e_s$) we define the 
inverse $A^{-1}_e$ by restricting Definition \ref{definition of inverse of A} to $B^{(e_d)}$.  In other words, we define 
$A_e^{-1} \deq e'$ if there exists a (necessarily unique) $e'$ satisfying $e_d \preceq e' \prec e$ and $A_{e'} = e$.  
Otherwise we set $A_e^{-1} \deq e$.

By the assumption $(\rr L_{e_d})$, every leaf $e \in B^{(e_d)}$ that is lonely in $A^{(e_d)}$ 
has tag $\tau^{(e_d)}(e) = (b,6)$. Therefore we may reproduce the proof of \eqref{bound on sum over bough labels} 
verbatim to get, for a fixed ${\b x}_S$,
\begin{equation} \label{bound on degenerate boughs}
R^{(e_d)}(\tau^{(e_d)}) \;\leq\; f(A) \, \prod_{\gamma \in \wt \Gamma} \E \prod_{e \in \gamma} 
\absB{P_{\tau(e)}\pb{\wh H_{x_{a(e)} x_{b(e)}}, \wh H_{x_{b(e)}
x_{a(e)}}}}\,,
\end{equation}
where
\begin{equation} \label{def f(A)}
f(A) \;\deq\;
\prod_{e \in B^{(e_d)} \text{ leaf}} \pbb{\frac{2 M^{2 \delta}}{M} + \indb{e = A_{e}} \, 2 M^{-1 + \mu + 5 \delta} + 
\indb{A^{-1}_{e} \prec e = A_{e}} \, 2 M^{2 \delta}}\,.
\end{equation}
Note that, unlike in \eqref{bound on sum over bough labels}, the product in \eqref{def f(A)} ranges only over leaves, 
since degenerate boughs consist only of leaves.

We may now put the estimate on both nondegenerate and degenerate boughs together. From \eqref{starting point for hard 
recursion}, Proposition \ref{proposition: step for hard leaves}, and \eqref{bound on degenerate boughs} we get
\begin{multline} \label{sum over all bough labels in hard case}
\sum_{\b x_B} \qBB{\prod_{e \in \cal E_B} \indb{\varrho_{\b x}(e) = \varrho_{\b x}(A_e)}}
\prod_{\gamma \in \Gamma(\wt \Gamma, A)} \absbb{\E \prod_{e \in \gamma} P_{\tau(e)}\pb{\wh 
H_{x_{a(e)} x_{b(e)}}, \wh H_{x_{b(e)}
x_{a(e)}}}}
\\
\leq\;
F(A) \, \prod_{\gamma \in \wt \Gamma} \E \prod_{e \in \gamma} \absB{P_{\tau(e)}\pb{\wh 
H_{x_{a(e)} x_{b(e)}}, \wh H_{x_{b(e)}
x_{a(e)}}}}
\,,
\end{multline}
where
\begin{multline} \label{expression for hard F(A)}
F(A) \;\deq\; \qBB{\; \prod_{e \in \cal E_B \text{ nonleaf}} \pB{\ind{A_e = e} + 2 M^{-1 + 2 \delta} + \indb{A_e \succ e 
\text{ is a leaf}\,} 2 M^{-\mu - \delta} }}
\qBB{\; \prod_{e \in \cal E_B \text{ bound leaf}} 2 M^{-\mu - \delta}}
\\ \times
\prod_{e \in \cal E_B \text{ free leaf}} \biggl[
\indb{A_{e} \in \{e, \sigma(e)\}} \indb{A_{\sigma(e)} = \sigma(e)} \, 2M^{-1 + \mu + 5 \delta}
+
\indb{A_{e} \in \{e, \sigma(e)\}} \, 2M^{-1 + 4 \delta}
+
 2M^{-1-\mu + \delta}
\biggr]
\\ \times
\prod_{e \in \cal E_B \text{ degenerate leaf}} \qB{2 M^{-1 + 2 \delta} + \indb{e = A_{e}} \, 2 M^{-1 + \mu + 5 \delta} + 
\indb{A^{-1}_{e} \prec e = A_{e}} \, 2 M^{2 \delta}}\,.
\end{multline}
As was advertised before the proof of
Proposition~\ref{proposition: step for hard leaves},
this estimate is designed to counterbalance the 
various smallness factors and the entropy factors for the lumping summation. 
For instance, in the second line, the prefactor is $M^{-1+(i-1)\mu+ O(\delta)}$
where $i=0,1,2$ is the number of hard constraints, so after
summation over the lumpings, each summand will be of the same order $M^{-1+\mu +5\delta}$.
The same balance can be seen among the first two summands in the last line,
while the last summand will be treated similarly to how
the second factor in  \eqref{definition of F(A)} was evaluated in 
Subsection~\ref{section: sum over easy bough lumpings}. Finally, in the product
over the nonleaves in  the first
line, only a weaker bound is available if $A_e\succ e$ is a leaf. But
this bound is strong enough to guarantee that, even after summation over $A$,
the total contribution of the nonleaves is $C^L$ instead of $C^{M^\mu}$, where $L$ is the number of leaves (see 
\eqref{bound on all nonleaves}).  Since some (small, at worst $O(M^{-\delta})$)
 gain is available for each leaf, a factor 
$C^L$ is affordable.

\subsection{General taggings and sum over bough lumpings} \label{subsection: sum over hard lumpings}
So far we assumed that all nonleaf bough edges had tag $(b,0)$ and all bough leaves
 tag $(b,1)$. As in Subsection 
\ref{section: sum over bough labels}, we split the tagging into a bough and 
stem tagging, $\tau = (\tau_B, \tau_S)$, and 
define
\begin{equation}\label{def of F(A) with tau}
F(A, \tau_B) \;\deq\; F(A) \prod_{e \in \cal E_B \text{ nonleaf}} \pb{M^{2 \delta - 
1}}^{\ind{\tau_B(e) \neq (b,0)}}\,.
\end{equation}
Then \eqref{sum over all bough labels in hard case} for arbitrary $\tau$ holds if $F(A)$ on the 
right-hand side is replaced by $F(A,\tau_B)$.

Now we sum over all bough lumpings $A \in \scr A(G \cup G')$. We start by summing over $A_e$ in \eqref{expression for 
hard F(A)} for all degenerate edges $e$. Now we proceed as in
Subsection~\ref{section: sum over easy bough lumpings}. In fact, we perform the summation of Subsection \ref{section: 
sum over easy bough lumpings} over $A$ only on the second factor of \eqref{definition of F(A)}, as the (trivial) nonleaf 
edges are treated separately. Moreover, the additional summand in each nonleaf factor $\indb{e = A_{e}} \, 2 M^{-1 + 
\mu + 5 \delta}$ is trivially accounted for.  Thus we get
\begin{multline}
\sum_{A^{(e_d)}} \; \prod_{e \in \cal E_B \text{ degenerate leaf}} \qB{2 M^{-1 + 2 \delta} + \indb{e = A_{e}} \, 2 
M^{-1 + \mu + 5 \delta} + \indb{A^{-1}_{e} \prec e = A_{e}} \, 2 M^{2 \delta}}
\\
\leq\; \pB{C M^{-1 + \mu + 7 \delta}}^{L^{(d)} /2}\,,
\end{multline}
where $L^{(d)} \equiv L^{(d)}(G \cup G')$ is the number of degenerate leaves in $G \cup G'$.
From \eqref{error bound for fixed graphs bis}, \eqref{sum over all bough labels in hard case}
 with an arbitrary tagging 
$\tau_B$,
and \eqref{def of F(A) with tau} we therefore get
\begin{multline}
E_{\cal G \cup \cal G'} \;\leq\;
\qBB{\; \prod_{e \in \cal E_B \text{ nonleaf}}\pB{1 + 2 M^{-1 + \mu + 2 \delta} + 2 L M^{-\mu -\delta}}
\pb{M^{-1 + 2 \delta}}^{\ind{\tau_B(e) \neq (b,0)}}
}
\qBB{\; \prod_{e \in \cal E_B \text{ bound leaf}} 2 M^{- \delta}}
\\ \times
\qBB{\; \prod_{e \in \cal E_B \text{ free leaf}} C M^{-1 + \mu + 5 \delta}}
\pB{C M^{-1 + \mu + 7 \delta}}^{L^{(d)} / 2}
\\ \times
\sum_{\wt \Gamma \in \scr G_{u,u'}} \sum_{\b x_S \,:\, \Gamma(\b x_S) = \wt \Gamma} Q(\b x_S) 
\prod_{\gamma \in \wt \Gamma} \E \prod_{e \in \gamma} \absB{P_{\tau(e)}\pb{\wh H_{x_{a(e)} 
x_{b(e)}}, \wh H_{x_{b(e)}
x_{a(e)}}}}\,,
\end{multline}
where $L$ is the total number of bough leaves in $G \cup G'$. Recalling that the number of 
edges of $G \cup G'$ is bounded by $M^\mu$, we find
\begin{equation} \label{bound on all nonleaves}
\prod_{e \in \cal E_B \text{ nonleaf}}\pB{1 + 2 M^{-1 + \mu + 2 \delta} + 2 L M^{-\mu - \delta}} \;\leq\; C^L\,.
\end{equation}

Recall that $L^{(d)}$ denotes the number of degenerate leaves of $G \cup G'$. Similarly, denote 
by $L^{(b)}$ the number of bound leaves of $G \cup G'$ and by $L^{(f)}$ the number of free 
leaves of $G \cup G'$.  We have proved the following result.

\begin{proposition} \label{proposition: bound on hard E_G}
For any $\cal G, \cal G' \in \fra G_\sharp$ we have
\begin{multline} \label{bound on hard E_G}
E_{\cal G \cup \cal G'} \;\leq\; \qBB{\prod_{e \in \cal E_B \text{\rm nonleaf}} \pb{M^{-1 + 2 \delta}}^{\ind{\tau(e) 
\neq (b,0)}}} \pb{C M^{- \delta}}^{L^{(b)}} \pb{C M^{-1 + \mu + 5 \delta}}^{L^{(f)}} \pb{C M^{-1 + \mu + 7 
\delta}}^{L^{(d)} / 2}
\\ \times
\sum_{\wt \Gamma \in \scr G_{u,u'}} \sum_{\b x_S \,:\, \Gamma(\b x_S) = \wt \Gamma} Q(\b x_S) 
\prod_{\gamma \in \wt \Gamma} \E \prod_{e \in \gamma} \absB{P_{\tau(e)}\pb{\wh H_{x_{a(e)} 
x_{b(e)}}, \wh H_{x_{b(e)}
x_{a(e)}}}}\,.
\end{multline}
\end{proposition}

\subsection{Decoupling and sum over the tagging}
We now proceed as in Subsection \ref{section: decoupling of the easy boughs} and prove the following result which is 
analogous to Lemma \ref{lemma: decoupling}. In order to state it, we split
\begin{equation*}
L^{(i)}(G \cup G') \;=\; L^{(i)}(G) + L^{(i)}(G')\,;
\end{equation*}
here $L^{(i)}(G)$ is the number of bough leaves of $G$ of type $i$, where $i$ can be $b$ (for ``bound''), $f$ (for 
``free''), or $d$ (for ``degenerate'').

\begin{proposition} \label{proposition: hard decoupling}
We have
\begin{multline} \label{estimate for hard decoupling}
E_1 \;\leq\;
\sum_{n + n' \leq M^\mu} \abs{\alpha_n(t) \alpha_{n'}(t)}  \sum_{r,r' \geq 0}
\pb{M^{-1 + 2 \delta}}^{r + r'} \sum_{u = 0}^{n-1} \sum_{u' = 0}^{n' - 1} h_{u,u'}
\\
\times
\sum_{\cal G, \cal G' \in \fra G_\sharp}
\qB{\indb{\abs{\cal E(\cal S(G))} = u} \indb{\deg(\cal S(G), \tau_S) = u + 2r} \indb{2 
\abs{\cal E(\cal B(G))} = n - u - 2r}} \qB{\text{\rm primed}}
\\
\times
\pb{C M^{- \delta}}^{L^{(b)}(G) + L^{(b)}(G')} \pb{C M^{-1 + \mu + 5 \delta}}^{L^{(f)}(G) + 
L^{(f)}(G')} \pb{C M^{-1 + \mu + 7 \delta}}^{L^{(d)}(G)/2 + L^{(d)}(G')/2}
\\
\times
\prod_{e \in \cal E_B \text{\rm nonleaf}} \pb{M^{-1 + 2 \delta}}^{\ind{\tau(e) \neq (b,0)}}\,.
\end{multline}
\end{proposition}

\begin{proof}
See Appendix \ref{appendix: hard decoupling}.
\end{proof}

We may now sum over the bough tagging $\tau_B$ in \eqref{estimate for hard decoupling} to get
\begin{align*}
E_1 &\;\leq\;
\sum_{n + n' \leq M^\mu} \abs{\alpha_n(t) \alpha_{n'}(t)}  \sum_{r,r' \geq 0}
\pb{M^{-1 + 2 \delta}}^{r + r'} \sum_{u = 0}^{n-1} \sum_{u' = 0}^{n' - 1} h_{u,u'}
\\
&\qquad \times
\sum_{G, G' \in \fra W}
\sum_{\tau_S}
\qB{\indb{\abs{\cal E(\cal S(G))} = u} \indb{\deg(\cal S(G), \tau_S) = u + 2r} \indb{2 
\abs{\cal E(\cal B(G))} = n - u - 2r}} \qB{\text{\rm primed}}
\\
&\qquad \times
\sum_{\tau_B}
\pb{C M^{- \delta}}^{L^{(b)}(G) + L^{(b)}(G')} \pb{C M^{-1 + \mu + 5 \delta}}^{L^{(f)}(G) + 
L^{(f)}(G')} \pb{C M^{-1 + \mu + 7 \delta}}^{L^{(d)}(G)/2 + L^{(d)}(G')/2}
\\
& \qquad\times
\prod_{e \in \cal E_B \text{\rm nonleaf}} \pb{M^{-1 + 2 \delta}}^{\ind{\tau(e) \neq (b,0)}}
\\
&\;\leq\;
\sum_{n + n' \leq M^\mu} \abs{\alpha_n(t) \alpha_{n'}(t)}  \sum_{r,r' \geq 0}
\pb{M^{-1 + 2 \delta}}^{r + r'} \sum_{u = 0}^{n-1} \sum_{u' = 0}^{n' - 1} h_{u,u'}
\\
&\qquad \times
\sum_{G, G' \in \fra W}
\sum_{\tau_S}
\qB{\indb{\abs{\cal E(\cal S(G))} = u} \indb{\deg(\cal S(G), \tau_S) = u + 2r} \indb{2 
\abs{\cal E(\cal B(G))} = n - u - 2r}} \qB{\text{\rm primed}}
\\
&\qquad \times C
\pb{C M^{- \delta}}^{L^{(b)}(G) + L^{(b)}(G')} \pb{C M^{-1 + \mu + 5 \delta}}^{L^{(f)}(G) + 
L^{(f)}(G')} \pb{C M^{-1 + \mu + 7 \delta}}^{L^{(d)}(G)/2 + L^{(d)}(G')/2}\,,
\end{align*}
where the second inequality follows analogously to \eqref{detail in sum over bough tags}.
Next, we sum over the stem tagging $\tau_S$ using \eqref{bound on sum over stem tags},
\begin{align}
E_1 &\;\leq\;
\sum_{n + n' \leq M^\mu} \abs{\alpha_n(t) \alpha_{n'}(t)}  \sum_{r,r' \geq 0}
\pb{M^{-1 + \mu + 2 \delta}}^{r + r'} \sum_{u = 0}^{n-1} \sum_{u' = 0}^{n' - 1} h_{u,u'}
\notag \\
&\qquad \times
\sum_{G, G' \in \fra W}
\qB{\indb{\abs{\cal E(\cal S(G))} = u} \indb{2 \abs{\cal E(\cal B(G))} = n - u - 2r}} 
\qB{\text{\rm primed}}
\notag \\
&\qquad \times C
\pb{C M^{- \delta}}^{L^{(b)}(G) + L^{(b)}(G')} \pb{C M^{-1 + \mu + 5 \delta}}^{L^{(f)}(G) + 
L^{(f)}(G')} \pb{C M^{-1 + \mu + 7 \delta}}^{L^{(d)}(G)/2 + L^{(d)}(G')/2}\,
\notag \\
&\;\leq\;
\sum_{n + n' \leq M^\mu} \abs{\alpha_n(t) \alpha_{n'}(t)}  \sum_{u = 0}^{n-1} \sum_{u' = 0}^{n' 
- 1} h_{u,u'}
\notag \\
&\qquad \times
\Biggl[ \sum_{r \geq 0}
\pb{M^{-1 + \mu + 2 \delta}}^r
\sum_{G \in \fra W}
\indb{\abs{\cal E(\cal S(G))} = u} \indb{2 \abs{\cal E(\cal B(G))} = n - u - 2r} \notag \\ 
\label{everything but graphs summed}
&\qquad \times C
\pb{C M^{- \delta}}^{L^{(b)}(G)} \pb{C M^{-1 + \mu + 5 \delta}}^{L^{(f)}(G)} \pb{C M^{-1 + \mu + 7 
\delta}}^{L^{(d)}(G)/2}
\Biggr] \qB{\text{\rm primed}}\,.
\end{align}

\subsection{Sum over the bough graphs} \label{section: sum over bough graphs}
Now we may sum over the graphs $G,G' \in \fra W$ in \eqref{everything but graphs summed}. A key ingredient is the 
following combinatorial estimate. Let $S_{kfb}$ be the number of nondegenerate boughs with $k$ edges, $f$ free leaves, 
and $b$ bound leaves. In other words, $S_{kfb}$ is the number of oriented, unlabelled, rooted trees with $k$ edges and 
$f + b$ leaves, such that the number of groups of adjacent leaves, excluding leaves incident to the root, is equal to 
$f$ (see Definition \ref{definition: types of leaves}).
\begin{lemma} \label{lemma: fine bound on tree graphs}
We have
\begin{equation*}
S_{kfb} \;\leq\; k^{2f - 2} \, 2^{2 f + b}\,.
\end{equation*}
\end{lemma}
\begin{proof}
We construct an arbitrary nondegenerate bough corresponding to the triple $(k,f,b)$ in two 
steps.
\begin{enumerate}
\item
We choose an oriented, unlabelled, rooted tree $T$ with $k - b$ edges and $f$ leaves such that 
no two leaves are adjacent and no leaf is incident to the root.
\item
We add $b$ leaves to $T$ by requiring that each newly added leaf be either adjacent to an 
existing leaf or incident to the root.
\end{enumerate}
Clearly, the number of possible choices for the tree $T$ in (i) is bounded by the number of oriented, unlabelled, rooted 
trees with $k$ edges and $f$ leaves. This was estimated in \eqref{Naranya number} by $k^{2 f - 2}$.

Next, let $\scr V$ denote the subset of vertices of $T$ consisting of the root of $T$ and of 
all initial vertices of the leaves of $T$. Step (ii) means that we have to add $b$ leaves to 
$T$ under the constraint that each newly added leaf be incident to a vertex of $\scr V$. The new leaves may be thought 
of as being added to a certain number, $z$, of allowed slots; each slot may receive several new leaves. The number of 
slots associated with a vertex $v \in \scr V$ is computed as follows. For a vertex $v$ of $T$, let $c_v \in \N$ denote 
the number of children of $v$ in $T$. It is easy to see that the number of slots associated with the vertex $v \in \scr 
V$ is equal to $c_v$ if $v$ is not the root and to $c_v + 1$ if $v$ is the root. In this
counting we have taken the orientation of the graph into account, i.e.\
existing edges are drawn in the plane, and the new edges emanate between them.
Thus the number of slots associated with a leaf $e$ is equal to the number of planar ``wedges'' delimited by edges 
incident to $a(e)$, whereby the two wedges on either side of $e$ count as one; see Figure \ref{figure: slots for adding 
leaves}.
\begin{figure}[ht!]
\begin{center}
\includegraphics{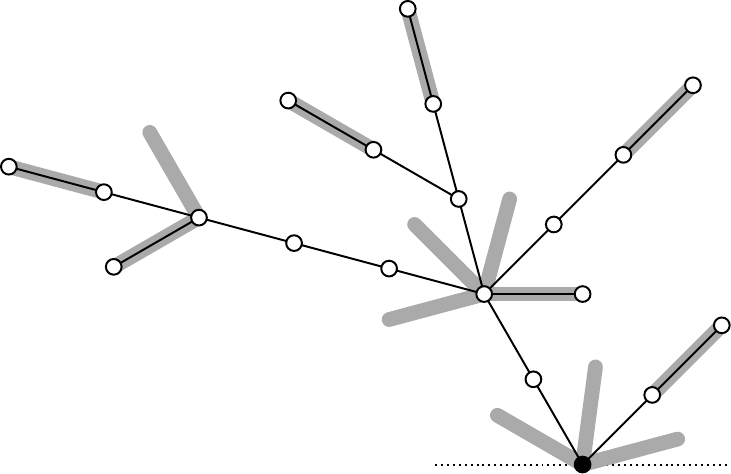}
\end{center}
\caption{A bough with $k = 19$ edges and $f = 7$ free leaves. The $z = 14$ slots for adding bound leaves are indicated 
in grey. The stem is indicated with a dotted line. \label{figure: slots for adding leaves}}
\end{figure}

Thus, the total number of slots is $z = 1 + \sum_{v \in \scr V} c_v$.
Now let us denote by $\scr V'$ the subset of vertices of $\cal V(T)$ that are not leaf vertices (or in other words the 
root together with the vertices that have degree greater than one). It is easy to see that we have
\begin{equation*}
1 + \sum_{v \in \scr V'} (c_v - 1) \;=\; f\,.
\end{equation*}
(This relation holds for any rooted tree with $f$ leaves.)
Therefore, using $\abs{\scr V} = f + 1$ and $\scr V \subset \scr V'$, we get
\begin{equation*}
z \;=\; f + 2 + \sum_{v \in \scr V} (c_v - 1) \;\leq\; 2 f + 1\,.
\end{equation*}
Therefore the number of ways to add $b$ leaves to $T$ according to (ii) is bounded by
\begin{equation*}
\binom{z + b - 1}{z - 1} \;\leq\; 2^{z + b - 1} \;\leq\; 2^{2f + b}\,.
\end{equation*}
The claim follows.
\end{proof}

We now proceed similarly to Subsection \ref{subsection: sum over bough graphs} in order to 
estimate the sum over $G, G' \in \fra W$ in \eqref{everything but graphs summed}. Let us first 
concentrate on $G$. The stem of $G$ has $u$ edges. Let $s \geq 0$ denote the number of 
nondegenerate boughs of $G$ and $q \geq 0$ the number of degenerate boughs of $G$.  The 
nondegenerate boughs consist of altogether $k \geq 0$ edges and the degenerate boughs of $m 
\geq 0$ edges.

We index the nondegenerate boughs in some arbitrary fashion using $i = 1, \dots, s$, and denote 
by $k_i \geq 1$ the number of edges in the $i$-th nondegenerate bough; we have $k_1 + \cdots + 
k_s = k$. Similarly, we index the degenerate boughs using $i = 1, \dots, q$, and denote by $m_i 
\geq 1$ the number of edges in the $i$-th degenerate bough (which is equal to the number of 
degenerate leaves in the $i$-th degenerate bough); we have $m_1 + \cdots + m_q = m$.

We use $f_i \geq 1$ to count the number of free leaves and $b_i \geq 0$ the number of bound
leaves in the $i$-th nondegenerate bough. Thus, we have the relations
\begin{equation*}
\sum_i b_i \;=\; L^{(b)}(G) \,, \qquad \sum_i f_i \;=\; L^{(f)}(G) \,,
 \qquad \sum_i m_i \;=\; m \;=\; L^{(d)}(G)\,.
\end{equation*}
Putting all of this together, we may bound the sum over $G \in \fra W$ in
 \eqref{everything but graphs summed} as
\begin{multline*}
\sum_{G \in \fra W} (\cdots) \;\leq\; \sum_{s,q} \binom{u+1}{s} \binom{u+1}{q} \sum_{k,m}
\indb{2 (k + m) = n - u - 2r} \indb{r + s + q \geq 1}
\\ \times
\sum_{k_1 + \cdots + k_s = k} \; \sum_{m_1 + \cdots + m_q = m} \; \sum_{f_1 = 1}^{k_1} \cdots 
\sum_{f_s = 1}^{k_s} \sum_{b_1 = 0}^{k_s} \cdots \sum_{b_s = 0}^{k_s} S_{k_1 f_1 b_1} \cdots 
S_{k_s f_s b_s} (\cdots)\,.
\end{multline*}
Here the sum $\sum_{k_1+ \cdots + k_s = k}$ is understood to mean
 $\ind{k = 0}$ if $s = 0$ (see \eqref{precise definition of sum} below); similarly for the sum $\sum_{m_1 + \cdots + m_q 
= m}$.
The binomial factors arise from the choice of the root vertices of 
the boughs (each bough root may be one of the $u+1$ 
stem vertices). The constraint $r + s + q \geq 1$ is an immediate 
consequence of the constraint that if $r = 0$ on the 
right-hand side of \eqref{everything but graphs summed} then the
 second indicator function implies $2(k + m) = n - u 
\geq 1$,
and that $k\ge 1$ implies $s\ge 1$, while $m\ge 1$ implies $q\ge 1$.
 This is simply a restatement of the fact that we must have either a bough
 ($s \geq 1$ or $q \geq 1$)
or a small edge in the stem ($r \geq 1$), in accordance with
 the original restriction $\cal G \in \fra G^*$.

We also introduce the analogous primed quantities associated with $G'$. Thus we get from 
\eqref{everything but graphs summed}
\begin{equation} \label{hard E_1 in terms of Z}
E_1 \;\leq\;
C \sum_{n + n' \leq M^\mu} \abs{\alpha_n(t) \alpha_{n'}(t)} \sum_{u = 0}^{n-1} \sum_{u' = 
0}^{n' - 1} h_{u,u'} \, Z_{n,u} Z_{n',u'}\,,
\end{equation}
where we defined
\begin{multline*}
Z_{n,u} \;\deq\;
\sum_{r, s, q, k, m \geq 0} \binom{u+1}{s} \binom{u+1}{q} \indb{2 (k + m) = n - u - 2r} \indb{r 
+ s + q \geq 1} \sum_{k_1 + \cdots + k_s = k} \; \sum_{m_1 + \cdots + m_q = m}
\\ \times
\sum_{f_1 = 1}^{k_1} \dots \sum_{f_s = 1}^{k_s} \sum_{b_1 = 0}^{k_1} \dots \sum_{b_s = 0}^{k_s}
S_{k_1 f_1 b_1} \cdots S_{k_s f_s b_s}
\\ \times
\pb{M^{-1 + \mu + 2 \delta}}^r
\pb{C M^{- \delta}}^{b_1 + \cdots + b_s} \pb{C M^{-1 + \mu + 5 \delta}}^{f_1 + \cdots + f_s}
\pb{C M^{-1 + \mu + 7 \delta}}^{m/2}\,.
\end{multline*}

Now Lemma \ref{lemma: fine bound on tree graphs} yields
\begin{equation*}
\sum_{f_i = 1}^{k_i} \sum_{b_i = 0}^{k_i} S_{k_i f_i b_i} \pb{C M^{- \delta}}^{b_i} \pb{C M^{-1 + \mu + 5 \delta}}^{f_i} 
\;\leq\; \frac{1}{k_i^2}
\sum_{f_i = 1}^{k_i} \sum_{b_i = 0}^{k_i} \pb{2 C M^{- \delta}}^{b_i} \pb{4 C k_i^2 M^{-1 + \mu + 5 \delta}}^{f_i} 
\;\leq\; C M^{-1 + \mu + 5 \delta}\,,
\end{equation*}
where we used that $k_i \leq M^\mu$ and $\mu < \frac{1}{3} - \frac{5}{3} \delta$.
This is one stage where the restriction $\mu < \frac{1}{3}$ is crucial.

Therefore,
\begin{multline*}
Z_{n,u} \;\leq\;
\sum_{r, s, q, k, m \geq 0} (M^\mu)^{s + q}  \indb{2 (k + m) = n - u - 2r} \indb{r + s + q \geq 
1}
\\ \times
\pb{M^{-1 + \mu + 2 \delta}}^r
\pb{C M^{-1 + \mu + 5 \delta}}^{s}
\pb{C M^{-1 + \mu + 7 \delta}}^{m/2}
\sum_{k_1 + \cdots + k_s = k} \;
\sum_{m_1 + \cdots + m_q = m} 1\,,
\end{multline*}
where we used that $u + 1 \leq M^\mu$.

Next, we estimate
\begin{equation} \label{precise definition of sum}
\sum_{k_1 + \cdots + k_s = k} 1 \;=\; \ind{s = 0} \ind{k = 0} + \ind{s \neq 0} \ind{s \leq k}
\binom{k - 1}{s - 1} \;\leq\; (M^\mu)^{[s - 1]_+}
I(s,k)
\,,
\end{equation}
where we defined
\begin{equation*}
I(s,k) \;\deq\; \ind{s = 0} \ind{k = 0}  + \ind{s \neq 0} \ind{s \leq k}\,.
\end{equation*}
Similarly, we have
\begin{equation*}
\sum_{m_1 + \cdots + m_q = m} 1 \;=\; \ind{q = 0} \ind{m = 0} + \ind{q \neq 0} \ind{q \leq m}
\binom{m - 1}{q - 1} \;\leq\; 2^m I(q,m)\,.
\end{equation*}
Using that $q \leq m$ and consequently
\begin{equation*}
2^m (M^\mu)^q
\pb{C M^{-1 + \mu + 7 \delta}}^{m / 2} \;\leq\; \pb{C M^{-1 + 3 \mu + 7 \delta}}^{m / 2}\,,
\end{equation*}
we therefore get the following result.
\begin{proposition} \label{proposition: sum over almost everything}
We have the bound
\begin{equation*}
E_1 \;\leq\;
C \sum_{n + n' \leq M^\mu} \abs{\alpha_n(t) \alpha_{n'}(t)} \sum_{u = 0}^{n-1} \sum_{u' = 
0}^{n' - 1} h_{u,u'} \, Z_{n,u} Z_{n',u'}\,,
\end{equation*}
where
\begin{multline} \label{main estimate on Z}
Z_{n,u} \;\leq\;
\sum_{r, s, q, k, m \geq 0} (M^{- \mu})^{\ind{s \geq 1}} \indb{2 (k + m) = n - u - 2r} \indb{r + s + q \geq 1} I(s,k) 
I(q,m)
\\ \times
\pb{M^{-1 + \mu + 2 \delta}}^r
\pb{C M^{-1 + 3 \mu + 5 \delta}}^{s}
\pb{C M^{-1 + 3 \mu + 7 \delta}}^{m / 2}\,.
\end{multline}
\end{proposition}

Proposition \ref{proposition: sum over almost everything} is the main
 result of this subsection. Note that
the restriction $\mu < \frac{1}{3}$
will be crucial in performing the summations in \eqref{main estimate on Z}
 over both $s$ and $m$; the summation over $r$ is
less critical.
This is an indication that both the number of boughs and their combinatorial complexity are \emph{critically}
 compensated by the smallness of the lonely leaves.
The geometric series in $s, m, r$ are the key
 ingredients of the complicated estimate \eqref{main estimate on Z}.
The other two summation variables, $k$ and $q$, are controlled
by these variables, so the sum is finite.
To ensure that it is actually small, the
rather baroque 
collection of indicator functions is necessary.
They make sure that at least one negative $M$-power is gained from one of the factors, as we shall see in the next 
subsection.

\subsection{Conclusion of the estimate}
What remains is an elementary and only moderately enlightening estimate of $E_1$ using 
\eqref{main estimate on Z}.

\begin{proposition} \label{proposition: hard Z}
We have
\begin{equation*}
Z_{n,u} \;\leq\; \pb{C M^{-1 + 3 \mu+ 7 \delta}}^{(n - u)/4} + o(1) M^{-\mu}\,.
\end{equation*}
\end{proposition}

\begin{proof}
In \eqref{main estimate on Z} we bound the indicator function
\begin{equation*}
\indb{r + s + q \geq 1} \;\leq\; \ind{r = s = 0} \ind{q \geq 1} + \ind{r \geq 1} + \ind{s \geq 
1}\,,
\end{equation*}
which yields the bound $Z_{n,u} \leq Z'_{n,u} + Z''_{n,u} + Z'''_{n,u}$ in self-explanatory 
notation.

If $r + s = 0$ in \eqref{main estimate on Z} then $I(s,k) = 1$
 implies $s=0$ and hence $k = 0$, so that we get
\begin{equation*}
Z_{n,u}' \;\leq\;
\sum_{q, m \geq 1} \indb{2 m = n - u} \ind{q \leq m}
\pb{C M^{-1 + 3 \mu + 7 \delta}}^{m / 2} \;\leq\;
\pb{C M^{-1 + 3 \mu + 7 \delta}}^{(n - u)/4}
\,.
\end{equation*}

Next, we get
\begin{align*}
Z_{n,u}'' &\;\leq\;
\sum_{r, s, q, k, m \geq 0} (M^{-\mu})^{\ind{s \geq 1}} \ind{r \geq 1} \indb{2 (k + m) = n - u 
- 2r} I(s,k) I(q,m)
\\
&\qquad \times
\pb{M^{-1 + \mu + 2 \delta}}^r
\pb{C M^{-1 + 3 \mu + 5 \delta}}^{s}
\pb{C M^{-1 + 3 \mu + 7 \delta}}^{m / 2}
\\
&\;\leq\;
\sum_{r \geq 1} \pb{M^{-1 + \mu + 2 \delta}}^r \sum_{s \geq 0} \pb{C M^{-1 + 3 \mu + 5 \delta}}^{s}
\sum_{m \geq 0} m \pb{C M^{-1 + 3 \mu + 7 \delta}}^{m / 2}
\sum_{k \geq 0}  \indb{2 (k + m) = n - u - 2r}
\\
&\;\leq\; C M^{-1 + \mu + 2 \delta}\,.
\end{align*}
Similarly we get
\begin{align*}
Z_{n,u}''' &\;\leq\; M^{-\mu}
\sum_{r \geq 0} \pb{M^{-1 + \mu + 2 \delta}}^r \sum_{s \geq 1} \pb{C M^{-1 + 3 \mu + 5 \delta}}^{s}
\sum_{m \geq 0} m \pb{C M^{-1 + 3 \mu + 7 \delta}}^{m / 2}
\sum_{k \geq 0}  \indb{2 (k + m) = n - u - 2r}
\\
&\;\leq\; o(1) M^{- \mu}\,.
\end{align*}
This concludes the proof.
\end{proof}

From \eqref{hard E_1 in terms of Z} and using Proposition \ref{proposition: hard Z} we find
\begin{multline} \label{preintermediate estimate for E_1}
E_1 \;\leq\;
C \sum_{n + n' \leq M^\mu} \abs{\alpha_n(t) \alpha_{n'}(t)} \sum_{u = 0}^{n-1} \sum_{u' = 0}^{n' - 1} h_{u,u'}
\\ \times
\qB{\pb{C M^{-1 + 3 \mu + 7 \delta}}^{(n - u)/4} + o(1) M^{-\mu}}
\qB{\pb{C M^{-1 + 3 \mu + 7 \delta}}^{(n' - u')/4} + o(1) M^{-\mu}}\,.
\end{multline}
Setting $v = n - u$ and $v = n' - u'$ yields
\begin{multline} \label{intermediate estimate for hard E_1}
E_1 \;\leq\; C \sum_{v,v' \geq 1} \ind{v + v' \leq M^\mu}
\qB{\pb{C M^{-1 + 3 \mu + 7 \delta}}^{v/4} + o(1) M^{-\mu}}
\qB{\pb{C M^{-1 + 3 \mu + 7 \delta}}^{v'/4} + o(1) M^{-\mu}}
\\
\times
\sum_{n \geq v} \sum_{n' \geq v'} \ind{n + n' \leq M^\mu} \, \abs{\alpha_n(t) \alpha_{n'}(t)} 
h_{n - v, n' - v'} \,.
\end{multline}
The second line of \eqref{intermediate estimate for hard E_1} is bounded by
\begin{equation*}
\sum_{n + n' \leq M^\mu} \abs{\alpha_n(t) \alpha_{n'}(t)} h_{n - v, n' - v'}^* + \sum_{n \leq 
M^\mu} \abs{\alpha_n(t) \alpha_{n -v + v'}(t)} h_{n - v, n - v} \;\leq\; C\,,
\end{equation*}
by \eqref{summing over alphas}, \eqref{bound on all lumpings}, and \eqref{sum over alphas}.
Therefore \eqref{intermediate estimate for hard E_1} yields
\begin{equation*}
E_1 \;\leq\; o(1)\,.
\end{equation*}

\subsection{Bound on $E_2$}
Finally, we outline how to bound $E_2$; the argument is almost identical to Subsection 
\ref{subsection: bound on E_2}. The preceding analysis carries over trivially to $E_2$, the 
only modification being that $\cal G' = \cal I_{n'}$ and $u' = n'$, i.e.\ we only have boughs 
in $\cal G$.
The analogue of \eqref{preintermediate estimate for E_1} yields
\begin{align*}
E_2 &\;\leq\;
C \sum_{n + n' \leq M^\mu} \abs{\alpha_n(t) \alpha_{n'}(t)} \sum_{u = 0}^{n-1}  h_{u,n'} 
\qB{\pb{C M^{-1 + 3 \mu + 7 \delta}}^{n - u} + o(1) M^{-\mu}}\,.
\end{align*}
Now we proceed exactly as in Subsection \ref{subsection: bound on E_2} and get $E_2 = o(1)$.  
Hence the proof of Proposition \ref{proposition: boughs vanish for large kappa} is complete.

\section{Proof of Theorem \ref{theorem: bound on lambda max}} \label{section: bound on lambda max}

The main ingredient in the proof of Theorem \ref{theorem: bound on lambda max} is the following estimate.

\begin{proposition} \label{proposition: estimate of trace}
Let $H$ be as in Theorem \ref{theorem: bound on lambda max} and $\wh H$ the matrix whose entries are truncated as in 
\eqref{truncation of matrix elements}. Let $\kappa < 1/3$. Then there is a constant $C_\kappa$, depending on $\kappa$, 
such that
\begin{equation*}
\E \tr U_n(\wh H/2) \;\leq\; C_\kappa N
\end{equation*}
for all $n \leq M^{\kappa}$. If $n$ is odd then $\E \tr U_n(\wh H / 2) = 0$.
\end{proposition}

\begin{proof}
The proof is a relatively straightforward consequence of the proof of Theorem \ref{theorem: main result}. The claim 
about odd $n$ is immediate since $U_n$ is odd for odd $n$. Using Proposition \ref{proposition: U_n as a sum over graphs} 
we write
\begin{equation*}
\E \tr U_n(\wh H/2) \;=\; \sum_x \sum_{\cal G \in \fra G_n} \E \fra V_{xx}(\cal G)\,.
\end{equation*}
The right-hand side is represented graphically, as in Section \ref{section: path expansion}, by a single stem whose ends 
are joined so as to produce a closed loop, to which are attached a family of boughs. Now the estimates of Sections 
\ref{section: lumping} -- \ref{section: boughs for kappa=1/3} carry over and yield the claim. This is a consequence of 
the following observations.
\begin{enumerate}
\item
Assume first that $\{\sigma_{xy}\}$ defines a band matrix, as in Section \ref{section: setup}.
The value associated with a graph $\cal G$ and lumping $\Gamma$ of the edges of $\cal G$ is equal to $\sum_x {V'_x}(\cal 
G, \Gamma)$, where $V'_x$ is given by $V_x$ (see \eqref{value of general lumping}) with one additional indicator 
function that constrains all stem vertices of $\cal G$, with the exception of its root, to be nonbacktracking. In the 
graph on the right-hand side of Figure \ref{figure: identifying vertices} this may be viewed as making the vertex $n$ 
black.

It is now straightforward that all estimates from Sections \ref{section: lumping} -- \ref{section: boughs for kappa=1/3} 
carry over; in fact, the additional indicator function in $V'_x$ results in somewhat smaller bounds.

\item
In order to extend the claim to the more general random matrices as defined in the statement of Theorem \ref{theorem: 
bound on lambda max}, we observe that the $\ell^1$-$\ell^\infty$-type estimates that form the backbone of Sections 
\ref{section: bare stem} -- \ref{section: boughs for kappa=1/3} remain unchanged. The spatial structure of the band 
defined by a shape function $f$ was used in two places: first, in the analysis of the ladder diagrams; second, in the 
ensuing heat kernel bounds on the right-hand side of \eqref{sum over all non-ladder pairings}.
As we are only interested in the trace (which corresponds to summing over all vertex labels), we do not need the precise 
spatial information associated with the ladders, merely a bound on the $\ell^1$-norm of their contribution (in fact, it 
is a simple matter to check that under the additional nonbacktracking condition the ladder pairings do not even 
appear).  Moreover, dropping the detailed heat kernel bounds in \eqref{sum over all non-ladder pairings} yields the 
bound
\begin{equation*}
\sum_{n+n' = 2p} h_{n,n'} \;\leq\; C_\kappa
\end{equation*}
instead of \eqref{bound on h star}. See the remarks after \eqref{sum over all non-ladder pairings}.
\qedhere
\end{enumerate}
\end{proof}

We may now complete the proof of Theorem \ref{theorem: bound on lambda max}. We need the following elementary results on 
Chebyshev polynomials.

\begin{lemma} \label{lemma: properties of Chebyshev}
Let $n$ be even.
\begin{enumerate}
\item
For $\xi \in \R$ we have $U_n(\xi) \geq -(n+1)$.
\item
$U_n(1 + \xi)$ is increasing for $\xi \geq 0$.
\item
For $\xi \in [0,1]$ we have
$U_n(1 + \xi) \geq \ee^{n \sqrt{\xi}}$.
\end{enumerate}
\end{lemma}
\begin{proof}
If $\xi \in [-1,1]$ the claim (i) is easily seen from either \eqref{definition of Chebyshev} or the recursion relation 
\eqref{recursion for Chebyshev}. For $\xi \geq 1$, the claim (i) follows immediately from the formula
\begin{equation} \label{Chebyshev for large arguments}
U_n(\cosh \zeta) \;=\; \frac{\sinh (n+1) \zeta}{\sinh \zeta}\,,
\end{equation}
itself a straightforward consequence of \eqref{definition of Chebyshev} and analyticity.

The claim (ii) follows from \eqref{Chebyshev for large arguments}.

In order to prove the claim (iii), pick $\zeta \geq 0$ such that $1 + \xi = \cosh \zeta$. Using \eqref{Chebyshev for 
large arguments} we get for $\xi \in [0,1]$
\begin{equation*}
U_n(1 + \xi) \;=\; \frac{\sinh (n+1) \zeta}{\sinh \zeta} \;\geq\; \ee^{n \zeta} \;\geq\; \ee^{n \sqrt{\xi}}\,. \qedhere
\end{equation*}
\end{proof}

Denote by $\wh \lambda_{\rr{max}}$ the largest eigenvalue of $\wh H$.
Then we get for $\xi \in [0,1]$, using Lemma \ref{lemma: properties of Chebyshev} (ii) and (iii),
\begin{equation*}
\P \pb{\wh \lambda_{\rr{max}} \geq 2 + 2 \xi} \;\leq\; \P \pB{U_n(\wh \lambda_{\rr{max}}/2) \geq U_n(1 + \xi)} \;\leq\; 
\frac{\E U_n(\wh \lambda_{\rr{max}}/2)}{\ee^{n \sqrt{\xi}}}\,.
\end{equation*}
Thus Lemma \ref{lemma: properties of Chebyshev} (i) and Proposition \ref{proposition: estimate of trace} yield
\begin{equation*}
\P \pb{\wh \lambda_{\rr{max}} \geq 2 + 2 \xi} \;\leq\; \frac{\E \tr U_n(\wh H/2) + N (n + 1)}{\ee^{n \sqrt{\xi}}} 
\;\leq\; \frac{N (C_\kappa + n + 1)}{\ee^{n \sqrt{\xi}}}\,,
\end{equation*}
for all $n \leq M^\kappa$. Setting $\xi = M^{-2/3 + \epsilon} / 2$ and invoking the bound \eqref{probability bound for 
entry cutoff} gives
\begin{equation}
\P \pbb{\lambda_{\rr{max}} \geq 2 + \frac{1}{M^{2/3 - \epsilon}}} \;\leq\; \P \pbb{\wh \lambda_{\rr{max}} \geq 2 + 
\frac{1}{M^{2/3 - \epsilon}}}
+
C N^2 \ee^{- M^{\alpha \delta}}
\;\leq\;
\frac{N (C_\kappa + M^\kappa + 1)}{\exp \p{M^{\epsilon / 2 - (1/3 - \kappa)}}}
+
C N^2 \ee^{- M^{\alpha \delta}}
\,.
\end{equation}
Choosing $\kappa$ satisfying $1/3 - \kappa = \epsilon /3$ and $\delta = \epsilon / 37$ (see \eqref{condition on delta})
completes the proof.
\appendix

\section{Proof of Proposition \ref{proposition: finite speed of propagation}} \label{appendix: speed of propagation}

\subsection{Control of the spread of time evolution} Abbreviate $\psi_t \deq \ee^{- \ii t H/2} 
\delta_0$.  We start by estimating $\scalar{\psi_t}{\abs{x}^2 \psi_t}$.
Using $\ii \partial_t \psi_t = H \psi_t/2$ we find
\begin{equation*}
\partial_t \scalar{\psi_t}{\abs{x}^2 \psi_t} \;=\;
\frac{\ii}{2} \scalarb{\psi_t}{\comb{H}{\abs{x}^2} \psi_t}
\;=\; \frac{-\ii}{2} \sum_{x,y} H_{xy} (\abs{x}^2 - \abs{y}^2) \ol \psi_t(x) \psi_t(y)\,.
\end{equation*}
This gives
\begin{align}
\absb{\partial_t \scalar{\psi_t}{\abs{x}^2 \psi_t}} &\;\leq\; \frac{1}{2} \sum_{x,y} 
\abs{H_{xy}} \absb{\abs{x}^2 - \abs{y}^2} \abs{\psi_t(x)} \abs{\psi_t(y)}
\notag \\
&\;\leq\; \frac{1}{2}\sum_{x,y} \abs{H_{xy}} \abs{x - y} (\abs{x} + \abs{y}) \abs{\psi_t(x)} 
\abs{\psi_t(y)}
\notag \\ \label{first step in control of spreading}
&\;\leq\; \sum_{x,y} \frac{\abs{H_{xy}} \abs{x - y}}{\avg{y}^{2\epsilon}} \, \abs{x} \, 
\abs{\psi_t(x)} \, \avg{y}^{2\epsilon} \abs{\psi_t(y)}\,,
\end{align}
for any $\epsilon > 0$. Here we defined
\begin{equation*}
\avg{y} \;\deq\; \sqrt{1 + \abs{y}^2}\,.
\end{equation*}

Next, we recall Schur's inequality, valid for any matrix $A$,
\begin{equation} \label{Holmgren-Schur}
\norm{A} \;\leq\; \pbb{\sup_x \sum_y \abs{A_{xy}}}^{1/2} \pbb{\sup_y \sum_x 
\abs{A_{xy}}}^{1/2}\,.
\end{equation}
Thus we get from \eqref{first step in control of spreading}, for any $\zeta > 0$,
\begin{align*}
\absb{\partial_t \scalar{\psi_t}{\abs{x}^2 \psi_t}} &\;\leq\; \pbb{\sup_x \sum_y 
\frac{\abs{H_{xy}} \abs{x - y}}{\avg{y}^{2\epsilon}}}^{1/2} \pbb{\sup_x \sum_y 
\frac{\abs{H_{xy}} \abs{x - y}}{\avg{x}^{2\epsilon}}}^{1/2} \scalarb{\psi_t}{\abs{x}^2 
\psi_t}^{1/2} \scalarb{\psi_t}{\avg{x}^{4 \epsilon} \psi_t}^{1/2}
\\
&\;\leq\; B \pbb{\zeta \scalarb{\psi_t}{\abs{x}^2 \psi_t} + \frac{1}{\zeta} 
\scalarb{\psi_t}{\avg{x}^{4 \epsilon} \psi_t}}\,,
\end{align*}
where we defined
\begin{equation*}
B \;\deq\; \pbb{\sup_x \sum_y \frac{\abs{H_{xy}} \abs{x - y}}{\avg{y}^{2\epsilon}}}^{1/2} 
\pbb{\sup_x \sum_y \frac{\abs{H_{xy}} \abs{x - y}}{\avg{x}^{2\epsilon}}}^{1/2}\,.
\end{equation*}

In order to estimate $B$ we observe that the inequality $\avg{x + y} \leq 2 \avg{x} \avg{y}$ 
implies
\begin{equation*}
\sup_x \sum_y \frac{\abs{H_{xy}} \abs{x - y}}{\avg{x}^{2\epsilon}} \;\leq\;
2^{2\epsilon} \sup_x \sum_y \frac{\abs{H_{xy}} \avg{x - y}^{1 + 2 
\epsilon}}{\avg{y}^{2\epsilon}}\,.
\end{equation*}
Thus we get
\begin{equation} \label{bound on B}
B \;\leq\; 2^{\epsilon} \sup_x \sum_y \frac{\abs{H_{xy}} \avg{x - y}^{1 + 2 
\epsilon}}{\avg{y}^{2 \epsilon}} \;\leq\; 2^{\epsilon} \sup_x \sum_y \frac{\abs{A_{xy}} 
\sigma_{xy} \avg{x - y}^{1 + 3 \epsilon}}{\avg{x - y}^\epsilon \avg{y}^{2 \epsilon}}
\;\leq\;
2^{2 \epsilon} \sup_x \sum_y \frac{\abs{A_{xy}}}{\avg{x}^\epsilon \avg{y}^\epsilon}
\sigma_{xy} \avg{x - y}^{1 + 3 \epsilon}\,.
\end{equation}

Next, for $u \geq 1$ define
\begin{equation*}
\Omega_u \;\deq\; \hbb{\sup_{x,y} \frac{\abs{A_{xy}}}{\avg{x}^\epsilon \avg{y}^\epsilon} \leq 
u}\,.
\end{equation*}
In order to find a bound on $\P(\Omega_u^c)$, we note that, by the uniform subexponential decay 
of the entries of $A$, we have
\begin{equation*}
\P\pbb{\frac{\abs{A_{xy}}}{\avg{x}^\epsilon \avg{y}^\epsilon} \geq u} \;\leq\; \beta \ee^{-u^\alpha \avg{x}^{\alpha 
\epsilon} \avg{y}^{\alpha \epsilon}}\,.
\end{equation*}
Therefore
\begin{equation} \label{probability of omega_u}
\P(\Omega_u^c) \;=\; \P\pbb{\sup_{x,y} \frac{\abs{A_{xy}}}{\avg{x}^\epsilon \avg{y}^\epsilon} > 
u} \;\leq\; \sum_{x,y} \P\pbb{\frac{\abs{A_{xy}}}{\avg{x}^\epsilon \avg{y}^\epsilon} \geq u} 
\;\leq\; C_\epsilon \, \ee^{-u^\alpha}\,.
\end{equation}

Moreover, from \eqref{bound on B} we get on $\Omega_u$
\begin{align*}
B &\;\leq\; 2^{2 \epsilon} u \sup_x \sum_y \sigma_{xy} \avg{x - y}^{1 + 3 \epsilon}
\\
&\;\leq\; 2^{2 \epsilon} u \sup_x \sum_y \frac{1}{\avg{x - y}^{d/2 + \epsilon}} \, \avg{x - 
y}^{1 + d/2 + 4 \epsilon} \sigma_{xy}
\\
&\;\leq\; C_\epsilon u \pbb{\sup_x \sum_y \avg{x - y}^{d + 2 + 8\epsilon} \sigma^2_{xy}}^{1/2}
\\
&\;\leq\; C_\epsilon u \, W^{d/2 + 1 + 4 \epsilon}\,,
\end{align*}
provided that $8 \epsilon \leq \eta$. Here we used \eqref{definition of sigma} and the 
assumption \eqref{decay of f}.

Summarizing: On $\Omega_u$ we have
\begin{equation*}
\absb{\partial_t \scalar{\psi_t}{\abs{x}^2 \psi_t}} \;\leq\; C u \, W^{d/2 + 1 + 4 \epsilon} 
\pbb{\zeta \scalarb{\psi_t}{\abs{x}^2 \psi_t} + \frac{1}{\zeta} \scalarb{\psi_t}{\avg{x}^{4 
\epsilon} \psi_t}}\,.
\end{equation*}
Choosing $\zeta^{-1} = u W^{d/2 + 1 + 4 \epsilon + d}$ yields
\begin{equation*}
\absb{\partial_t \scalar{\psi_t}{\abs{x}^2 \psi_t}} \;\leq\; C \pbb{\frac{1}{W^d} 
\scalarb{\psi_t}{\abs{x}^2 \psi_t} + u^2 W^{2 d + 2 + 8 \epsilon} \scalarb{\psi_t}{\avg{x}^{4 
\epsilon} \psi_t}}\,.
\end{equation*}
Let us take $\epsilon \leq 1/4$. Then we have, for any $\xi > 0$,
\begin{equation*}
\abs{x}^{4 \epsilon} \;\leq\; \xi^{4 \epsilon/ (4 \epsilon - 2)} + \xi \abs{x}^2\,.
\end{equation*}
Choosing $\xi^{-1} = u^2 W^{3d + 2 + 8 \epsilon}$ therefore yields
\begin{equation*}
\absb{\partial_t \scalar{\psi_t}{\abs{x}^2 \psi_t}} \;\leq\; C \pbb{\frac{1}{W^d} 
\scalarb{\psi_t}{\abs{x}^2 \psi_t} + u^4 W^{5d + 8}}\,.
\end{equation*}
Thus Gr\"onwall's lemma, together with $\scalar{\psi_0}{\abs{x}^2 \psi_0} = 0$, implies that on $\Omega_u$ we have
\begin{equation*}
\scalar{\psi_t}{\abs{x}^2 \psi_t} \;\leq\; C u^4 W^{5d + 8} t \ee^{C t / W^d}\,.
\end{equation*}
Therefore we have showed that, for all $t \leq W^d$, we have
\begin{equation} \label{estimate on variance of wave function}
\scalar{\psi_t}{\abs{x}^2 \psi_t} \;\leq\; C u^4 W^{6d + 8}
\end{equation}
on $\Omega_u$.

\subsection{Conclusion of the proof}
Let us abbreviate $\psi_t = \ee^{-\ii t H /2} \delta_0$ and $\wt \psi_t = \ee^{-\ii t 
\wt H / 2} \delta_0$. Then we have
\begin{equation*}
\partial_t \norm{\psi_t - \wt \psi_t}^2 \;=\;
\frac{\ii}{2} \pB{\scalarb{H \psi_t - \wt H \wt \psi_t}{\psi_t - \wt 
\psi_t} - \scalarb{\psi_t - \wt \psi_t}{H \psi_t - \wt H \wt \psi_t}}
\;=\; \im \scalar{\wt \psi_t}{(\wt H - H) \psi_t}\,.
\end{equation*}
Thus, using $\norm{\wt \psi_t} = 1$, we get
\begin{equation} \label{bound on derivative of error}
\absb{\partial_t \norm{\psi_t - \wt \psi_t}^2} \;\leq\; \norm{(H - \wt H) 
\psi_t}\,.
\end{equation}
Next, we observe that
\begin{align*}
\abs{H_{xy} - \wt H_{xy}} &\;=\; \qB{1 - \ind{\abs{x} \leq \wt N} \ind{\abs{y} 
\leq \wt N}} \abs{H_{xy}}
\\
&\;\leq\; \indb{\abs{y} \geq \wt N / 2} \abs{H_{xy}} + \indb{\abs{x - y} \geq \wt 
N / 2} \abs{H_{xy}}\,.
\end{align*}
This gives
\begin{align}
\norm{(H - \wt H) \psi_t}^2 &\;\leq\; \sum_{x,y,z} \absb{H_{xy} - \wt H_{xy}} 
\absb{H_{xz} - \wt H_{xz}}
\abs{\psi_t(y)} \abs{\psi_t(z)}
\notag \\
&\;\leq\; \sum_{x,y,z}
\indb{\abs{y} \geq \wt N / 2}
\indb{\abs{z} \geq \wt N / 2}
\absb{H_{xy}} \absb{H_{xz}}
\abs{\psi_t(y)} \abs{\psi_t(z)}
\notag \\
&\qquad
+ 2 \sum_{x,y,z}
\indb{\abs{y} \geq \wt N / 2}
\indb{\abs{x - z} \geq \wt N / 2}
\absb{H_{xy}} \absb{H_{xz}}
\abs{\psi_t(y)} \abs{\psi_t(z)}
\notag \\ \label{main splitting for finite speed of propagation}
&\qquad
+ \sum_{x,y,z}
\indb{\abs{x - y} \geq \wt N / 2}
\indb{\abs{x - z} \geq \wt N / 2}
\absb{H_{xy}} \absb{H_{xz}}
\abs{\psi_t(y)} \abs{\psi_t(z)}\,.
\end{align}
We estimate the second term of \eqref{main splitting for finite speed of propagation}; the two 
other terms are dealt with in exactly the same way.
On $\Omega_u$ the second term of \eqref{main splitting for finite speed of propagation} is 
bounded by
\begin{align*}
&\mspace{-40mu} 2 u^2 \sum_{x,y,z} \avg{x}^{2 \epsilon} \avg{y}^\epsilon \avg{z}^\epsilon 
\sigma_{xy} \indb{\abs{x - z} \geq \wt N / 2}
\sigma_{xz} \, \indb{\abs{y} \geq \wt N /2} \abs{\psi_t(y)}  \abs{\psi_t(z)}
\\
&\leq\; C u^2 \sum_{x,y,z} \avg{x - y}^{\epsilon} \sigma_{xy} \, \avg{x - z}^\epsilon 
\indb{\abs{x - z} \geq \wt N / 2} \sigma_{xz} \, \indb{\abs{y} \geq \wt N /2} 
\avg{y}^{2 \epsilon} \abs{\psi_t(y)}  \, \avg{z}^{2 \epsilon} \abs{\psi_t(z)}
\\
&\leq\; C u^2 \normb{\indb{\abs{x} \geq \wt N / 2} \avg{x}^{2 \epsilon} \psi_t}  
\normb{\avg{x}^{2 \epsilon} \psi_t}
\\
&\qquad \times
\pbb{\sup_y \sum_{x,z} \avg{x - y}^{\epsilon} \sigma_{xy} \, \avg{x - z}^\epsilon \indb{\abs{x 
- z} \geq \wt N / 2} \sigma_{xz}}^{1/2}
\\
&\qquad \times
\pbb{\sup_z \sum_{x,y} \avg{x - y}^{\epsilon} \sigma_{xy} \, \avg{x - z}^\epsilon \indb{\abs{x 
- z} \geq \wt N / 2} \sigma_{xz}}^{1/2}
\\
&\leq\; C u^2 \normb{\indb{\abs{x} \geq \wt N / 2} \avg{x}^{2 \epsilon} \psi_t}  
\normb{\avg{x}^{2 \epsilon} \psi_t}
\\
&\qquad \times
\pbb{\sup_x \sum_{y} \avg{x - y}^{\epsilon} \sigma_{xy}}
\, \pbb{\sup_x \sum_{y} \avg{x - y}^\epsilon \indb{\abs{x - y} \geq \wt N / 2} 
\sigma_{xy}}\,,
\end{align*}
where we used Schur's inequality \eqref{Holmgren-Schur}.
Next, we observe that \eqref{definition of sigma} and \eqref{decay of f} yield
\begin{equation*}
\sup_x \sum_{y} \avg{x - y}^\epsilon \sigma_{xy} \;\leq\; \sup_x \pbb{\sum_y \avg{x - y}^{-d - 
2 \epsilon}}^{1/2} \pbb{\sum_y \avg{x - y}^{d + 4 \epsilon} \sigma^2_{xy}}^{1/2} \;\leq\; 
C_\epsilon W^{d/2 + 2 \epsilon}\,,
\end{equation*}
as well as
\begin{align*}
&\mspace{-40mu} \sup_x \sum_{y} \avg{x - y}^\epsilon \indb{\abs{x - y} \geq \wt N / 2} 
\sigma_{xy}
\\
&\;\leq \pbb{\sum_{\abs{y} \geq \wt N /2} \frac{1}{\abs{y}^{d + 2 + 2 \epsilon}}}^{1/2}
\pbb{\sum_y \abs{y}^{d + 2 + 4 \epsilon} \sigma^2_{0 y}}^{1/2}
\\
&\leq\; \frac{C}{\wt N} \pbb{\sum_{\abs{y} \geq \wt N /2} \frac{1}{\abs{y}^{d + 2 
\epsilon}}}^{1/2}
W^{d/2 + 1 + 2 \epsilon}
\\
&\leq\; \frac{C_\epsilon}{\wt N} W^{d/2 + 1 + 2 \epsilon}\,.
\end{align*}
Moreover,
\begin{equation*}
\normb{\indb{\abs{x} \geq \wt N / 2} \avg{x}^{2 \epsilon} \psi_t} \;\leq\; \wt 
N^{2 \epsilon - 1} \norm{\avg{x} \psi_t}
\;\leq\; \wt N^{-1/2} \norm{\avg{x} \psi_t}
\,.
\end{equation*}

Estimating the first and third terms of \eqref{main splitting for finite speed of propagation} 
along the same lines, and putting everything together, yields
\begin{equation*}
\norm{(H - \wt H) \psi_t}^2 \;\leq\; \frac{C_\epsilon u^2}{\wt N} W^{d + 2} \, 
\norm{\avg{x} \psi_t}^2\,.
\end{equation*}
Using \eqref{estimate on variance of wave function} we therefore get
\begin{equation*}
\norm{(H - \wt H) \psi_t}^2 \;\leq\; \frac{C_\epsilon u^6}{\wt N} W^{7d + 10}\,.
\end{equation*}
Integrating \eqref{bound on derivative of error} we find the bound, valid on $\Omega_u$,
\begin{equation*}
\norm{\psi_t - \wt \psi_t} \;\leq\; C_\epsilon \pbb{\frac{u^6}{\wt N} W^{8d + 10}}^{1/2}\,,
\end{equation*}
uniformly for $t \leq W^d$. Setting $u = W$ and recalling \eqref{probability of omega_u} yields 
the claim.

\section{Proof of Proposition \ref{proposition: bound on H}} \label{section: proof of bound on 
H}

We start by partitioning $\Lambda_N = \bigcup_A \Lambda_{N,A}$ into cubes $\Lambda_{N,A}$ of 
side length $W$. In order to simplify notation, we assume that $N = 2 L W$ for some integer $L 
\in \N$.  The $(2L)^d$ cubes are indexed by $A \in \cal A_L \deq \{-L, \dots, L-1\}^d$. We set
\begin{equation*}
\Lambda_{N,A} \;\deq\; \hb{W A + \wt x \,:\, \wt x \in \{0, \dots, W - 1\}^d}
\end{equation*}
Let $P_A$ denote the projection $(P_A \psi)(x) \;\deq\; \ind{x \in \Lambda_{N,A}} \psi(x)$.

Next, decompose $\wh H$ into its cube components $\wh H_{AB} \deq P_A \wh H P_B$. Thus, $\wh 
H_{AB}$ is a $W^d \times W^d$ matrix. By Schur's inequality \eqref{Holmgren-Schur}, we have
\begin{equation*}
\norm{\wh H} \;\leq\; \sup_{A \in \cal A_L} \sum_{B \in \cal A_L} \norm{\wh 
H_{AB}}\,.
\end{equation*}
Let $g$ be a periodic function on $\cal A_L$ to be chosen later, and set
\begin{equation*}
\Omega_0 \;\deq\; \hb{\norm{\wh H_{AB}} \leq 3 M^{2 \delta} g(A - B) \text{ for all } A,B 
\in \cal A_L}\,.
\end{equation*}
Thus, on $\Omega_0$ we have
\begin{equation} \label{bound on H on good set}
\norm{\wh H} \;\leq\; 3 M^{2 \delta} \sum_{A} g(A)\,.
\end{equation}

In order to derive an estimate on the probability of $\Omega_0$, we use the Marcinkiewicz-Zygmund inequality: If $Z_1, 
\dots, Z_n$ are independent mean-zero complex random variables and $a_1, \dots, a_n \in \C$, then
\begin{equation}
\E \absbb{\sum_i a_i Z_i}^p \;\leq\; (Cp)^{p/2} \, \E \pbb{\sum_i \abs{a_i Z_i}^2}^{p/2}\,.
\end{equation}
(See e.g.\ \cite{strook}, Exercise 2.2.30, for a proof that gives the constant $(Cp)^{p/2}$.) Defining
$A^2 \deq \sum_i \abs{a_i}^2$, Jensen's inequality therefore yields for $p \geq 2$
\begin{equation} \label{consequence of MZ}
\E \absbb{\sum_i a_i Z_i}^p \;\leq\; (Cp)^{p/2} A^p \, \E \pbb{\sum_i \frac{\abs{a_i}^2}{A^2} \abs{Z_i}^2}^{p/2}
\;\leq\;
(Cp)^{p/2} A^p \, \sum_i \frac{\abs{a_i}^2}{A^2} \E \abs{Z_i}^p \;\leq\;
(C A^2 p)^{p/2} \max_i \, \E \abs{Z_i}^p\,.
\end{equation}

Next, we have, for $\wt x, \wt y \in \{0, \dots, W-1\}^d$,
\begin{equation*}
(\sigma_{AB})^2_{\wt x \wt y} \;\deq\; \sigma^2_{WA + \wt x, WB + \wt y} \;=\; 
\frac{1}{M} f\pbb{\frac{[AW + \wt x - BW - \wt y]_N}{W}} \;=\; \frac{1}{M} f 
\pb{[A - B]_{2L} + R_{\wt x \wt y}}\,,
\end{equation*}
where $\abs{R_{\wt x \wt y}} \leq 1$. Thus \eqref{definition of sigma} yields
\begin{equation} \label{bound on block variance}
(\sigma_{AB})^2_{\wt x \wt y} \;\leq\; \frac{1}{M} \wt f \pb{[A - B]_{2L}}\,.
\end{equation}
We may now estimate $\E \absb{\scalar{\psi_1}{\wh H_{AB} \psi_2}}^p$ for any $p \geq 2$ and $\psi_1, \psi_2 \in 
\C^{W^d}$ satisfying $\norm{\psi_1}, \norm{\psi_2} \leq 1$.  Let us first take $A \neq B$.  Then we get
\begin{equation} \label{rewriting of high moment}
\E \absb{\scalar{\psi_1}{\wh H_{AB} \psi_2}}^p \;=\; \E \absbb{\sum_{\wt x, 
\wt y} \underbrace{\frac{(\wh H_{AB})_{\wt x \wt 
y}}{(\sigma_{AB})_{\wt x \wt y}}}_{\eqd Z_{\wt x \wt y}} \, 
(\sigma_{AB})_{\wt x \wt y} \, \psi_1(\wt x) \psi_2(\wt y)}^p\,,
\end{equation}
where we restrict the summation to $\wt x, \wt y$ satisfying 
$(\sigma_{AB})_{\wt x \wt y} \neq 0$. Observing that
\begin{equation*}
(\wh H_{AB})_{\wt x \wt y} \;=\; \wh H_{WA + \wt x, WB + \wt y}\,,
\end{equation*}
we see that the random variables $(Z_{\wt x \wt y})_{\wt x \wt y \in \{0, \dots, W-1\}^d}$ are independent and satisfy 
$\abs{Z_{\wt x \wt y}} \leq M^\delta$. Therefore \eqref{consequence of MZ} and \eqref{rewriting of high moment} yield
\begin{equation*}
\E \absb{\scalar{\psi_1}{\wh H_{AB} \psi_2}}^p \;\leq\; \pbb{C p M^{2 \delta} \sum_{\wt x, \wt y} (\sigma_{AB})^2_{\wt x 
\wt y} \abs{\psi_1(\wt x)}^2 \abs{\psi_2(\wt y)}^2}^{p/2} \;\leq\; \pB{C p M^{-1 + 2 \delta} \wt f \pb{[A - 
B]_{2L}}}^{p/2}\,,
\end{equation*}
where in the last step we used \eqref{bound on block variance}. If $A = B$ then the random 
variables $Z_{\wt x \wt y}$ are no longer independent; this is easily remedied by 
splitting the summation over $\wt x, \wt y$ in \eqref{rewriting of high moment} 
into two parts: $\wt x \leq \wt y$ and $\wt x > \wt y$.  Using the 
estimate $\abs{a + b}^p \leq \abs{2a}^p + \abs{2b}^p$ we therefore get the bound
\begin{equation} \label{estimate on high moment of expectation of H}
\E \absb{\scalar{\psi_1}{\wh H_{AB} \, \psi_2}}^p \;\leq\; \pB{C p M^{-1 + 2 \delta} \wt f \pb{[A - B]_{2L}}}^{p/2}\,,
\end{equation}
valid for all $A,B$.

Next, we estimate, using \eqref{estimate on high moment of expectation of H},
\begin{equation*}
\P \pB{\absb{\scalar{\psi_1}{\wh H_{AB} \, \psi_2}} \geq M^{2 \delta} g(A - B)} \;\leq\; \frac{\E 
\absb{\scalar{\psi_1}{\wh H_{AB} \psi_2}}^p}{\pb{M^{2 \delta} g(A - B)}^p} \;\leq\; \pBB{\frac{C p \wt f \pb{[A - 
B]_{2L}}}{M^{1 + 2 \delta} g^2(A - B)}}^{p/2}\,.
\end{equation*}
Setting $p = \nu M$ for some fixed $\nu > 0$ and defining $g(A) \deq \sqrt{\wt f \pb{[A]_{2L}}}$ yields
\begin{equation} \label{final bound on expectation}
\P \pB{\absb{\scalar{\psi_1}{\wh H_{AB} \psi_2}} \geq M^{2\delta} g(A - B)} \;\leq\; \pbb{\frac{C \nu}{
M^{2 \delta}}}^{\nu M/2}\,.
\end{equation}
Note that this choice of $g$ implies
\begin{equation} \label{bound on sum of g}
\sum_A g(A) \;\leq\; \sum_{A \in \Z^d} \sqrt{\wt f(A)} \;\leq\; \pBB{\sum_{A \in \Z^d} \wt f(A) 
\avg{A}^{d + 1}}^{1/2} \pBB{\sum_{A \in \Z^d} \avg{A}^{-d - 1}}^{1/2} \;\leq\; C\,,
\end{equation}
by \eqref{decay of f}.

In order to estimate $\norm{\wh H_{AB}}$, we define the rectangular lattice
\begin{equation*}
I \;\deq\; \hbb{\psi \in \frac{1}{2 W^{d/2}} \Z^{W^d} \,:\, \norm{\psi} \leq 1}\,.
\end{equation*}
It is easy to see that $\abs{I} \leq (4 W^{d/2})^{W^d}$.  Now set
\begin{equation*}
\Omega_{AB} \;\deq\; \hbb{\sup_{\psi_1, \psi_2 \in I} \absb{\scalar{\psi_1}{\wh H_{AB} \, 
\psi_2}} \leq M^{2 \delta} g(A - B)}\,.
\end{equation*}
Therefore \eqref{final bound on expectation} yields
\begin{equation*}
\P(\Omega_{AB}^c) \;\leq\; \abs{I}^2 \pbb{\frac{C \nu}{M^{2 \delta}}}^{\nu M/2} \;\leq\; \pbb{\frac{C \nu M^{C / 
\nu}}{M^{2 \delta}}}^{\nu M/2}\,.
\end{equation*}

We now do an approximation argument using the lattice $I$. Let $\psi^*_1, \psi^*_2$ satisfy 
$\norm{\psi^*_1} , \norm{\psi^*_2}\leq 1$ and
\begin{equation*}
\norm{\wh H_{AB}} \;=\; \scalar{\psi^*_1}{\wh H_{AB} \, \psi^*_2}\,.
\end{equation*}
Now by definition of $I$, there are $\psi_1, \psi_2 \in I$ such that $\norm{\psi_1 - \psi^*_1}, 
\norm{\psi_2 - \psi^*_2} \leq 1/4$.  This gives
\begin{equation*}
\norm{\wh H_{AB}} \;=\; \scalarB{\psi_1^* - \psi_1 + \psi_1}{\wh H_{AB} (\psi_2^* - 
\psi_2 + \psi_2)} \;\leq\; \norm{\wh H_{AB}} \pbb{2 \frac{1}{4} + \frac{1}{4^2}} + 
\absb{\scalar{\psi_1}{\wh H_{AB} \, \psi_2}}\,.
\end{equation*}
Thus, on $\Omega_{AB}$ we have
\begin{equation*}
\norm{\wh H_{AB}} \;\leq\; \frac{16}{7} M^{2 \delta} g(A-B)\,.
\end{equation*}
We have therefore proved that $\Omega_0 \supset \bigcap_{A,B \in \cal A_L} \Omega_{AB}$, which 
yields the probability bound
\begin{equation*}
\P(\Omega_0^c) \;\leq\; \abs{\cal A_L}^2 \pbb{\frac{C \nu M^{C / \nu}}{M^{2 \delta}}}^{\nu M/2}
 \;\leq\; N^{2d} \pbb{\frac{C \nu M^{C / \nu}}{M^{2 \delta}}}^{\nu M/2}\,.
\end{equation*}
Choosing $\nu$ large enough yields
\begin{equation*}
\P(\Omega_0^c) \;\leq\; M^{-\epsilon M}\,,
\end{equation*}
for large enough $M$ and some fixed $\epsilon > 0$.

Moreover, \eqref{bound on H on good set} and \eqref{bound on sum of g} imply that on $\Omega_0$ 
we have
\begin{equation*}
\norm{\wh H} \;\leq\; C M^{2 \delta}\,.
\end{equation*}

\section{Proof of Lemma \ref{lemma: existence of refining pairing}} \label{appendix: higher 
lumpings}

We start with the following observation which allows us to rule out the simple case $n + n' \leq 8$. Assume that $n + n' 
\leq 8$ and that $\Gamma \in \scr G_{n,n'} \setminus \cal P_{n,n'}$. In order to prove \eqref{lower bound on number of 
skeleton edges}, we have to construct a refining pairing $\Pi$ of $\Gamma$ satisfying $m(\Pi) \geq 2$. It may be easily 
checked that this is always possible. Throughout this appendix we therefore assume that
\begin{equation} \label{more than eight edges}
n + n' \;>\; 8\,.
\end{equation}

Choose some ordering of the edges $\cal E(I_n \cup I_{n'})$. Then lumps are ordered by their 
smallest edge.

In a first step, we construct a special refining $\Gamma'$ of $\Gamma$ whose lumps are of size 
$2$ or $4$. Start by setting $\Gamma_0 \deq \Gamma$ and $j = 0$.
\begin{itemize}
\item
Denote by $\gamma$ the first lump in $\Gamma_j$ that satisfies $\abs{\gamma} \geq 6$; if there 
is no such lump, stop.
\item
Denote by $\gamma'$ the union of the first four edges of $\gamma$; define $\Gamma_{j + 1} \deq 
\Gamma_j \cup \{\gamma', \gamma \setminus \gamma'\} \setminus \gamma$. (That is, cut the lump 
$\gamma$ into two lumps of sizes $4$ and $\abs{\gamma} - 4$.)
\item
Set $j \mapsto j + 1$ and repeat this procedure.
\end{itemize}
After the algorithm has terminated, set $\Gamma' = \Gamma_j$.
We now claim that
\begin{equation} \label{bound on lump size}
p(\Gamma') \;\geq\; \frac{1}{2} \, p(\Gamma)\,.
\end{equation}
Indeed, let $n_i$ denote the number of lumps of size $i$ in $\Gamma$. Thus we have
\begin{equation*}
p(\Gamma) \;=\; 2 n_4 + 4 n_6 + 6 n_8 + 8 n_{10} + 10 n_{12} + \cdots\,.
\end{equation*}
From the definition of $\Gamma'$ we get
\begin{equation*}
p(\Gamma') \;=\; 2 n_4 + 2 n_6 + 4 n_8 + 4 n_{10} + 6 n_{12} + \cdots\,,
\end{equation*}
and \eqref{bound on lump size} follows.

In a second step, we construct a refining pairing $\Pi$ of $\Gamma'$ using a greedy algorithm 
that generates a finite sequence of lumpings $(\Gamma_j)$ that are successive refinements of 
each other. Additionally, along this construction some bridges will get a \emph{mark}. Bridges 
that received a mark at some stage retain it for all later stages. (To avoid confusion, we 
stress that this marking has nothing to do with the bridge tags; it is only used in this 
proof.) We shall construct the algorithm and the marking in such a way that, in the resulting 
pairing $\Pi$, no two marked bridges belong to the same (anti)ladder. Thus, the number of 
marked bridges will be a lower bound for $m(\Pi)$.
As usual we call lumps of size $2$ bridges. We call lumps of size $4$ \emph{four-lumps}. We say 
that two bridges are \emph{compatible} if they are neither parallel nor antiparallel; otherwise 
they are said to be \emph{incompatible}.  


The following notions will prove helpful. We say that two edges $e_1$ and $e_2$ are \emph{bridged} in $\Gamma_j$ if 
$\{e_1, e_2\} \in \Gamma_j$. For a four-lump of the form $\gamma = \{e_1, e_2, e_3, e_4\}$ we introduce the operation of 
\emph{bridging} $e_1$ with $e_2$ and $e_3$ with $e_4$; this means that we set $\Gamma_{j+1} \deq \Gamma_j \cup 
\hb{\{e_1, e_2\}, \{e_3, e_4\}} \setminus \gamma$, i.e.\ we split the four-lump into two bridges.

We now define the greedy algorithm and the marking. Start by setting $\Gamma_0 = \Gamma'$ and 
$j = 0$, and let all bridges of $\Gamma_0$ be unmarked.

Let $\gamma$ be the first four-lump of $\Gamma_j$ (recall that lumps have a fixed ordering). We 
define $\Gamma_{j+1}$ by refining $\gamma$ into two bridges, and marking one of the bridges of 
$\Gamma_{j + 1}$. We do this in such a way that
\begin{enumerate}
\item
the newly marked bridge is compatible with all other bridges of $\Gamma_{j + 1}$, and
\item
each newly created bridge is incompatible with at most one other bridge of $\Gamma_{j + 1}$.
\end{enumerate}

Now we show that such a refining process together with an appropriate marking is possible.  
First we deal with the case that there are two adjacent edges $e_1, e_2 \in \gamma$. By the 
nonbacktracking constraint in $Q_x(\b x)$, this is only possible if the common vertex of $e_1$ 
and $e_2$ is either $0$ or $n$. Denote by $e_3, e_4$ the two other edges of $\gamma$. It is easy to see that there is an 
$i = 1,2$ and an $i' = 3,4$ such that the bridge $\{e_i, e_{i'}\}$ is compatible with all bridges of $\Gamma_j$. We then 
define the lumping $\Gamma_{j+1}$ by bridging $e_i$ with $e_{i'}$ as well as the two remaining edges of $\gamma$ with 
each other. We mark the newly created bridge $\{e_i, e_{i'}\}$. That properties (i) and (ii) hold follows readily from 
the definition of (anti)parallel bridges.

\begin{figure}[ht!]
\begin{center}
\includegraphics{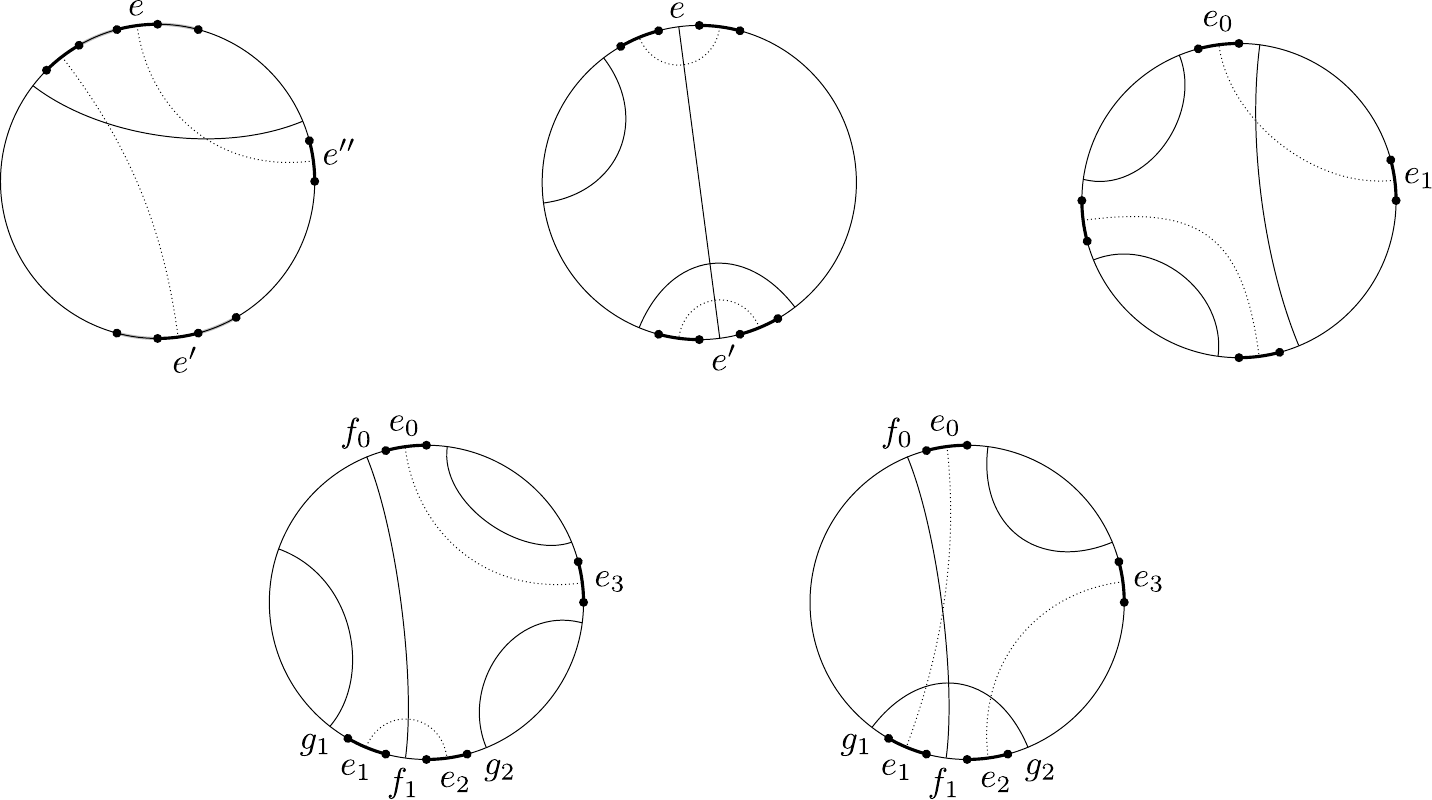}
\end{center}
\caption{The main step of the greedy algorithm. Top: (left to right) Case (a), Case (b), Case (c1). Bottom (left to 
right): Case (c2'), Case (c2'').  For each case we draw a typical scenario, in which edges of $\gamma$ are separated by 
a single edge only if this is required by the case in question. The edges in $\gamma$ are drawn using thick black 
lines.  Bridges already present in $\Gamma_j$ are drawn using solid lines, and bridges added by the current step using 
dotted lines. In Case (a), the edges in $\gamma'$ are drawn using thick grey lines.  \label{figure: refining}}
\end{figure}
Let us therefore assume from now on that no two edges of $\gamma$ are adjacent.
The lumping $\Gamma_{j+1}$ with marked bridges is defined according to the following four 
cases.  (See Figure \ref{figure: refining} for an illustration of each case.) In each case, 
both properties (i) and (ii) are easy to check. (Note that, under the additional assumption 
$\wh H_{xx} = 0$ for all $x$, it is easy to see that any two edges of $\gamma$ must be 
separated by at least two edges, so that only Case (c1) below needs to be considered.)

\begin{itemize}
\item[(a)] \emph{There are two edges $e,e' \in \gamma$ whose neighbouring edges all belong to 
another four-lump $\gamma' \in \Gamma_j$.} We choose an edge $e'' \in \gamma$ that has at least 
one neighbouring edge not in $\gamma'$ (it is easy to see that, since $\Gamma_j$ cannot consist of two interlacing 
four-lumps by \eqref{more than eight edges}, there always exists such an $e''$). We bridge $e$ with $e''$, as well as 
the two remaining edges of $\gamma$ with each other. We mark the newly created bridge $\{e,e''\}$.
\item[(b)]
\emph{There is a bridge $\{e,e'\} \in \Gamma_j$ such that every edge in $\gamma$ is adjacent to 
either $e$ or $e'$.} We bridge both edges adjacent to $e$ with each other, as well as both 
edges adjacent to $e'$ with each other. We mark the bridge $\{e,e'\}$.
\end{itemize}
\begin{itemize}
\item[(c)] \emph{Neither} (a) \emph{nor} (b) \emph{applies.} We choose $e_0 \in \gamma$ so that the set of four edges 
adjacent to $e_0$ and its two neighbours contains at most one other edge in $\gamma$. (By \eqref{more than eight edges} 
such an $e_0$ always exists.) Define
\begin{equation*}
\zeta \;\deq\;  \hb{e_1 \in \gamma \setminus \{e_0\} \,:\, \text{neither neighbour of $e_1$ is 
bridged in $\Gamma_j$ with a neighbour of $e_0$}}\,.
\end{equation*}
\begin{itemize}
\item[(c1)]
If $\zeta \neq \emptyset$, it is not hard to see that there is an $e_1 \in \zeta$ such that the bridge $\gamma \setminus 
\{e_0, e_1\}$ is incompatible with at most one bridge of $\Gamma_j$.  We bridge $e_0$ with $e_1$, and both remaining 
edges of $\gamma$ with each other. We mark the bridge $\{e_0,e_1\}$.
\item[(c2)]
If $\zeta = \emptyset$, there is a bridge $\{f_0, f_1\} \in \Gamma_j$ such that $f_0$ is adjacent to $e_0$, and $f_1$ is 
adjacent to two edges, $e_1$ and $e_2$; see Figure \ref{figure: refining}. We choose $e_2$ to be the edge ``antipodal'' 
to $e_0$ in the circular ordering of the four edges of $\gamma$, i.e.\ $e_2$ is the edge that cannot be reached from 
$e_0$ along the circle without crossing another edge of $\gamma$.  Clearly, one of the two selected edges has this 
property. Define $e_3 \deq \gamma \setminus \{e_0, e_1, e_2\}$.  Let $g_1 \neq f_1$ and $g_2 \neq f_1$ denote the two 
other neighbours of $e_1$ and $e_2$.

\begin{itemize}
\item[(c2')]
Assume first that $g_1$ and $g_2$ are not bridged in $\Gamma_j$. In this case we bridge $e_1$ with $e_2$ and $e_0$ with 
$e_3$; we mark the bridge $\{e_1, e_2\}$. It is immediate that $\{e_1, e_2\}$ is compatible with all bridges in 
$\Gamma_j$, and that $\{e_0, e_3\}$ is incompatible with precisely one bridge in $\Gamma_j$.
\item[(c2'')]
Assume now that $g_1$ and $g_2$ are bridged in $\Gamma_j$. Then we bridge $e_2$ with $e_3$ and $e_0$ with $e_1$. We mark 
the bridge $\{e_2, e_3\}$. Since Case (b) is excluded, we find that the bridge $\{e_2, e_3\}$ is compatible with all 
bridges of $\Gamma_j$. Moreover, the bridge $\{e_0, e_1\}$ is incompatible with precisely one bridge of $\Gamma_j$.
\end{itemize}
\end{itemize}
\end{itemize}

The pictures in Figure \ref{figure: refining} depict typical scenarios,
in which edges of $\gamma$ are separated by a single
edge (they are next-nearest neighbours) only if this
is explicitly required in the case being considered.
It is also possible that additional edges are
next-nearest neighbours; e.g.\ it may happen
that $f_0=g_1$ in the last picture. Checking
the few such explicit cases, one can see that
the algorithm
described above works for these cases as well, even
though the pictures are not accurate. It is this
step where the special choice of $e_0$ made in Case (c) is
necessary.

Set $j \mapsto j+1$. If $\Gamma_j$ is not yet a pairing, we repeat the procedure. Otherwise, we 
set $\Pi \deq \Gamma_j$ and stop the recursion; this is the completion of the algorithm. We 
need two crucial observations about the algorithm.

First, no bridge of $\Pi$ is marked twice. Indeed, in Cases (a) and (c), the bridge marked at 
step $j$ is new (i.e.\ does not exist in $\Gamma_j$); in Case (b) the bridge marked at step 
$j$, i.e.\ $\{e,e'\}$, was unmarked in $\Gamma_j$, as follows from the definition of Case (a).  
(The marking of $\{e, e'\}$ could only have been done in Case (a) if there $e$ had been bridged 
with $e'$, but this does not happen.) Therefore, the number of marked bridges of $\Pi$ is equal 
to the number of steps of the algorithm, i.e.\ the number of four-lumps in $\Gamma'$, which is 
$p(\Gamma')/2$.

Second, no two marked bridges of $\Pi$ belong to the same (anti)ladder. Indeed, by 
construction, the bridge marked at step $j$ of the algorithm is compatible with all bridges of 
$\Gamma_j$. Thus, if two marked bridges of $\Pi$, $\gamma$ and $\gamma'$, belong to the same 
(anti)ladder in $\Pi$, then there must exist a $j$ such that at step $j$ we added a bridge 
$\gamma''$ (marked or not) that was (anti)parallel to two bridges of $\Gamma_j$, one belonging 
to an (anti)ladder containing $\gamma$ and the other to an (anti)ladder containing $\gamma'$.  
By construction, however, this never happens; see (ii).

In conclusion: $\Pi$ has $p(\Gamma')/2$ marked bridges, such that no two of them lie in the 
same ladder or antiladder of $\Pi$. Therefore, for any choice of tags of the bridges of $\Pi$, 
the resulting skeleton will always contain at least $p(\Gamma')/2$ bridges. From \eqref{bound 
on lump size} we therefore get $m(\Pi) \geq p(\Gamma)/4$.

That $m(\Pi) \geq 2$ is easy to see from the fact that $m(\Pi) = 1$ would imply that $\Pi$ is 
either a complete ladder or a complete antiladder; this never happens by the property (i) of 
the greedy algorithm.

\section{Proof of Proposition \ref{proposition: hard decoupling}} \label{appendix: hard 
decoupling}

The key to the proof Proposition \ref{proposition: hard decoupling} is a decoupling of the bough tagging from the bough 
graph. The is done by adding an appropriate number of bough edges to $G \cup G'$, as in the proof of Lemma \ref{lemma: 
decoupling}.

\begin{lemma} \label{lemma: hard decupling in appendix}
There is an injective map $Y : \fra G_\sharp \to \fra G_\sharp$ such that for any $\cal G = (G, 
\tau_G)$ and $\cal {\wt G} = (\wt G, \tau_{\wt G}) = Y(\cal G)$ the following properties hold.
\begin{enumerate}
\item
The tagged stems of $\cal G$ and $\wt {\cal G}$ are identical.
\item
$\deg \pb{\cal B(G), \tau_G} =  2 \abs{\cal E(\cal B(\wt G))}$.
\item
For any $\cal G, \cal G' \in \fra G_\sharp$ we have the bound
\begin{multline} \label{estimate for hard decoupled graphs}
E_{\cal G \cup \cal G'} \;\leq\; \qBB{\prod_{e \in \cal E_B \text{\rm nonleaf}} \pb{M^{-1 + 2 \delta}}^{\ind{\tau(e) 
\neq (b,0)}}} \pb{C M^{- \delta}}^{L^{(b)}} \pb{C M^{-1 + \mu + 5 \delta}}^{L^{(f)}} \pb{C M^{-1 + \mu + 7 
\delta}}^{L^{(d)}/2}
\\ \times
\sum_{\wt \Gamma \in \scr G_{u,u'}} \sum_{\b x_S \,:\, \Gamma(\b x_S) = \wt \Gamma} Q(\b x_S) 
\prod_{\gamma \in \wt \Gamma} \E \prod_{e \in \gamma} \absB{P_{\tau(e)}\pb{\wh H_{x_{a(e)} 
x_{b(e)}}, \wh H_{x_{b(e)}
x_{a(e)}}}}\,,
\end{multline}
where all quantities on the right-hand side of \eqref{estimate for hard decoupled graphs} are defined in terms of $\wt 
{\cal G} \cup \wt{\cal G}'$, i.e.\ $L^{(i)} \equiv L^{(i)}(\wt {\cal G} \cup \wt{\cal G}')$ for $i = b,f,d$, and $\tau 
\equiv \tau_{\wt G \cup \wt G'}$.
\end{enumerate}
\end{lemma}

Using Lemma \ref{lemma: hard decupling in appendix} we find that Proposition \ref{proposition: 
hard decoupling} follows easily by repeating to the letter the argument at the beginning of 
Subsection \ref{section: decoupling of the easy boughs}.

\begin{proof}[Proof of Lemma \ref{lemma: hard decupling in appendix}]
For any graph $G$ we define the two following cases.
\begin{enumerate}
\item[(a)]
$\cal B(G)$ is either empty or contains at least one nondegenerate bough.
\item[(b)]
$\cal B(G)$ consists exclusively of degenerate boughs.
\end{enumerate}

Consider first the case that both $G$ and $G'$ satisfy (a).
Then we may proceed exactly as in the proof of Lemma \ref{lemma: decoupling}. Thus, we define
\begin{equation*}
D \;\deq\; \deg\pb{\cal B(G), \tau_G} - 2 \abs{\cal E(\cal B(G))}\,.
\end{equation*}
If $D = 0$ set $\wt {\cal G} = \cal G$. Otherwise $\cal B(G)$ contains a nondegenerate bough.  
Let $e$ be the nonleaf bough edge that is reached first on the walk around $G$ (see the proof 
of Proposition \ref{proposition: properties of graphs} for the definition of the walk around 
$G$).  Define $\wt {\cal G}$ as $\cal G$ in which we replaced the edge $e$ with a path of $D+1$ 
edges; here the first edge of the path carries the tag $\tau_G(e)$ and all other edges of the 
path the tag $(b,0)$.

Now set $Y(\cal G) \deq \wt {\cal G}$. By construction, we have that
\begin{equation*}
L^{(b)}(\wt G) \;=\; L^{(b)}(G) \,, \qquad L^{(f)}(\wt G) \;=\; L^{(f)}(G)\,, \qquad 
L^{(d)}(\wt G) \;=\; L^{(d)}(G)\,.
\end{equation*}
Moreover, $\cal G$ and $\cal {\wt G}$ have the same number of small nonleaf bough edges. It is 
also easy to see that Claims (i) and (ii) hold. Moreover, as in the proof of Lemma \ref{lemma: decoupling}, we find that the map $\cal G \mapsto \wt {\cal G}$ 
is injective. Defining $\wt {\cal G}'$ in the same way, we find that Claim (iii) follows from 
Proposition \ref{proposition: bound on hard E_G}.

Next, consider the case where $G$ satisfies (b) and $G'$ satisfies (a). The complication here 
is that we cannot add bough edges to $G$ without changing the numbers $L^{(b)}, L^{(f)}, L^{(d)}$. If $D = 0$ then we 
can set $\wt {\cal G} = \cal G$ and proceed as above. If $D > 0$ then there must be a (degenerate) bough edge $\wt e \in 
\cal E(\cal B(G))$ whose tag is $\tau_G(\wt e) = (b,i)$ for $i = 2,3,4$. We now use the additional small factor arising 
from such an edge. We claim that in this case we can improve the bound \eqref{bound on hard E_G} to
\begin{multline} \label{better estimate for small degenerate leaf}
E_{\cal G \cup \cal G'} \;\leq\; \qBB{\prod_{e \in \cal E_B \text{\rm nonleaf}} \pb{M^{-1 + 2 \delta}}^{\ind{\tau(e) 
\neq (b,0)}}} \pb{C M^{- \delta}}^{L^{(b)}} \pb{C M^{-1 + \mu + 5 \delta}}^{L^{(f)}} \pb{C M^{-1 + \mu + 7 
\delta}}^{(L^{(d)} - 1)/2} M^{-1 + 2 \delta}
\\ \times
\sum_{\wt \Gamma \in \scr G_{u,u'}} \sum_{\b x_S \,:\, \Gamma(\b x_S) = \wt \Gamma} Q(\b x_S) 
\prod_{\gamma \in \wt \Gamma} \E \prod_{e \in \gamma} \absB{P_{\tau(e)}\pb{\wh H_{x_{a(e)} 
x_{b(e)}}, \wh H_{x_{b(e)}
x_{a(e)}}}}\,.
\end{multline}
Note the additional factor $M^{-1 + 2\delta}$ at the expense of reducing the
exponent of $ M^{-1 + \mu + 7 \delta}$ by 1/2.
We outline the proof of \eqref{better estimate for small degenerate leaf}, which is almost identical to the proof of 
\eqref{bound on hard E_G}.
In choosing the ordering of edges $\preceq$, we require that $\wt e$ be the first degenerate bough edge. When tackling 
the edge $\wt e$ immediately after the recursive algorithm (used for nondegenerate boughs) of Proposition 
\ref{proposition: step for hard leaves} has terminated, we get a bound $\xi =  M^{-1 + 2 \delta} = M^{-1 + \mu + 5 
\delta} M^{-\mu - 3 \delta}$. Here the first term is the worst-case estimate using \eqref{definition of (p,6)}, and the 
second arises from the fact that, thanks to the assumption on $\tau(\wt e)$, the estimate \eqref{bound using tau tilde} 
is now in fact valid if we multiply the right-hand side by a factor $M^{-\mu - 3 \delta}$.  The remaining $L^{(d)} - 1$ 
degenerate edges are estimated exactly as in Section \ref{section: sum over deg}.  Thus we get \eqref{better estimate 
for small degenerate leaf}.

Now we may proceed as above. Let $e$ be the (degenerate) leaf that is reached first on the walk 
around $G$. Define $\wt {\cal G}$ as $\cal G$ in which we replaced the edge $e$ with a path of 
$D+1$ edges; here the first edge of the path carries the tag $\tau_G(e)$ and all other edges of the path carry the tag 
$(b,0)$. Denoting by $l \geq 1$ the number of leaves in $G$ belonging to the bough containing $e$, we 
have
\begin{equation*}
L^{(b)}(\wt G) \;=\;  l - 1\,,\qquad
L^{(f)}(\wt G) \;=\; 1\,,\qquad
L^{(d)}(\wt G) \;=\; L^{(d)}(G) - l\,.
\end{equation*}
These identities are simply an expression of the fact that the degenerate bough of $G$ that 
contains $e$ becomes a nondegenerate bough in $\wt G$ with one free leaf. Moreover, the mapping $\cal G \mapsto \wt 
{\cal G}$ clearly satisfies Claims (i) and (ii). That it is injective can be seen from the fact that $\cal G$ can be 
reconstructed from $\wt {\cal G}$, similarly to the construction given in the proof of Lemma \ref{lemma: decoupling}.

Choosing $\wt {\cal G}' = Y(\cal G')$ as above, we find that the bound \eqref{estimate for hard 
decoupled graphs} follows from \eqref{better estimate for small degenerate leaf} and the bound
\begin{equation*}
\pb{C M^{-1 + \mu + 7 \delta}}^{(L^{(d)} - 1)/2} M^{-1 + 2 \delta}
\;\leq\;
\pb{C M^{- \delta}}^{l - 1} C M^{-1 + \mu + 5 \delta} \pb{C M^{-1 + \mu + 7 \delta}}^{(L^{(d)} - l)/2}\,,
\end{equation*}
which is easy to check for all $l \geq 1$.

Finally, the case when both $G$ and $G'$ satisfy (b) is dealt with exactly as the previous 
case.
\end{proof}

\section{List of concepts and symbols} \label{section: notations}

\begin{tabbing}
\hspace{5cm} \= \\[-0.5cm]

\nitem{graph $G, G', \dots$}{A rooted, oriented, unlabelled tree graph; see p.\ \pageref{page: graphs}.}

\nitem{stem $\cal S(G)$ of a graph $G$}{The path of $G$ joining the vertices $a(G)$ and $b(G)$; see p.\ \pageref{page: 
stem}.}

\nitem{boughs $\cal B(G)$ of a graph $G$}{The subgraph of $G$ that does not contain the stem of $G$; a collection of 
disjoint trees; see p.\ \pageref{page: boughs}.}

\nitem{tagging $\tau$}{A map from the edges of a graph to the set of tags
encoding the contribution of an edge; see p.\ \pageref{page: tags} as well as 
Figures \ref{figure: edge decorations} and \ref{figure: decorated tree}.}

\nitem{nonbacktracking encoding $l$}{A map assigning to each pair of vertices $v,w$ a number $l(v,w) = 0,1$, used to 
encode any nonbacktracking conditions; see p.\ \pageref{page: l}.}

\nitem{decorated graph $\cal G, \cal G', \dots$}{A graph $G$ together with a tagging $\tau_G$ of $G$ and a map $l_G$ 
implementing any nonbacktracking conditions; see p.\ \pageref{page: decorated graph}.}

\nitem{label $x_v$ of a vertex $v$}{An element $x_v \in \Lambda_N$ assigned to $v$; see p.\ \pageref{page: label}.}

\nitem{lumping $\Gamma$}{An equivalence relation on the set of edges arising from taking the expectation value; see p.\ 
\ref{page: lumping}.}

\nitem{pairing $\Pi$}{The simplest type of lumping, whose equivalence classes each consist of two edges; see p.\ 
\pageref{page: pairing}.}

\nitem{bridge $\pi$}{A lump of a pairing; see p.\ \pageref{page: pairing}.}

\nitem{(anti)parallel bridges}{See p.\ \pageref{page: parallel}.}

\nitem{twisted and straight bridges}{See p.\ \pageref{page: twisted and straight bridge}.}

\nitem{tagged pairing $(\Pi,\theta)$}{A pairing $\Pi$ whose bridges $\pi$ carry a
 tag $\vartheta(\pi) \in \{\text{straight, 
twisted}\}$; see p.\ \pageref{page: tagged bridge}.}

\nitem{ladder $L_n$}{The simplest pairing, whose contribution is of leading order; see p.\ \pageref{page: ladder}.}

\nitem{skeleton pairing $S(\Pi, \vartheta)$}{Tagged pairing obtained from the
 tagged pairing $(\Pi, \vartheta)$ by 
collapsing parallel straight bridges and antiparallel twisted bridges; see p.\ \ref{page: skeleton}.}

\nitem{lonely leaf}{A leaf that is the only edge of its lump; see Definition \ref{definition of inverse of A}.}

\nitem{degenerate, bound,  free leaves}{See Definition \ref{definition: types of leaves}.}

\nitem{$\fra W$}{The set of graphs; see p.\ \pageref{page: graphs} and Figure \ref{figure: bare tree}.}

\nitem{$\fra G$}{The set of decorated graphs; see p.\ \pageref{page: decorated graph}.}

\nitem{$\fra G'$}{The set of decorated graphs corresponding to terms of the main path expansion \eqref{main path 
expansion}; see Definition \ref{definition of G prime}.}

\nitem{$\fra V_{xy}(\cal G)$}{The value of the decorated graph, a random variable; see \eqref{definition of value of 
graph}.}

\nitem{$\cal F_n$}{The operation that makes nonbacktracking the first backtracking stem vertex of a decorated graph; see 
p.\ \pageref{page: F_n}.}

\nitem{$\cal F_c$}{The operation that collapses the two stem neighbours of the first backtracking stem vertex of a 
decorated graph; see p.\ \pageref{page: F_n}.}

\nitem{$\scr B_{\cal G}$}{The set of decorated graphs obtained from $\cal G \in \fra G'$ by applying the operations 
$\cal F_n, \cal F_n$ until the stem is completely nonbacktracking; see Definition \ref{definition of B_G}.}

\nitem{$\fra G_\sharp$}{The union of all $\scr B_{\cal G}$ for $\cal G \in \fra G'$; see Definition \ref{definition of 
B_G}.}

\nitem{$\deg(\cal G)$}{The degree of the decorated graph $\cal G$, representing the degree of the polynomial $\fra 
V_{xy}(\cal G)$; see p.\ \pageref{page: degree}.}

\nitem{$\fra G_n$}{The set of decorated graphs in $\fra G_\sharp$ of degree $n$; see p.\ \pageref{D_n as a sum over 
graphs}.}

\nitem{$\scr G(G \cup G')$}{The set of lumpings of the edges of $G \cup G'$; see p.\ \pageref{expectation as a sum over 
lumpings}.}

\nitem{$I_n$}{The graph in $\fra W$ that consists of a bare stem with $n$ edges; see p.\ \pageref{page: I_n}.}

\nitem{$\cal I_n$}{The decorated graph in $\fra G_n$ obtained by assigning the tag $(s,0)$ to all edges of $I_n$; see 
p.\ \pageref{page: I_n}.}

\nitem{$\fra G^*_n$}{The set $\fra G_n \setminus \{\cal I_n\}$; see p.\ \pageref{page: G_n^*}.}

\nitem{$\scr G_{n,n'}$}{An abbreviation for $\scr G(I_n \cup I_{n'})$; see p.\ \pageref{page: I_n}.}

\nitem{$\scr P_{n,n'}$}{The subset of $\scr G_{n,n'}$ consisting of pairings; see p.\ \pageref{page: pairing}.}

\nitem{$V_x(\cal G \cup \cal G', \Gamma)$}{The value of the lumping $\Gamma \in \scr G(G \cup G')$ of the edges of the 
decorated graph $\cal G \cup \cal G'$; see \eqref{value of general lumping}.}

\nitem{$\preceq$}{A total order on the edges of $G \cup G'$
describing the order used for summing out bough vertices; see pp.\ \pageref{page: total order 1} and \pageref{page: 
total order 2}.}

\nitem{$A$}{A map that assigns to each bough edge $e$ an edge $A_e \succeq e$, and is used to parametrize the lumping of 
the bough edges; see p.\ \pageref{page: A}.}

\end{tabbing}

\end{document}